\PassOptionsToPackage{capitalise}{cleveref}
\documentclass[sigconf]{acmart}

\def\BibTeX{{\rm B\kern-.05em{\sc i\kern-.025em b}\kern-.08emT\kern-.1667em\lower.7ex\hbox{E}\kern-.125emX}}
\usepackage{nicefrac}
\usepackage{siunitx}
\usepackage{array,framed}
\usepackage{booktabs}
\usepackage{
  color,
  float,
  epsfig,
  wrapfig,
  graphics,
  graphicx,
  subcaption
}

\usepackage{textcomp}
\usepackage{setspace}
\floatname{algorithm}{Algorithm}
\usepackage{latexsym,fancyhdr,url}
\usepackage{enumerate}
\usepackage{graphics}
\usepackage{xparse}
\usepackage{xspace}
\usepackage{multirow}
\usepackage{csvsimple}
\usepackage{balance}

\usepackage{xcolor}
\usepackage{fix-cm}
\usepackage[table]{xcolor}

\usepackage{tabularx}
\usepackage{nicefrac}
\usepackage{siunitx}
\usepackage{array,framed}
\usepackage{booktabs}
\usepackage{
  color,
  float,
  epsfig,
  wrapfig,
  graphics,
  graphicx,
}
\usepackage{setspace}
\usepackage{latexsym,fancyhdr,url}
\usepackage{enumerate}

\usepackage{
  tikz,
  pgfplots,
  pgfplotstable
}
\usepackage{hyperref}

\usepackage{amsmath,amssymb,mathtools}

\usepackage{algorithm}
\usepackage[noend]{algpseudocode}
\algrenewcommand\algorithmiccomment[1]{\hfill{\scriptsize// #1}}
\usepackage{float}
\usepackage{dsfont}
\usepackage{hyperref}
\usepackage{hyperxmp}
\usepackage{natbib}
\usepackage{wrapfig}
\usepackage{multirow}

\copyrightyear{2025}
\acmYear{2025}
\setcopyright{cc}
\setcctype{by}
\acmConference[CCS '25]{Proceedings of the 2025 ACM SIGSAC Conference on Computer and Communications Security}{October 13--17, 2025}{Taipei, Taiwan}
\acmBooktitle{Proceedings of the 2025 ACM SIGSAC Conference on Computer and Communications Security (CCS '25), October 13--17, 2025, Taipei, Taiwan}\acmDOI{10.1145/3719027.3765151}
\acmISBN{979-8-4007-1525-9/2025/10}

\usepackage{booktabs}
\usepackage{enumerate}
\usepackage{soul}

\usepackage{bigints}
\usepackage{hyperref}

\usepackage[toc, page]{appendix}
\newcommand{\Adj}{\text{Adj}}

\newcommand{\Proba}{\mathbb P}
\newcommand{\RN}[1]{
  \textup{\uppercase\expandafter{\romannumeral#1}}
}
\usepackage{enumitem}
\usepackage{adjustbox,multirow}

\newcommand{\DP}{differential privacy }

\usepackage[thinc]{esdiff}
\usepackage{courier}
\usepackage{amsthm}
\usepackage{pifont}
\usepackage{graphicx}

\usepackage{caption}
\usepackage{subcaption}

\usepackage{pifont}
\usepackage{xcolor}

\newcommand{\xmark}{\textcolor{red}{\xmark}}

\usetikzlibrary{
  shapes.geometric,
  arrows,
  external,
  pgfplots.groupplots,
  matrix
}

\pgfplotsset{compat=1.9}

\DeclareMathAlphabet{\mathcal}{OMS}{cmsy}{m}{n}

\DeclareGraphicsExtensions{ 
    .png,.PNG, 
    .pdf,.PDF, 
    .jpg,.mps,.jpeg,.jbig2,.jb2,.JPG,.JPEG,.JBIG2,.JB2}

\usepackage{xparse}
\newcommand{\bnm}{\begin{newmath}}
\newcommand{\enm}{\end{newmath}}

\newcommand{\bea}{\begin{eqnarray*}}%
\newcommand{\eea}{\end{eqnarray*}}%

\newcommand{\bne}{\begin{newequation}}
\newcommand{\ene}{\end{newequation}}

\newcommand{\bal}{\begin{newalign}}
\newcommand{\eal}{\end{newalign}}

\newenvironment{newalign}{\begin{align}%
\setlength{\abovedisplayskip}{4pt}%
\setlength{\belowdisplayskip}{4pt}%
\setlength{\abovedisplayshortskip}{6pt}%
\setlength{\belowdisplayshortskip}{6pt} }{\end{align}}

\newenvironment{newmath}{\begin{displaymath}%
\setlength{\abovedisplayskip}{4pt}%
\setlength{\belowdisplayskip}{4pt}%
\setlength{\abovedisplayshortskip}{6pt}%
\setlength{\belowdisplayshortskip}{6pt} }{\end{displaymath}}

\newenvironment{newequation}{\begin{equation}%
\setlength{\abovedisplayskip}{4pt}%
\setlength{\belowdisplayskip}{4pt}%
\setlength{\abovedisplayshortskip}{6pt}%
\setlength{\belowdisplayshortskip}{6pt} }{\end{equation}}

\newcounter{ctr}

\newcounter{mytable}
\def\mytable{\begin{centering}\refstepcounter{mytable}}
\def\endmytable{\end{centering}}

\newcounter{myfig}
\def\myfig{\begin{centering}\refstepcounter{myfig}}
\def\endmyfig{\end{centering}}

\newlength{\saveparindent}
\newlength{\saveparskip}
\newcommand{\E}{{\rm I\kern-.3em E}}
\newcommand{\Prob}[1]{{\Pr\left[\,{#1}\,\right]}}

\renewcommand{\eqref}[1]{\mbox{Equation~(\ref{#1})}}

\def \part {part}

\def \blackslug{\hbox{\hskip 1pt \vrule width 4pt height 8pt
    depth 1.5pt \hskip 1pt}}
\def \qed{\quad\blackslug\lower 8.5pt\null\par}

\newcounter{mynote}[section]

\newcommand\ignore[1]{}

\newcounter{rcnote}[section]

\newcounter{mrnote}[section]

\newcounter{fknote}[section]

\newcounter{anote}[section]

\DeclareMathSymbol{\mlq}{\mathord}{operators}{``}
\DeclareMathSymbol{\mrq}{\mathord}{operators}{`'}

\newcommand{\rhf}[2]{R_{f, \gamma}}

\DeclareDocumentCommand{\edist}{o o}{
  \ensuremath{
    \IfNoValueTF{#1}{{d}}{{\sf d}(#1,#2)}
  }
}

\newcommand{\olrk}[1]{\ifx\nursymbol#1\else\!\!\mskip4.5mu plus 0.5mu\left(\mskip0.5mu plus0.5mu #1\mskip1.5mu plus0.5mu \right)\fi}

\NewDocumentCommand{\indseq}{ O{1} O{r} }{{#1}\ldots {#2}}

\usepackage{comment}

\newtheorem*{remark}{Remark}

\begin{document}
\fancyhead{}
\def\thetitle{\textsf{PLRV-O}: Advancing Differentially Private Deep Learning via Privacy Loss Random Variable Optimization}
\title{\thetitle}

\date{}

\author{Qin Yang}
\authornote{Both authors contributed equally to this work.}
\affiliation{
  \institution{University of Connecticut}
  \city{Storrs}
  \country{USA}
}

\author{Nicholas Stout}
\authornotemark[1]
\affiliation{
  \institution{Iowa State University}
  \city{Ames}
  \country{USA}
}

\author{Meisam Mohammady}
\affiliation{
  \institution{Iowa State University}
  \city{Ames}
  \country{USA}
}

\author{Han Wang}
\affiliation{
  \institution{The University of Kansas}
  \city{Lawrence}
  \country{USA}
}

\author{Ayesha Samreen}
\affiliation{
  \institution{Iowa State University}
  \city{Ames}
  \country{USA}
}

\author{Christopher J. Quinn}
\affiliation{
  \institution{Iowa State University}
  \city{Ames}
  \country{USA}
}

\author{Yan Yan}
\affiliation{
  \institution{University of Illinois at Chicago}
  \city{Chicago}
  \country{USA}
}

\author{Ashish Kundu}
\affiliation{
  \institution{Cisco Research}
  \city{San Jose}
  \country{USA}
}

\author{Yuan Hong}
\affiliation{
  \institution{University of Connecticut}
  \city{Storrs}
  \country{USA}
}

\begin{abstract}
Differentially Private Stochastic Gradient Descent (DP-SGD) is a standard method for enforcing privacy in deep learning, typically using the Gaussian mechanism to perturb gradient updates. However, conventional mechanisms such as Gaussian and Laplacian noise are parameterized only by variance or scale. This single degree of freedom ties the magnitude of noise directly to both privacy loss and utility degradation, preventing independent control of these two factors. The problem becomes more pronounced when the number of composition rounds $T$ and batch size $B$ vary across tasks, as these variations induce task-dependent shifts in the privacy–utility trade-off, where small changes in noise parameters can disproportionately affect model accuracy. To address this limitation, we introduce \textsf{PLRV-O}, a framework that defines a broad search space of parameterized DP-SGD noise distributions, where privacy loss \textit{moments} are tightly characterized yet can be optimized more independently with respect to utility loss. This formulation enables systematic adaptation of noise to task-specific requirements, including (i) model size, (ii) training duration, (iii) batch sampling strategies, and (iv) clipping thresholds under both training and fine-tuning settings. Empirical results demonstrate that \textsf{PLRV-O} substantially improves utility under strict privacy constraints. On CIFAR-10, a fine-tuned ViT achieves 94.03\% accuracy at $\epsilon \approx 0.5$, compared to 83.93\% with Gaussian noise. On SST-2, RoBERTa-large reaches 92.20\% accuracy at $\epsilon \approx 0.2$, versus 50.25\% with Gaussian.\footnote{Source code is available at \url{https://github.com/datasec-lab/plrvo}. This is the full version of the paper to appear in CCS'25.}
\end{abstract}



\maketitle

\section{Introduction}
\label{sec:intro}
Training and fine-tuning deep learning models pose significant privacy risks. During training, adversaries can exploit gradient updates, model outputs, and pre-trained parameters to reconstruct sensitive data, making privacy a critical concern~\cite{10.5555/3454287.3455610}. Moreover, fine-tuning large pre-trained language models, such as BERT~\cite{liu2019roberta} and GPT families~\cite{yu2021differentially}, is essential for achieving state-of-the-art performance in various tasks, including sentence classification~\cite{liu2019roberta}, text generation~\cite{novikova2017e2e}, and code generation~\cite{wang2018glue}. Data reconstruction attacks, such as Updates-Leak~\cite{247690} and Inverting Gradients~\cite{10.5555/3495724.3497145}, achieve success rates of up to 80\%, while the Dynamic Memory Model Inversion Attack (DMMIA)~\cite{10.1145/3581783.3612072} enhances realism, reaching 93.54\% on FaceScrub; in the same line of work, Carlini et al.~\cite{carlini2019} demonstrate that Large Language Models (LLMs) can memorize and regurgitate up to 1.6\% of training tokens verbatim. Additionally, prompt-based techniques further expose LLMs to data extraction attacks, increasing the risk of recovering individual samples.

\begin{table*}[h!]
\centering
\small
\caption{Comparison of Gaussian, Laplace, and \textsf{PLRV-O}. \textbf{Gaussian} cannot reach high accuracy in high-privacy regimes (strong DP). \textbf{Laplace} cannot support subsampling accounting for privacy amplification, large models, or $\ell_2$ clipping. Both mechanisms also lack the ability to optimize their randomization for the best privacy-utility tradeoff. \textbf{Large Moments} indicates support for higher-order R\'enyi DP moments, which improves utility in tight privacy accounting.}
\resizebox{\textwidth}{!}{
\begin{tabular}{|l|c|c|c|c|c|c|c|}
\hline
\textbf{Mechanisms} & \textbf{Optimizing Noise} & \textbf{High Acc@Strong DP} & \textbf{Large Moments} & \textbf{Subsampling} & \textbf{Small Model} & \textbf{Large Model} & $\boldsymbol{\ell_2}$ \textbf{Clip.} \\ \hline
Gaussian     & {\textcolor{red}{\ding{55}}} & {\textcolor{red}{\ding{55}}} & {\textcolor{red}{\ding{55}}} & {\textcolor{green}{\ding{51}}} & {\textcolor{green}{\ding{51}}} & {\textcolor{green}{\ding{51}}} & {\textcolor{green}{\ding{51}}} \\ \hline
Laplace      & {\textcolor{red}{\ding{55}}} & {\textcolor{red}{\ding{55}}} & {\textcolor{red}{\ding{55}}} & {\textcolor{red}{\ding{55}}}    & {\textcolor{green}{\ding{51}}} & {\textcolor{red}{\ding{55}}}   & {\textcolor{red}{\ding{55}}}   \\ \hline
\textsf{PLRV-O} (ours) & {\textcolor{green}{\ding{51}}} & {\textcolor{green}{\ding{51}}} & {\textcolor{green}{\ding{51}}} & {\textcolor{green}{\ding{51}}}  & {\textcolor{green}{\ding{51}}} & {\textcolor{green}{\ding{51}}} & {\textcolor{green}{\ding{51}}} \\ \hline
\end{tabular}
}
\label{Table:dp-dl}
\vspace{-0.1in}
\end{table*}
\normalsize

To mitigate privacy risks during model training and fine-tuning, Differentially Private Stochastic Gradient Descent (DP-SGD)~\cite{abadi2016deep, cherubin2024closed, du2023dp, 10.5555/3495724.3497145} is widely used. DP-SGD applies the Gaussian mechanism~\cite{dwork2014algorithmic}, which preserves privacy by clipping gradients and adding Gaussian noise, protecting the membership status of individual samples in the model parameters. Additionally, DP-SGD bounds cumulative privacy leakage over thousands of training or fine-tuning steps by tightly analyzing the Privacy Loss Random Variable (PLRV) of the Gaussian mechanism. However, DP-SGD often incurs significant accuracy loss, primarily due to the magnitude of the Gaussian noise, or ``noise multiplier'' ($\sigma$).\footnote{In DP-SGD, gradient clipping primarily impacts training stability, as overly aggressive thresholds can render models untrainable~\cite{abadi2016deep, balle2018improving, gopi2021numerical}.}

\vspace{0.05in}

\noindent\textbf{Motivation}. Existing DP-SGD solutions face significant limitations, and it is desirable to design a more sophisticated class of techniques.

First, DP-SGD using Gaussian noise offers limited flexibility in the achievable noise multiplier due to fundamental constraints. Gaussian noise faces intrinsic challenges in high-privacy regimes due to its $\sigma = \Theta(1/\epsilon)$ behavior~\cite{balle2018improving}.
Also, Mironov et al.~\cite{mironov2019r} and Sander et al.~\cite{mironov2019r, 10.5555/3618408.3619650} identify the ``privacy wall'' phenomenon, where Gaussian noise fails to efficiently leverage higher DP moment orders ($\lambda$), leading to overestimation in privacy accounting~\cite{balle2018improving}. This occurs because the PLRV of the Gaussian mechanism causes the Moments Accounting Function (MAF) to grow rapidly with $\lambda$ due to its light-tailed behavior. Interestingly, Balle et al.~\cite{balle2018improving} show that although $\sigma = \Theta(1/\epsilon)$ may help in low-privacy regimes (e.g., $\epsilon > 1$), the effect does not persist: the noise scale instead decreases proportionally to $1/\sqrt{\epsilon}$, leading to continued overcharging in privacy accounting even when looser privacy guarantees are acceptable.

Second, the Laplace mechanism~\cite{holohan2018bounded} offers a PLRV with an exponentially heavier tail, allowing it to more efficiently leverage higher DP moment orders ($\lambda$) to tighten the accounted privacy budget.  
    \footnote{Earlier work on query answering under DP shows that the Laplace mechanism can preserve accuracy better than the Gaussian mechanism in stricter privacy regimes, particularly in low-dimensional settings~\cite{geng2015optimal,mohammady2020r2dp}.} 
However, this advantage has not been fully realized in practice due to the challenges introduced by $\ell_1$-norm clipping in high-dimensional spaces (\emph{DP-SGD with the Laplace mechanism requires $\ell_1$-norm clipping by default}). Specifically, for an $n$-dimensional gradient vector, the $\ell_1$ norm can be up to $\sqrt{n}$ times larger than the $\ell_2$ norm. As a result, setting a clipping threshold $C$ based on $\ell_2$-norm bounds would translate to an effective threshold of approximately $C \times \sqrt{n}$ in $\ell_1$-norm, rendering the resulting DP guarantees unbearably loose. On the other hand, using a smaller threshold to clip directly in $\ell_1$-norm quickly collapses the gradient space, severely harming model trainability and often making training impossible.

Third, the privacy-utility trade-off is heavily influenced by task-specific factors such as the number of epochs ($E$) and batch size ($B$), which directly affect privacy accounting through their role in the moments accounting function and subsampling amplification. Whether a task involves training or fine-tuning also shapes the trade-off space. For instance, DP-SGD with a smaller clipping threshold ($C \approx 0.1$) is often sufficient for fine-tuning large models, whereas training typically requires a larger clipping threshold. Each case leads to a radically different privacy-utility trade-off. Effectively leveraging these variations within the noise Probability Density Function (PDF) demands a more versatile class of DP-SGD noise mechanisms, as opposed to traditional single-parameter approaches like Gaussian and Laplacian noise.

\vspace{0.05in}

\noindent\textbf{The \textsf{PLRV-O} Approach}. To our knowledge, this work is the first to introduce the Privacy Loss Random Variable Optimization (\textsf{PLRV-O}) framework for noise design in DP-SGD. Unlike conventional single-parameter mechanisms, \textsf{PLRV-O} defines a structured search space of \emph{randomized-scale Laplace} distributions, where privacy loss moments are tightly characterized and directly optimizable through the parameters of the scale distribution. This construction enables systematic exploration of diverse privacy–utility trade-offs under rigorous DP guarantees. Table~\ref{Table:dp-dl} provides a comparison of \textsf{PLRV-O} with existing DP noise mechanisms used in DP-SGD.

\vspace{0.05in}

\noindent\textbf{\textsf{PLRV-O} for Vision and Language Models}. \textsf{PLRV-O} is a generalized DP framework applicable to both model training and fine-tuning. In our evaluations, we instantiate it on vision tasks (training from scratch) and on large language models (fine-tuning). By replacing Gaussian noise in state-of-the-art methods with optimized noise distributions from \textsf{PLRV-O}, we enable more efficient use of the privacy budget and consistently improve model utility across diverse architectures, including ResNet, ViT, BERT, and RoBERTa~\cite{liu2019roberta, devlin2018bert, targ2016resnet}. Orthogonal techniques such as Ghost Clipping~\cite{li2021large} can also be combined with \textsf{PLRV-O}, though their integration lies beyond the scope of this work. Therefore, our primary contributions are summarized below. 
\begin{enumerate}[label=(\arabic*)]
\item We propose \textsf{PLRV-O}, the first DP-SGD framework with an optimizable non-Gaussian noise mechanism. It defines a search space of PDFs with randomized scale parameters governed by a Gamma distribution, where varying the shape and scale enables systematic exploration of privacy–utility trade-offs. Privacy guarantees are tightly accounted for, ensuring rigorous bounds while adapting to task-specific objectives.

\item We introduce a theoretical advancement by applying Schur-convexity and majorization theory~\cite{schur1923uber} to extend the Laplace family to the $\ell_2$ metric, avoiding $\ell_1$ clipping. This leads to a tight upper bound for the moments accountant, enabling more accurate privacy accounting and faster DP training.

\item We conduct comprehensive experiments across a diverse range of tasks in both computer vision (CV) and natural language processing (NLP), covering both training and fine-tuning, to validate our \textsf{PLRV-O}  mechanism. 
The results demonstrate \textsf{PLRV-O}'s effectiveness across various tasks, showing improved privacy-utility trade-offs over SOTA methods, especially in the stronger privacy regime ($0 < \epsilon < 1$). 
\end{enumerate}

\section{Preliminaries}
\label{sec:preliminaries}

We review some background on \DP and DP-SGD for the theoretical foundations of the \textsf{PLRV-O} framework. 

\subsection{Differential Privacy}\label{section: DP def}
\label{def: differential privacy original}

Differential privacy protects data privacy by introducing randomness into the output, 
either by injecting explicit noise or by leveraging inherent stochasticity in the mechanism or inputs
\footnote{Developing a unified framework to capture all randomness sources and translate them into differential privacy guarantees remains an open challenge~\cite{mironov2019r}. Here, we assume a fixed probability space \((\Omega, \mathcal{F}, P)\), where \(\Omega\) is the sample space, \(\mathcal{F}\) is a \(\sigma\)-algebra of events, and \(P\) is the probability measure.}. Let $\mathcal{D}$ be the space of datasets. For a query $q:\mathcal{D}\to\mathbb{R}^n$, 
we denote by $M_q: \mathcal{D} \times \Omega \to \mathbb{R}^n$ the randomized mechanism answering $q$. We say that two datasets $d,d'\in\mathcal{D}$ are \emph{adjacent}, written $\Adj(d,d')$, 
if they differ in the data of exactly one participant.

\begin{definition}[Differential Privacy]
A mechanism $M_q$ satisfies $\epsilon$-DP if, for all adjacent $d, d' \in \mathcal{D}$ and all measurable $S \subseteq \mathbb{R}^n$:
\[
\Pr[M_q(d) \in S] \leq e^{\epsilon} \Pr[M_q(d') \in S].
\]
\end{definition}

If the inequality fails, an \(\epsilon\)-breach occurs, indicating a non-negligible difference between the prior and posterior distributions. We next recall the Laplace Mechanism~\cite{Dwork06}, a fundamental tool for achieving $\epsilon$-differential privacy. Recall that the (univariate) Laplace distribution with mean zero and scale \(b\), denoted \(\mathrm{Lap}(b)\),
has density \(p(x; b) = \tfrac{1}{2b}\,\exp\bigl(-|x|/b\bigr)\) and variance \(2b^2\).
For \(w \in \mathbb{R}^n\) with i.i.d.\ components \(w_i \sim \mathrm{Lap}(b)\), 
the joint density is \(\bigl(\tfrac{1}{2b}\bigr)^n \exp(-\|w\|_1/b)\) and $\|\cdot\|_1$ denotes the  $\ell_1$ norm~\cite{Walker1965ProbabilityTA}.

\begin{definition}[Laplace Mechanism]
\label{def:LapMech}
Given a numerical query \(q(d)\) and a scale \(b\), the Laplace mechanism \(M_q(d,b)\) modifies \(q(d)\) by adding noise 
\(z \sim \mathrm{Lap}(b)\), i.e., \(M_q(d,b) = q(d) + z\).
\end{definition}

The scale parameter in the Laplace mechanism controls the level of privacy. Specifically:

\begin{theorem}
\label{thm:LapMech}
Let \(q : \mathcal{D} \to \mathbb{R}\) be a query with (global) $\ell_1$-sensitivity 
$ \Delta_1 q
  \;=\;
  \max_{d,d' : \mathrm{Adj}(d,d')} 
  \,\bigl\|\,q(d) - q(d')\bigr\|_{1}$.
If the Laplace mechanism \(M_q(d, b)\) uses a scale parameter 
\(b \ge \tfrac{\Delta q}{\epsilon}\), 
then \(M_q(d, b)\) is \(\epsilon\)-differentially private.
\end{theorem}

After the original definition of pure DP~\cite{Dwork06}, 
approximate DP was introduced to account for Gaussian noise, 
adding a $\delta$ typically on the order of $1/|\mathcal{D}|$, that represents the probability of violating 
the $\epsilon$-DP guarantee~\cite{DKM+06}.

\begin{definition}[Approximate Differential Privacy] %
\label{def:epsDP}
 A randomization mechanism \(M_q: \mathcal{D} \to \mathcal{R}\) is said to provide \((\epsilon, \delta)\)-differential privacy for releasing the query results \(q(\mathcal{D})\) if it randomizes its output such that, 
 for any two adjacent datasets $d, d' \subseteq \mathcal{D}$  and all subset \(S \subseteq \mathbb{R}^n\), the following holds:
\begin{align}
\label{def:DP}
\Prob{M_q(d) \in S} \leq e^{\epsilon} \Prob{M_q( d') \in S} + \delta.
\end{align}

\end{definition}
The Gaussian mechanism, described below, achieves approximate $(\epsilon,\delta)$-differential privacy~\cite{DKM+06}.
\begin{definition}[Gaussian Mechanism]
    The Gaussian mechanism $M_q(d, \sigma)$ modifies the answers to query $q$ in the dataset $d$ by adding noise $z \sim \mathcal{N}(0, \sigma^2 I)$, that is, $M_q(d, \sigma) = q(d) + z$. 
\end{definition}
 In particular, for 
\(\epsilon < 1\), its standard deviation is given by $\sigma \;=\; \frac{\Delta_2 q}{\epsilon} \,\sqrt{2\,\ln\bigl(\tfrac{1.25}{\delta}\bigr)}$, where \(\Delta_2 q\) is the \(\ell_2\)-sensitivity of the query \(q\) across all adjacent datasets.

Sequentially applying DP mechanisms increases the overall privacy cost, as shown by various composition theorems~\cite{dwork2010boosting,7883827,FengMWLQH24}. 
Na\"ive composition states that chaining \(k\) mechanisms, each \((\epsilon,\delta)\)-DP, results in an overall \((k\epsilon, k\delta)\)-DP guarantee, which can be overly conservative. Stronger composition theorems, such as the \emph{advanced composition} theorem~\cite{kairouz2015composition}, provide tighter bounds. In particular, the overall DP guarantee can become $\Bigl(\epsilon \sqrt{2k \,\ln\! \ \bigl(\tfrac{1}{\delta}\bigr)} 
  \;+\; k\,\epsilon\,\bigl(e^\epsilon - 1\bigr)\Bigr)$.

\subsection{Deep Learning with Differential Privacy}
\label{sec:DP-DPSGD}

DP-SGD~\cite{abadi2016deep} was the first method to incorporate DP into deep neural network (DNN) training. At each iteration \(t\), the gradient \(\mathbf{g}_t(\mathbf{x}_i)\) 
for a data point \(\mathbf{x}_i\) is clipped using a threshold \(C\). Formally, the clipped variant of gradients is given by $ \bar{\mathbf{g}}_t(\mathbf{x}_i) 
   \;=\; \mathbf{g}_t(\mathbf{x}_i)/\max\Bigl(1, \tfrac{\|\mathbf{g}_t(\mathbf{x}_i)\|_2}{C}\Bigr)$. 
This limits the sensitivity of each gradient, preparing it for the Gaussian mechanism, yielding ``sanitized'' gradients for all $1\leq i \leq n$
\begin{equation}
\label{eqn:gradientperturb}
    \mathbf{g}^s_t(\mathbf{x}_i)= \bar{\mathbf{g}}_t(\mathbf{x}_i) + z,
\end{equation}
where $z \sim \mathcal{N}(0, C^2 \sigma^2 \mathbf{I}_n)$.
The final update \(\mathbf{\bar{g}}^s_t\) is computed by averaging 
\(\mathbf{g}^s_t(\mathbf{x}_i)\) over the batch of size \(B\).
Even with tighter composition bounds, training over many iterations (e.g., thousands of rounds) can lead to high cumulative privacy loss (\(\epsilon\)). DP-SGD mitigates this by formulating an accounting function over the privacy loss terms across rounds, namely the \emph{Moments Accountant Function}. Consequently, DP-SGD with moments accounting achieves a significantly improved bound of $O(\epsilon \sqrt{T} B/|d|, \delta)$-DP.

\subsubsection{Moments Accounting Function}
\label{sec:MomentAccounting}
Consider two neighboring datasets \(d,d'\in \mathcal{D}\) and an outcome \(o\in\mathbb{R}^n\) 
from the minibatch average  \(\mathbf{\bar{g}}^s_t\) of sanitized gradients. 
The \emph{privacy loss} associated with \(o\) is defined by
\begin{align}
\label{eq:PrivLoss:c}
  c(o; \text{aux}, d, d') 
  \;=\; 
  \log\! \ 
    \frac{\Prob{{\mathbf{\bar{g}}^s_t}(\text{aux}, d)=o}}{\Prob{{\mathbf{\bar{g}}^s_t}(\text{aux}, d')=o}}
  .
\end{align}

\noindent where ``\(\text{aux}\)'' represents any state accumulated over sequential DP mechanisms. 
Since \(\mathbf{\bar{g}}^s_t(\cdot)\) is randomized, \( c(o; \text{aux}, d, d') \) is itself a random variable. 
To bound privacy loss more tightly, the moments of the random variable associated with \( c(o; \text{aux}, d, d') \) are tightly calculated and bounded. 
Precisely, for \(\lambda>0\), define
\begin{align}
\label{eq:lambda-moment-gen}
  \alpha_{\mathbf{\bar{g}}^s_t}(\lambda;\mathrm{aux},d,d') 
  \;=\;
  \log\,\mathbb{E}_{\,o \sim \mathbf{\bar{g}}^s_t(\mathrm{aux},d)}\!\bigl[\exp\!\bigl(\lambda\,c(o;\mathrm{aux},d,d')\bigr)\bigr].
\end{align}
We then take the worst-case value over \(\mathrm{aux}, d, d'\):
\begin{align}
\label{alpha func}
  \alpha_{\mathbf{\bar{g}}^s_t}(\lambda)
  \;=\;
  \max_{\mathrm{aux},\,d,\,d'} 
 \Bigl( \alpha_{\mathbf{\bar{g}}^s_t}(\lambda;\mathrm{aux},d,d') \Bigr).
\end{align}

\noindent
Key properties of this \textit{Moments Accounting Function (MAF)} include:
\begin{enumerate}
    \item \textbf{Composability.} 
    For a sequence of adaptive mechanisms \(\mathbf{\bar{g}}^s_1,\dots,\mathbf{\bar{g}}^s_t\),
     \begin{align}
    \label{composmaf}
      \alpha_{\mathbf{\bar{g}}^s_1,\dots,\mathbf{\bar{g}}^s_t}(\lambda) 
      \;\le\; 
      \sum_{i=1}^t \alpha_{\mathbf{\bar{g}}^s_i}(\lambda).
    \end{align}
    \item
    \textbf{Tail Bound.}
    For any \(\epsilon>0\), releasing \(\mathbf{\bar{g}}^s_1,\dots,\mathbf{\bar{g}}^s_t\) is \((\epsilon,\delta)\)-DP if
    \begin{align}
    \label{tail}
      \delta 
      \;=\;
      \min_{\lambda>0} 
     \Bigl( \exp\! \ \bigl( \alpha_{\mathbf{\bar{g}}^s_1,\dots,\mathbf{\bar{g}}^s_t}(\lambda) \;-\;\lambda\,\epsilon\bigr)\Bigr).
    \end{align}
\end{enumerate}

\subsubsection{Privacy Amplification with Subsampling}
While the Moments Accountant (MAF) significantly tightens the privacy budget, 
DP-SGD's main strength arises from \emph{privacy amplification via random sampling}, 
where the privacy cost per evaluation decreases \emph{quadratically}—rather than linearly—with 
the sampling rate \(\zeta = \tfrac{B}{|\mathcal{D}|}\). 
For instance, Bun et al.~\cite{BDRS18} show that 
\(\alpha(\lambda) \propto \zeta^2 \,\tfrac{6\,\lambda}{\sigma^2}\). 
Mironov et al.~\cite{mironov2019r} (Section 3.3) presented a tighter privacy amplification bound, which remains the state-of-the-art for the subsampled Gaussian mechanism. The bound is expressed as:
\begin{equation} \label{A_alpha}
\hspace{-.2cm}
\alpha_{\mathbf{\bar{g}}^s_1,\dots,\mathbf{\bar{g}}^s_t}(\lambda) \leq t 
\log \Biggl[ \sum_{\eta=0}^{\lambda+1} \binom{\lambda+1}{\eta} (1 - \zeta)^{\lambda+1 - \eta} \zeta^{\eta} 
\exp\left(\frac{\eta^2 - \eta}{2\sigma^2}\right) \Biggr].
\end{equation}
 The final \((\epsilon, \delta)\)-DP guarantee is derived using the tight conversion formula provided by Balle et al.~\cite{pmlr-v108-balle20a}, which holds for any mechanism, by substituting in the moment bounds $\alpha_{\mathbf{\bar{g}}^s_1,\dots,\mathbf{\bar{g}}^s_t}(\lambda)$ for $\alpha(\lambda)$:
\begin{equation} \label{eq:MAF:Gauss:Balle}
    \epsilon(\delta) = \min_{\lambda > 0} \Bigl(\frac{
    \alpha(\lambda)
    }{\lambda} + \log\left(\frac{\lambda}{\lambda + 1}\right) - \frac{\log(\delta) + \log(\lambda + 1)}{\lambda}\Bigr).
\end{equation}
\noindent
This formulation significantly enhances the analysis of privacy loss in scenarios involving repeated subsampling and Gaussian noise.

\section{Problem Statement}
\label{sec:problem_statement}

Let \( \Theta \) denote the noise in DP-SGD, and let \(X\) denote a set of fixed or tunable learning parameters under consideration.
Specifically, we may take \( X = (\mathrm{E}, \mathrm{B}, \mathrm{C}, n) \), where \(\mathrm{E}\) is the number of training epochs, \(\mathrm{B}\) is the batch size, \(\mathrm{C}\) is the clipping threshold, and \(n\) is the number of model parameters. Note that the pair \((\mathrm{E}, \mathrm{B})\) can equivalently be represented by the \textbf{total number of rounds} \(T\) and the \textbf{sample rate} \(\zeta = \tfrac{\mathrm{B}}{N}\), where \(N\) is the dataset size, since \(T = \lceil\tfrac{\mathrm{E} \cdot N}{\mathrm{B}}\rceil\). The learning rate \(lr\) and the initial state (random initialization or pre-trained weights for fine-tuning) are additional training parameters, but they are not part of our optimization problem.

Define \( \epsilon(X, \Theta, \delta) \) as the privacy budget achieved by \( \Theta \) under parameters \( X \) and target DP failure probability \( \delta \). We denote by \( \mathbf{U}(X, \Theta) \) the user's chosen utility function, which measures task-specific performance at the end of training/fine-tuning with DP-SGD noise parameters $\Theta$ and learning parameters $X$ for a given initial neural network, dataset, and loss function. %
Given fixed learning parameters \( X = (\mathrm{E, B, C, n}) \), our objective is to design the noise distribution (i.e., select PDF for optimal noise \( \Theta^{*} \)) that maximizes the utility subject to a constraint on privacy budget:
\begin{equation} \label{eq:noise-opt:1}
    \max_{\Theta}\  \mathbf{U}(X, \Theta^{*}) \quad \text{subject to} \quad \epsilon(X, \Theta, \delta^{*}) \leq \epsilon^{*}.
\end{equation}
\subsection{Optimizing Gaussian DP-SGD}
\label{sec:problem_statement}
In practice, the utility \( \mathbf{U} \) lacks an analytic form%
    \footnote{Even if the utility is a loss function that is differentiable with respect to the model parameters, the loss at the end of training/fine-tuning's dependence on DP-SGD noise parameters would not have a (known) analytic form.} 
and does not admit provable numerical optimization, offering only zeroth-order feedback (i.e., observable outputs without derivative information) after $E$ epochs of training/fine-tuning and then model evaluation with respect to $\mathbf{U}$. Likewise, the privacy budget $\epsilon(X, \Theta, \delta^{*})$ in general does not have a known expression, though analytic upper-bounds on privacy loss may be known. Consequently, because solving \eqref{eq:noise-opt:1} is intractable, a common alternative is to optimize a surrogate utility function, such as the noise multiplier, after carefully setting the clipping threshold~$\mathrm{C}$. 
If $\mathrm{C}$ is too small, nearly all gradients are clipped, leading to an insufficient signal-to-noise ratio (SNR) and preventing convergence. 
Conversely, if $\mathrm{C}$ is too large, the effective noise scale in Gaussian DP-SGD ($\mathrm{C}\sigma$) becomes excessive, which degrades model utility.

\vspace{4pt}

\noindent\textbf{Signal to Noise Ratio:} The interaction between clipping and noise variance directly impacts learning performance, and prior work has explored optimizing the clipping value $\mathrm{C}$. Importantly, $\mathrm{C}$ cannot be set arbitrarily small, as this prevents effective learning. While very small clipping values (e.g. 0.1) have been shown to work for fine-tuning large pre-trained models~\cite{li2021large}, pre-training tasks typically require much larger thresholds, in the range of 3–20~\cite{abadi2016deep}.
Thus, the surrogate optimization problem becomes \[\label{eq:noise-opt:2}
    \min_{\sigma > 0} \; \sigma 
    \quad \text{subject to} \quad 
    \epsilon(X, \Theta, \delta^{*}) \leq \epsilon^{*},
\]
where $X$ implicitly assumes the task-specific minimal clipping value $\mathrm{C}^*$.
Using the privacy loss upper bound \eqref{eq:MAF:Gauss:Balle} in place of $\epsilon(X, \Theta, \delta^{*})$, the optimization in \eqref{eq:noise-opt:2} reduces to a simple trade-off: larger $\sigma$ yields a smaller privacy budget $\epsilon$, and vice versa. Exploiting this monotonic relationship, we can then perform a simple numerical search to identify the smallest feasible $\sigma^*$ such that $\epsilon(\delta) < \epsilon^*$.

Going beyond the Gaussian case, we define a generic distortion measure to capture the effect of noise.  
Let $\mathbb{P}_{\Theta}(x)$ denote the PDF of a noise distribution parameterized by $\Theta$, where $x$ represents the noise sample value.  
The $\ell_1$ error, or distortion, of a noise distribution is defined as
\[
\mathsf{Distortion}(\Theta) = \int |x| \, \mathbb{P}_{\Theta}(x) \, dx.
\] For instance, this expression for Gaussian noise is $\sigma \sqrt{\frac{2}{\pi}}$. The corresponding optimization problem is
\begin{equation}
\label{eq:noise-opt:3}
\min_{\Theta} \; \mathsf{Distortion}(\Theta) 
\quad \text{subject to} \quad 
\epsilon(X, \Theta, \delta^{*}) \leq \epsilon^{*}.
\end{equation}

\vspace{-0.1in}
 
\subsection{Drawbacks of Gaussian and Laplace Noise}

Optimizing over \( \sigma \) alone has critical limitations with a single decision variable, which simultaneously controls both utility and privacy. This coupling is inherently limiting: with only one degree of freedom, optimizing for one objective (utility or privacy) directly impacts the other. 

First, the standard delta test, $\delta(\epsilon \rightarrow 0)$, provided by Balle et al. in~\cite{pmlr-v108-balle20a}, is an asymptotic test that demonstrates the failure probability $\delta$ of the Gaussian mechanism performs nearly 100 times worse than that of the Laplace mechanism in high-privacy regimes.  \begin{figure}[!h]
    \centering
\includegraphics[width=0.6\linewidth, trim=20pt 170pt 60pt 203pt, clip]{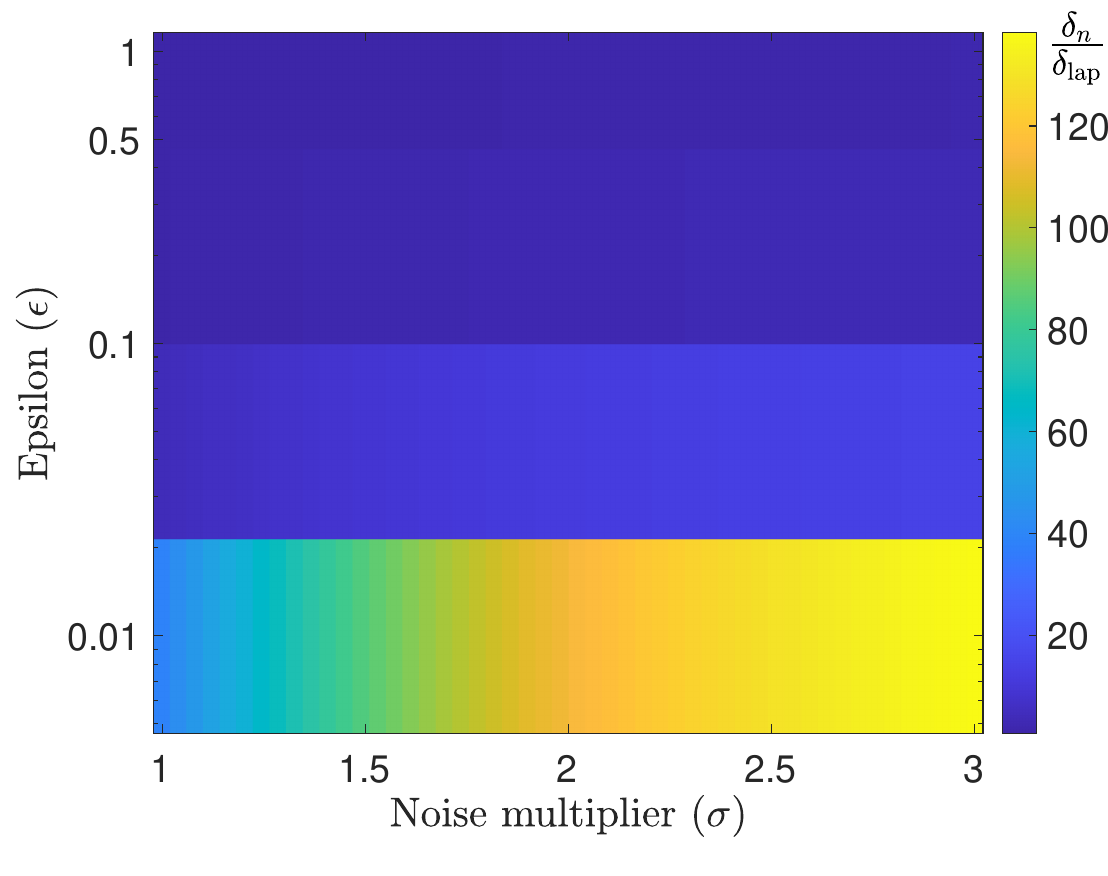}\vspace{-0.05in}
    \caption{Delta test ratio $\delta(\epsilon \rightarrow 0)$ of the Gaussian to Laplace mechanisms under varying noise levels. Laplace mechanism significantly outperforms Gaussian in $\epsilon \rightarrow 0$.}\vspace{0.1in}
    \label{fig:deltagauss}
\end{figure}
Figure~\ref{fig:deltagauss} shows the ratio between the delta tests for Gaussian and Laplace mechanisms, highlighting the increased cost of using Gaussian noise to achieve stronger protection. This inflated $\delta$ in high-privacy regimes directly increases the accounted budget $\epsilon(\delta)$ through the $-\log(\delta)$ term in the tight conversion formula~\eqref{eq:MAF:Gauss:Balle}. The limitation stems from the light tail of the Gaussian PDF, with its quadratic exponential decay. By contrast, the Laplace distribution, with heavier tails, is known to be optimal as $\epsilon \to 0$ under $\ell_1$ and $\ell_2$ mean error for scalar queries~\cite{geng2015optimal}.

However, Laplace introduces its own challenges for DP-SGD, particularly requiring \(\ell_1\)-norm clipping, which severely degrades performance in high-dimensional spaces. Recall the relationship $\|x\|_1 \leq \sqrt{n} \|x\|_2$, a direct consequence of the Cauchy–Schwarz inequality~\cite{steele2004cauchy}. 
The volumes of clipped spaces—the $\ell_1$ cross-polytope and $\ell_2$ ball—scale as $V_{\ell_1}(n, \mathrm{C}) = \frac{(2\mathrm{C})^n}{n!}$ and $V_{\ell_2}(n, \mathrm{C}) = \frac{\pi^{n/2} \mathrm{C}^n}{\Gamma(n/2 + 1)}$, respectively, with their ratio $\frac{V_{\ell_1}(n, \mathrm{C})}{V_{\ell_2}(n, \mathrm{C})} = \left(\frac{2}{\sqrt{\pi}}\right)^n \frac{(\frac{n}{2})!}{n!}$ shrinking exponentially in $n$. 
As shown in Figure~\ref{fig:l1vsl2vol}, $\ell_1$ clipping rapidly becomes restrictive, discarding far more gradients than $\ell_2$ clipping. This likely explains why Laplace mechanisms are rarely used in DP-SGD, as retaining sufficient learning capacity under $\ell_1$ clipping requires inflating the noise by a factor proportional to $\sqrt{n}$, i.e., the noise must grow with the model size.

\begin{figure}[!h] 
    \centering
    \includegraphics[width=0.6\linewidth, trim=45 190 70 215, clip]{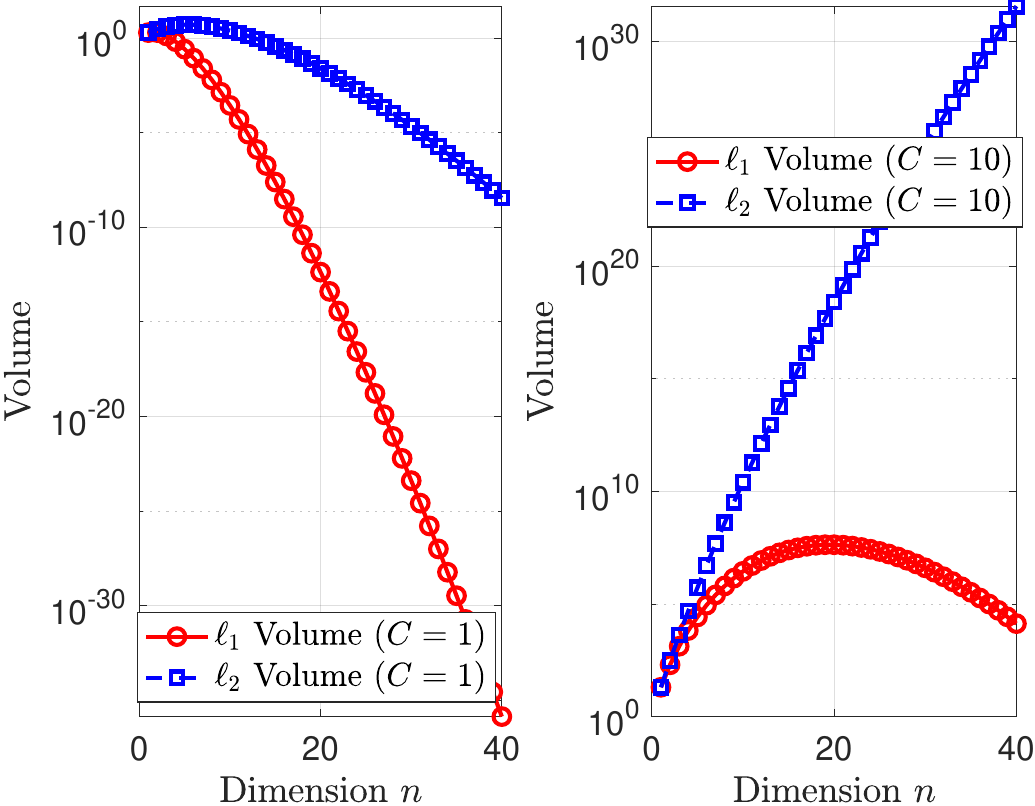}
    \caption{The volume of $n$-dimensional vector space clipped by $C$-$\ell_1$ and $C$-$\ell_2$ norms.}
     \label{fig:l1vsl2vol}
\end{figure}

Second, task-specific DP-SGD optimization requires a much larger decision space than Gaussian or Laplace mechanisms, which each rely on a single parameter (variance $\sigma^2$ or $2b^2$). Factors such as epochs ($E$), batch size ($B$), and whether the task is fine-tuning or pretraining directly shape privacy accounting and clipping needs, producing very different privacy–utility trade-offs. This nonlinearity calls for noise mechanisms more versatile than Gaussian or Laplace. An ideal mechanism would (i) combine their strengths and (ii) introduce free parameters in its standard deviation to flexibly balance 
\((\epsilon, \delta)\)-DP guarantees with task-specific accuracy.

\section{Methodology}
\subsection{$\mathbf{\ell_2}$ Clipping over Laplace's Privacy Loss}
\label{sec:PLRV-First}
Recall that the Laplace mechanism faces limitations with $\ell_1$ clipping, while Gaussian noise exhibits the ``privacy wall'' phenomenon. To address these issues and move toward a mechanism that combines the strengths of both, we first characterize the privacy loss of the Laplace mechanism (proof in Appendix~\ref{prooflaploss}) and then define the resulting challenges and an effective solution.

\begin{theorem}[Privacy Loss of a Laplace Mechanism]
\label{thm:lap2privacyloss}
The privacy loss of a Laplace mechanism with scale parameter $b$, applied to gradients clipped by $C$ under the $\ell_1$ norm, is bounded as:
\begin{equation}
\label{Laplacebound}
   c(o; \text{aux}, d, d', b) \leq \frac{C}{b}.
\end{equation}
\end{theorem}

Since this bound is tight for the Laplace mechanism~\cite{dwork2014algorithmic}, DP-SGD with Laplace noise inherently requires scaling the noise with the $\ell_1$ norm of the $n$-dimensional clipped gradient. Given any required $\ell_2$-clipped gradient space—for example, $C \geq 0.1$ in fine-tuning tasks—the ratio $\frac{\ell_1}{\ell_2} \leq \sqrt{n}$ causes the effective $\ell_1$ clipping threshold to grow rapidly with $n$, forcing $C$ to become very large. 

We will later propose an alternative approach that applies $\ell_2$ clipping to gradients and computes the moments accountant function (MAF) separately for each model parameter, composing the results across all parameters. 
This strategy will be motivated by the following behavior of the MAF for univariate Laplace mechanisms (e.g., using $\ell_1$ clipping), which allows us to leverage a much larger number of privacy loss moments $\lambda$ than with the Gaussian mechanism (on the order of thousands) when analyzing a single query. See Appendix~\ref{subsampledLap2} for the proof.

\begin{theorem}[Subsampled univariate Laplace Mechanisms]
\label{MAflaplace}
Let $M$ be a univariate Laplace mechanism with scale parameter $b$ and sampling probability 
$\zeta = \frac{L}{N}$,
where $L$ is the mini-batch size and 
$N$ 
is the dataset size. Suppose $M$ is applied to a single 
partial derivative query (i.e., with respect to a single model parameter)
$q(d)=\mathbf{g}(d)$ clipped by threshold $C$, i.e., $|\mathbf{g}| \leq C$.

Then, the moments accountant function of $M_{q}$ satisfies
\begin{equation}
\label{lap_2account1}
\alpha_{M_{q}}(\lambda) = \log \left[
\sum_{\eta = 0}^{\lambda + 1} \binom{\lambda + 1}{\eta} (1-\zeta)^{\lambda + 1 - \eta} \zeta^\eta F(C, \eta)
\right],
\end{equation}

where
\begin{equation}
\label{eqn:G}
F(C, \eta) = \frac{\eta e^{\frac{(\eta-1)C}{b}} + (\eta-1) e^{-\frac{\eta C}{b}}}{2\eta-1}.
\end{equation}
\end{theorem}

While the univariate analysis in~\eqref{lap_2account1} provides an exact expression for each parameter, directly summing over millions of parameters can still lead to significant overestimation of the total privacy loss unless a tight holistic (multivariate) bound is applied.%
\footnote{Here, the summation is over model parameters rather than training/fine-tuning iterations, similar to how composability traditionally applies over iterations.  
For the overall training/fine-tuning process, we would also need to compose over iterations.}

To address this overestimation, we apply \emph{Majorization Theory}~\cite{marshall1979inequalities,b81119c4-48ac-33d6-b3b8-e1de6925b554}, a powerful tool for comparing vectors based on their spread or concentration. After $\ell_2$ clipping with threshold $C$, the marginal clipped gradients $|\mathbf{g}_i|$ are mostly small but unevenly distributed. To capture this structure, we construct a \textit{majorization set}—a specially ordered vector that dominates the original gradient vector.

This majorization set enables us to bound any Schur-convex function of the original gradients by its value on the majorization set. Since the MAF depends on the gradients through a Schur-convex structure, applying majorization allows us to obtain a tight holistic (multivariate) privacy bound for $\ell_2$-clipped gradients.

In what follows, we first introduce background on majorization theory, then show that the moments accountant function given in ~\eqref{lap_2account1} is Schur-convex. Finally, we derive a tight majorization set and use it to bound the total moment accountant function for Laplace mechanisms.

\subsubsection{Majorization Theory Background}\label{Majorization} For two vectors \( \mathbf{x}, \mathbf{y} \in \mathbb{R}^n \), we say that \( \mathbf{x} \) \emph{weakly majorizes} \( \mathbf{y} \) from below, denoted \( \mathbf{x} \succ_w \mathbf{y} \), if the following condition holds:
\[
\sum_{i=1}^{k} x_i^{\downarrow} \geq \sum_{i=1}^{k} y_i^{\downarrow} \quad \text{for all } k = 1, \dots, n,
\]
where \( x_i^{\downarrow} \) and \( y_i^{\downarrow} \) denote the \( i \)-th largest components of \( \mathbf{x} \) and \( \mathbf{y} \), respectively~\cite{marshall1979inequalities,b81119c4-48ac-33d6-b3b8-e1de6925b554}. That is, both vectors are assumed to be sorted in non-increasing order before comparison.

A function \( F: \mathbb{R}^n \to \mathbb{R} \) is called \emph{Schur-convex} if \( x \prec y \Rightarrow F(x) \leq F(y) \). Intuitively, Schur-convex functions favor vectors that are more ``spread out,'' making them useful for bounding symmetric functionals.

\paragraph{Schur--Ostrowski Criterion~\cite{peajcariaac1992convex}.}
Let \( f: \mathbb{R}^d \to \mathbb{R} \) be a symmetric function with continuously differentiable partial derivatives. Then \( f \) is Schur-convex if and only if for all \( \mathbf{x} \in \mathbb{R}^d \) and for all \( 1 \leq i, j \leq d \), the following inequality holds: 
\[
(x_i - x_j)\left( \frac{\partial f}{\partial x_i} - \frac{\partial f}{\partial x_j} \right) \geq 0 
.  
\]
The following result is proven in Appendix~\ref{mafschur}.
\begin{theorem}[Schur-convexity of MAF]
\label{schurmaf}
Let $\bar{\mathbf{G}} = \{ \mathbf{\bar{g}}_1, \dots, \mathbf{\bar{g}}_n \}$ be the set of differentially private gradients obtained by applying $\ell_2$ clipping followed by a multivariate Laplace mechanism. Let $\alpha_{\mathbf{\bar{g}}_i}(\lambda)$ denote the moment accountant function for coordinate $i$.

Then, the total moment accountant function
\begin{equation}
\label{mafpara}
\alpha_{\bar{\mathbf{G}}}(\lambda) = \sum_{i=1}^n \alpha_{\mathbf{\bar{g}}_i}(\lambda)
\end{equation}
is \textbf{Schur-convex} with respect to the marginal clipped gradients $|\mathbf{g}_i|$.
\end{theorem}

With $\alpha(\lambda)$ being Schur-convex, we now introduce a majorization set over $\mathbf{G}$, the $\ell_2$ clipped (but not noisy) gradient vector. The following result is proven in Appendix~\ref{majorseet}.

\begin{lemma}[Majorization Set Construction]
\label{lem:majorset}
Let ${\mathbf{G}} = \{ \mathbf{g}_1, \dots, \mathbf{g}_n \}$ be the marginal clipped gradients sorted in descending order. Then, $|\mathbf{G}|$ is weakly majorized by the vector $x = \{x_1, \dots, x_n\}$ defined by
\[
x_i = C \left( \sqrt{i} - \sqrt{i-1} \right), \quad i = 1, \dots, n,
\]
where $C$ is the $\ell_2$ clipping threshold. That is,
\[
|\mathbf{G}| \prec_w x,
\]
where $\prec_w$ denotes the weak majorization order.
\end{lemma}

We now extend the univariate moments accountant to the multivariate setting by summing across coordinates, where each coordinate $i$ uses its corresponding majorization bound $x_i$. The following theorem naturally follows from Theorems~\ref{MAflaplace}, \ref{schurmaf}, and Lemma~\ref{lem:majorset}.

\begin{theorem}[Subsampled multivariate Laplace Mechanism]
\label{MAflaplace-multi}
Let $\bar{\mathbf{G}} = \{ \mathbf{\bar{g}}_1, \dots, \mathbf{\bar{g}}_n \}$ be the set of differentially private gradients obtained by applying $C$-$\ell_2$ clipping followed by a multivariate Laplace mechanism with scale $b$ and sampling probability $\zeta$. Define the majorization set $x = \{x_1, \dots, x_n\}$ by
\[
x_i = C \left( \sqrt{i} - \sqrt{i-1} \right), \quad i = 1, \dots, n.
\]
Then, the total moment accountant satisfies
\begin{equation}
\label{mafpara}
\alpha_{\bar{\mathbf{G}}}(\lambda) \leq \sum_{i=1}^n \log \left[
 \sum_{\eta=0}^{\lambda+1} \binom{\lambda+1}{\eta} (1-\zeta)^{\lambda+1-\eta} \zeta^\eta F(x_i, \eta)
\right],
\end{equation}
where
\begin{equation}
\label{eqn:G-multi}
F(x_i, \eta) = \frac{\eta e^{\frac{(\eta-1)x_i}{b}} + (\eta-1) e^{-\frac{\eta x_i}{b}}}{2\eta-1}.
\end{equation}
\end{theorem}

\subsection{DP-SGD Beyond Gaussian and Laplace}
To address the second objective, we propose a new class of noise, namely \emph{PLRV noise}, a randomized-scale extension of the standard Laplace mechanism. The scale parameter \( b \) is drawn from a non-negative probability density function \( f: \mathbb{R}_{\geq 0} \to \mathbb{R}^+ \), defined over the probability space \( (\Omega, \mathcal{F}, P) \). a PLRV noise is particularly notable for its ability to effectively \textbf{separate} accounting from distortion—using its scale PDF parameters. In the next section, we construct a search space of candidate PLRV distributions that can be efficiently optimized to guarantee very small privacy budgets.

\begin{definition}[PLRV Noise]
PLRV noise \( z \in \mathbb{R}^n \) is drawn from a mixture of zero-mean multivariate Laplace distributions, where the scale \( b \) is randomized according to a non-negative density \( f(b) \). Formally, its density is given by
\begin{equation}
\label{eqn:plrvpdf}
    p(z; f) = \int_{\mathbb{R}_{\geq 0}} f(b) \left(\frac{1}{2b}\right)^n \exp\left(-\frac{\|z\|_1}{b}\right) \, db,
\end{equation}

with \( \|z\|_1 = \sum_{i=1}^n |z_i| \) and \( f \) referred to as the PLRV seed.
\end{definition}

\begin{definition}[PLRV Mechanism] \label{def:PLRVmechanism}
Let \( q(d) \in \mathbb{R}^n \) be a numerical query and \( z \sim p(z; f) \) be PLRV noise as defined in ~\eqref{eqn:plrvpdf}. The PLRV mechanism is given by \( M_q(d, f) = q(d) + z \).
\end{definition}

Technically, DP-SGD with a PLRV mechanism involves a two-step sampling process: first, sampling a scale parameter, and second, using the sampled scale to generate noise from a zero-mean multivariate Laplace distribution. 

Derived from the Laplace mechanism, the following exact moments accounting result is obtained by applying the law of total expectation to Theorem~\ref{MAflaplace}. The proof is in Appendix~\ref{proofplrvuni}.

\begin{theorem}[Subsampled univariate PLRV Mechanism]
\label{the:subsampledMAF2232}
Let \( M_{\mathbf{g}} \) be a univariate PLRV mechanism with scale \( b \sim f(b) \), where \( f \) is a non-negative density. It is applied to a single partial derivative \( \mathbf{g} \in [-C, C] \), clipped at threshold \( C \). Assume \( \mathbf{g} \) is computed from a mini-batch sampled with probability \( \zeta = \frac{L}{N} \), where \( L \) is the batch size and \( N \) is the dataset size. Define \( u = 1/b \) and let \( h(u) = \frac{1}{u^2} f\left(\frac{1}{u}\right) \) denote the density of the reciprocal scale. Then, the moments accounting function satisfies:
\[
\alpha_{M_{q}}(\lambda) \leq \log \left\{
\sum_{\eta = 0}^{\lambda + 1} \binom{\lambda + 1}{\eta} (1 - \zeta)^{\lambda + 1 - \eta} \zeta^\eta \mathcal{G}(C, \eta)
\right\},
\]
where
\[
\mathcal{G}(C, \eta) =
\frac{\eta}{2\eta - 1} \mathcal{M}_u\big((\eta - 1)C\big) + \frac{\eta - 1}{2\eta - 1} \mathcal{M}_u\big(-\eta C\big).
\]
\end{theorem}

\noindent 
$\mathcal{M}_u$ denotes the moment generating function (MGF) for a random variable $u$. (see Appendix~\ref{apndx:Background:prob} for details). %
Using the majorization trick, the following bound is obtained.\footnote{The Schur-convexity of the MAF for a PLRV mechanism (Theorem~\ref{the:subsampledMAF2232}) follows directly from our earlier Laplace proof, as \(\mathcal{G}(C,\eta)\) can be viewed as the expectation of \(F(C,\eta)\).}

\begin{theorem}[Multivariate PLRV Mechanism]
\label{the:plrvmulti}
Let $\bar{\mathbf{G}} = \{ \mathbf{\bar{g}}_1, \dots, \mathbf{\bar{g}}_n \}$ be the set of differentially private gradients obtained by applying $C$-$\ell_2$ clipping followed by a multivariate PLRV mechanism with scale \( b \sim f(b) \), where \( f \) is a non-negative density. Define \( u = 1/b \) and let \( h(u) = \frac{1}{u^2} f\left(\frac{1}{u}\right) \) be the density of the reciprocal scale. Then, the moments accounting function under subsampling probability \( \zeta \) satisfies:

\begin{equation}
\label{mafpara}
\alpha_{\bar{\mathbf{G}}}(\lambda) \leq \sum_{i=1}^n \log \left[
 \sum_{\eta=0}^{\lambda+1} \binom{\lambda+1}{\eta} (1-\zeta)^{\lambda+1-\eta} \zeta^\eta \mathcal{G}(x_i, \eta)
\right],
\end{equation}
where
\begin{equation}
\label{eqn:G-multiplrv}
\mathcal{G}(x_i, \eta) =
\frac{\eta}{2\eta - 1} \mathcal{M}_u\big((\eta - 1)x_i\big) + \frac{\eta - 1}{2\eta - 1} \mathcal{M}_u\big(-\eta x_i\big).
\end{equation}
\end{theorem}

The PLRV mechanism models seed distributions as linear combinations of non-negative random variables, enabling composable MGFs and a clean MAF bound via weighted log-MGF sums.


\begin{figure}[ht]
    \centering
    \includegraphics[width=1\linewidth, trim=0pt 280pt 20pt 300pt, clip]{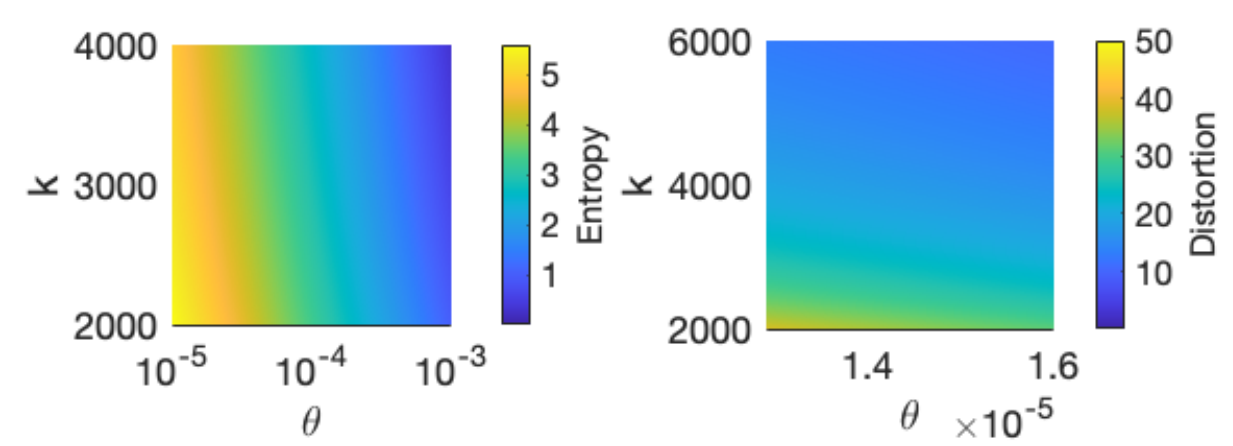}
    \caption{Visualization of $\theta$ controlling entropy (privacy, learnability) and $k$ controlling distortion (accuracy).}
    \label{fig:segreagation}
\end{figure}

\begin{figure*}[!t]
    \centering
    \vspace{-0.6in}
\includegraphics[width=0.9\linewidth, trim=20pt 50pt 50pt 51pt, clip]{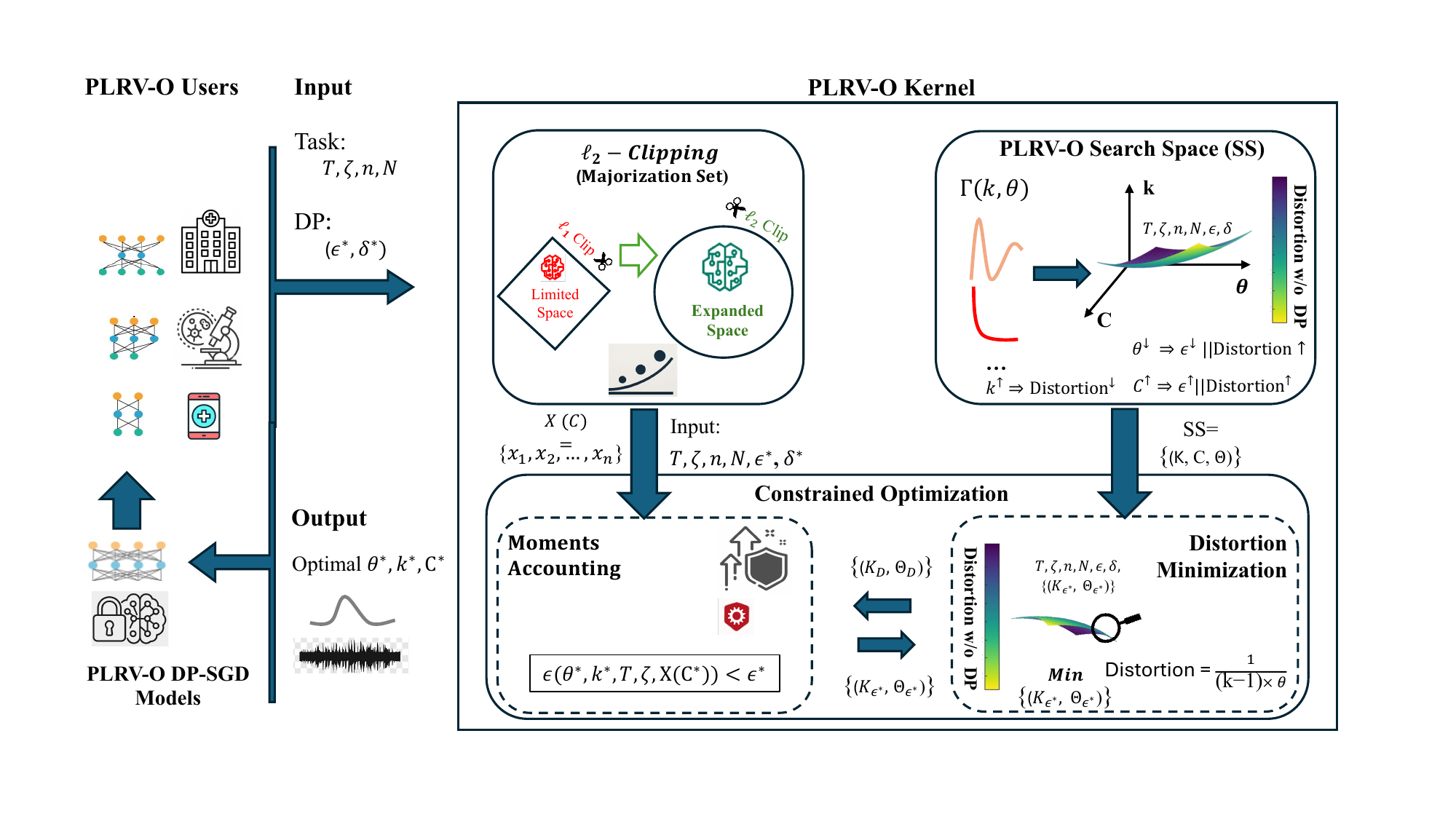}
    \caption{An overview of the \textsf{PLRV-O} framework.}
    \vspace{0.1in}
    \label{fig:overview}
\end{figure*}

\noindent\textbf{Gamma-Distributed Seeds in PLRV.} This composability enables flexible MAF dynamics optimized for learning parameters \( X = (\mathrm{E}, B, C) \) and privacy targets \( (\epsilon, \delta) \). Positive-domain distributions with known MGFs—e.g., exponential, chi-squared, Weibull—can be combined to tailor noise for specific trade-offs. Among them, gamma distributions are especially effective due to their two-parameter flexibility and maximum entropy property (see Figure~\ref{fig:segreagation}).\footnote{The gamma distribution maximizes entropy among distributions with fixed positive mean \( E[X] = \alpha / \lambda \) and fixed \( E[\ln X] = \psi(\alpha) - \ln \lambda \), where \( \psi \) is the digamma function~\cite{Walker1965ProbabilityTA}.}  We adopt this family, denoted \(\Gamma\)-PLRV, in the remainder of the paper.

\begin{definition}[\(\Gamma\)-PLRV Noise]
The \(\Gamma\)-PLRV noise is a specific instantiation of the PLRV noise, where the inverse scale parameter \( u \) follows a gamma distribution with shape parameter \( k > 0 \) and scale parameter \( \theta > 0 \). The resulting noise density is expressed as:
\[
p(z; k, \theta) = \int_{0}^{\infty} \frac{b^{k-1} e^{-b/\theta}}{\theta^n \Gamma(k)} \left(\frac{1}{2b}\right)^n \exp\left(-\frac{\|z\|_1}{b}\right) \, db,
\]

where \(\Gamma(k)\) is the gamma function. 
\end{definition}
\noindent
The \(\Gamma\)-PLRV class allows flexibility in adjusting the noise properties through \( k \) and \( \theta \).

\section{The \textsf{PLRV-O} Framework}
\label{sec:analysis}

The \textsf{PLRV-O} framework (see Figure~\ref{fig:overview}) 
enables task-specific optimization of DP training by adapting noise and clip parameters 
to the learning setup. The \textsf{PLRV-O} Kernel takes user-provided input \(\big[T, \zeta, n, N, \epsilon^*, \delta^*\big]\) and searches for the optimal configuration \((\theta^*, k^*, C^*)\) using Algorithms~\ref{alg:optim} and \ref{alg:sumplrvo}. The selected noise and clip parameters are then applied in our \textsf{PLRV-O} DP-SGD (see Algorithm~\ref{alg:dp-adam} in the Appendix). \textsf{PLRV-O} comprises three key modules:

\vspace{0.05in}

\noindent\textbf{(1) \textsf{PLRV-O} Search Space (SS).} This module constructs a candidate set of \textsf{PLRV-O} configurations—triplets of shape $k$, scale $\theta$, and clipping threshold $C$—from the user-specified training and privacy parameters $(T, \zeta, n, N, \epsilon^*, \delta^*)$. The result is a search space $\{(k,\theta,C)\}$, described in detail in Section~\ref{sec:opt-plrvo}, with alternative construction strategies discussed in Appendix~\ref{parameter}.

\vspace{0.05in}
\noindent\textbf{(2) $\ell_2$-Clipping (Majorization Set).} This module generates a majorization set based on the model size and clipping parameter (in search space), thereby extending the $\ell_2$-clipping to the \textsf{PLRV-O} mechanism, as discussed in Section~\ref{Majorization}.

\vspace{0.05in}

\noindent\textbf{(3) Constrained Optimization.} 
This module solves a constrained optimization problem to obtain the optimal 
$(k^*,\theta^*, C^*)$. The \textit{Moments Accounting }block (left-hand component) computes 
a tight upper bound on $\epsilon$ given $\delta$ for the \textsf{PLRV-O} mechanism. 
The \textit{Distortion Minimization }block (right-hand component) filters the candidate 
search space to retain only configurations under a specified distortion threshold 
(e.g., distortion $<10$). Both blocks jointly define the constraints for the optimization. 

\begin{algorithm}
\caption{Constrained Optimization of \textsf{PLRV-O} Parameters}
\label{alg:optim}
\begin{algorithmic}[1]
\State \textbf{Input:} Privacy budget $\epsilon(\delta$), sampling rate $\zeta$, steps $T$, max moment order $\lambda_{\max}$, convergence tol.\ $\tau$, clip-span ratio $\rho\!\approx\!2$
\State \textbf{Output:} $(k^*,\theta^*,C^*)$

\Statex \textbf{// Phase I: determine a feasible clip range}
\State $C_{\min}\gets $\textsc{smallest $C$ with above random accuracy after a few epochs}
\State $C_{\max}\gets \rho\,C_{\min}$  \Comment{experimentally, $\rho\approx2$}
\State $\mathbf{c_0}:~ C_{\min}\le C\le C_{\max}$ \Comment{clip bounded search}
\Statex \textbf{// Phase II: joint optimization over $(k,\theta,C)$}

\State \textbf{Objective:} $J(k,\theta,C)= C\,(k-1)\,\theta$ 

\Statex \textbf{// Constraints}
\State $\mathbf{c_1}:~ \mathrm{GammaCDF}(0.1; k,\theta)\approx 0$ \Comment{avoids unstable large-$b$ regimes}
\State $\mathbf{c_2}:~ \epsilon(\delta)  \gets  \textsc{MAF\_\textsf{PLRV-O}}(k,\theta,\zeta,T,C, \delta)$ \Comment{privacy budget match}
\State $\mathbf{c_3}:~ k>1$ \Comment{ensures finite expected $\ell_1$-error}
\State $\mathbf{c_4}:~ \theta>\dfrac{0.1}{k-1}$ \Comment{caps distortion}

\State \textbf{Solve:} $\displaystyle \max_{k,\theta,C}\ J(k,\theta,C)\ \ \text{s.t.}\ \ c_0\text{--}c_4$ \Comment{use a constrained nonlinear solver, e.g., \texttt{fmincon}}
\State \Return $(k^*,\theta^*,C^*)$
\end{algorithmic}
\end{algorithm}

We now analyze the \textsf{PLRV-O} framework, showing that Algorithm~\ref{alg:dp-adam} (in Appendix) solves Problem~\eqref{eq:noise-opt:1} from Section~\ref{sec:problem_statement}.

\begin{algorithm}[t]
\caption{\textsf{PLRV-O} Moments Accountant (\textsc{MAF\_\textsf{PLRV-O}})}
\label{alg:sumplrvo}
\begin{algorithmic}[1]
\Require Moment $\lambda\in\mathbb{N}_+$, Sampling rate $\zeta$, \textsf{PLRV-O} parameters $(k,\theta)$, Population $N$, Clipping threshold $C$
\Ensure Order-$\lambda$ log moment $\alpha_\lambda$ 
\State $n\gets \lambda+1$
\State sum$\gets 0$
\State $\log w_\eta\gets\emptyset$
\Statex \vspace{0.5ex} \textbf{    // Define \textsf{PLRV-O} coefficients}
    \State $b_1(\eta)\gets \dfrac{\eta}{2\eta-1}$
    \State $b_2(\eta)\gets \dfrac{\eta-1}{2\eta-1}$
    \State $\alpha_1(\eta)\gets C(\eta-1)\theta$
    \State $\alpha_2(\eta)\gets C\eta\theta$
\For{$i=1$ \textbf{to} $N$} \Comment{can be parallelized}
  \State $g\gets \sqrt{i}-\sqrt{i-1}$
  \State $\ell^F\gets\emptyset$
  \For{$\eta=0$ \textbf{to} $n$} \Comment{logsum: can be parallelized}
    \Statex \vspace{0.5ex} \textbf{\quad \,\quad \, \,    // Log binomial weights}
    \State $\log w_\eta \gets \log\binom{n}{\eta} + (n-\eta)\log(1-\zeta) + \eta\log \zeta$
    \Statex \vspace{0.5ex} \textbf{\quad \,\quad \, \,   // Log \textsf{PLRV-O} terms (for all $\eta$)}
    \State $\ell^{-} \gets \log b_1(\eta)\;-\;k\,\log\!\big(1-\alpha_1(\eta)\,g\big)$
    \State $\ell^{+} \gets \log b_2(\eta)\;-\;k\,\log\!\big(1+\alpha_2(\eta)\,g\big)$
    \State max$_\ell\gets\max(\ell^{-},\ell^{+})$
    \Statex \vspace{0.5ex} \textbf{\quad \,\quad \, \, // Two-branch merge (log-sum-exp)}
    \State $\ell^{F}_\eta \gets max_\ell + \log\!\big(e^{\ell^{-}-max_\ell}+e^{\ell^{+}-max_\ell}\big)$
      \EndFor
    \Statex \vspace{0.5ex} \textbf{\quad \,  // Combine with weights (log-sum-exp over $\eta$)}
    \State \small $\ell_i \gets \max_{\eta}\{\log w_\eta + \ell^{F}_\eta\} + \log\!\sum_{\eta} \exp\!\big((\log w_\eta + \ell^{F}_\eta) - \max_{\eta}\big)$ \normalsize
    \State sum$\gets$ sum + $\ell_i$
\EndFor
\State $\alpha_\lambda \gets$ sum
\State \Return $\alpha_\lambda$ 
\end{algorithmic}
\end{algorithm}

\subsection{Analysis: Formulating \textsf{PLRV-O} Constraints}
\label{sec:opt-plrvo}
Recent work on scaling laws for baseline Gaussian DP-SGD, e.g., Sander et al.~\cite{10.5555/3618408.3619650}, shows in a one-dimensional analysis that the signal-to-noise ratio (SNR) governs the privacy budget. In this view, the effective noise multiplier $\sigma/\zeta$ captures learning degradation under a fixed budget, independent of $C$, since clipping cancels out in the SNR and does not appear in privacy accounting (only entering through the noise scale $\sigma(\epsilon)\!\times\!C$).
 Consequently, once $C$ is large enough for convergence, further increases merely add distortion without improving accuracy.
\begin{figure}[ht]
    \centering
\includegraphics[width=0.9\linewidth, trim=0pt 180pt 20pt 213pt, clip]{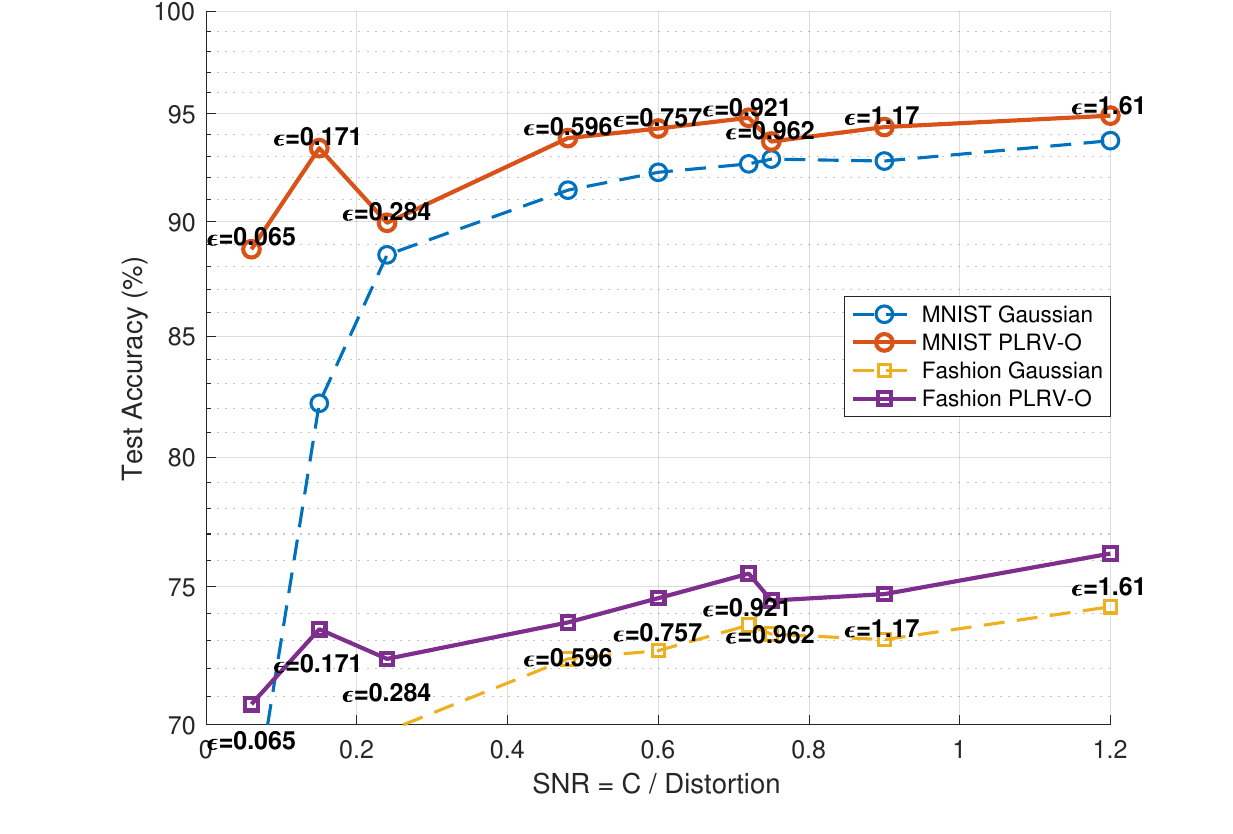}
    \caption{Signal-to-noise ratio (SNR) as a surrogate utility function for optimizing \textsf{PLRV-O} parameters.}
    \label{fig:SNR}
\end{figure}

For \textsf{PLRV-O}, however, clipping appears directly in the MAF equation. This coupling means the privacy–utility trade-off depends jointly on $(k,\theta,C)$, with possible inflection points where changes in $C$ alter the balance between distortion and privacy. Thus, unlike Gaussian, PLRV-O requires joint optimization over all three parameters rather than fixing the clip a priori. Figure~\ref{fig:SNR} confirms that this SNR definition closely tracks accuracy across privacy budgets, with each point annotated by its $\epsilon$ value. 

Motivated by this connection, we adopt an analogous SNR view for \textsf{PLRV-O}. In particular, we use the expected $\ell_1$-error (proved in Appendix~\ref{apdx:proof:error_PLRV}) as a tractable surrogate for distortion, and combine it with clipping to obtain
\[
\mathrm{SNR}(k,\theta,C) \;=\; \frac{C}{\ell_1(z)} \;=\; C\,(k-1)\theta.
\]
The distortion of \textsf{PLRV-O} noise, characterized in Theorem~\ref{error_PLRV} (proved in Appendix~\ref{apdx:proof:error_PLRV}). 

\begin{theorem}[Expected $\ell_1$-Error of $\Gamma$-\textsf{PLRV-O} Noise] \label{error_PLRV} Let $z$ be \textsf{PLRV-O} noise with inverse scale $u$. Then the expected $\ell_1$-error is \[ \ell_1(z) = \mathbb{E}[\|z\|_1] = \int_{0}^{\infty} \mathcal{M}_u(-z)\,dz. \] For $\Gamma(k,\theta)$-\textsf{PLRV-O} noise with $k>1$, this simplifies to $\ell_1(z)=\tfrac{1}{(k-1)\theta}$, and diverges otherwise. \end{theorem}

This observation allows us to formulate PLRV-O parameter optimization directly as maximizing
\(
J(k,\theta,C) = C (k-1)\theta
\),
subject to feasibility, stability, and privacy constraints. 

The essential constraints arise from both distributional validity and the privacy budget:

\vspace{3pt}
 \noindent \textbf{$\mathbf{c_0}$-Clipping Bounds:} $C_{\min} \le C \le C_{\max}$, where $C_{\min}$ is the smallest value that ensures convergence and $C_{\max}\!\approx\!2C_{\min}$ prevents unnecessary inflation.

    \vspace{3pt}
\noindent \textbf{Moment Generating Function:} since $u \sim \Gamma(k,\theta)$ with $b=1/u$, feasibility requires $\theta < 1/t$, equivalently $\lambda_{\max} C \theta < 1$, because
    \[
    \mathcal{M}_u(t) = (1 - t\theta)^{-k}, \quad t < \tfrac{1}{\theta}.
    \]
    With $C < 1$ in practice and typically $\theta < 10^{-3}$, this condition is loose, allowing $\lambda_{\max} \approx 10^3$.

  \vspace{3pt}  
\noindent\textbf{$\mathbf{c_1}$- Laplace Scale-parameter Stability:} $\text{GammaCDF}(0.1; k,\theta) \approx 0$,
    which suppresses configurations with $b > 10$ and avoids unstable noise regimes.

    \vspace{3pt}
\noindent \textbf{$\mathbf{c_2}$- Privacy Constraint:}
 $\epsilon(\delta)  \gets  \textsc{MAF\_\textsf{PLRV-O}}(k,\theta,\zeta,T,C, \delta)$.

    \vspace{3pt}
\noindent \textbf{$\mathbf{c_3}$ \& $\mathbf{c_4}$- Distortion Constraints:} $\mathbf{c_3:}~~k>1,~~~ 
\mathbf{c_4:}~\theta \geq \tfrac{0.1}{k-1}$.

Here $c_3$ ensures finite expected $\ell_1$-error, while $c_4$ caps distortion at 10, a threshold that in practice corresponds to unstable regimes with weak accuracy. Together, constraints $\mathbf{c_0}$–$\mathbf{c_4}$ define the feasible search space, yielding the constrained optimization problem:
\[
\max_{k,\theta,C}\; J(k,\theta,C) = C\,(k-1)\theta
\quad \text{s.t.} \quad c_0\text{--}c_4.
\]

We solve this nonlinear program using a standard nonlinear constrained optimization routine, which accommodates bound, linear, and nonlinear constraints, to obtain the optimal $(k^*,\theta^*,C^*)$.

\section{Evaluation}
\label{sec:eval}

We evaluate DP-SGD with Gaussian and \textsf{PLRV-O} on computer vision (CV) and natural language processing (NLP) tasks, and aim to: (i) assess task utility,  
(ii) compare empirical privacy guarantees vs. theoretical counterparts, and  
(iii) analyze the role of the parameters $k$, $\theta$, and $C$ for interpretability of \textsf{PLRV-O}. Results on runtime, convergence, and boosting orthogonal methods are also provided.

\subsection{Performance for CV Tasks}

We evaluated our method with Gaussian-based DP-SGD on CIFAR-10, CIFAR-100~\cite{krizhevsky2009learning}, and MNIST~\cite{lecun1998gradient}. To demonstrate that the framework can optimize training noise from scratch, it was tested on three image classification models: a 4-layer CNN, a ResNet18, and a Vision Transformer (ViT) model with stride 16. The model sizes are about 26 thousand, 11 million, and 85 million, respectively. 
In these experiments, the clipping threshold has a wide range, from $0.3$ to $15$, and $\theta$ is very small (e.g., $5\times10^{-4}$), which is consistent with the theory that the parameter $k$ controls the accuracy and $\theta$ controls the stability. 
The ideal result, then, will have $larger$ $k$ for less destructive noise and thus better accuracy, while also having smaller $\theta$ for noise generation stability. 

\vspace{-0.05in}

\begin{figure}[!h]
    \centering
    \begin{subfigure}[b]{0.225\textwidth}       \includegraphics[width=\linewidth]{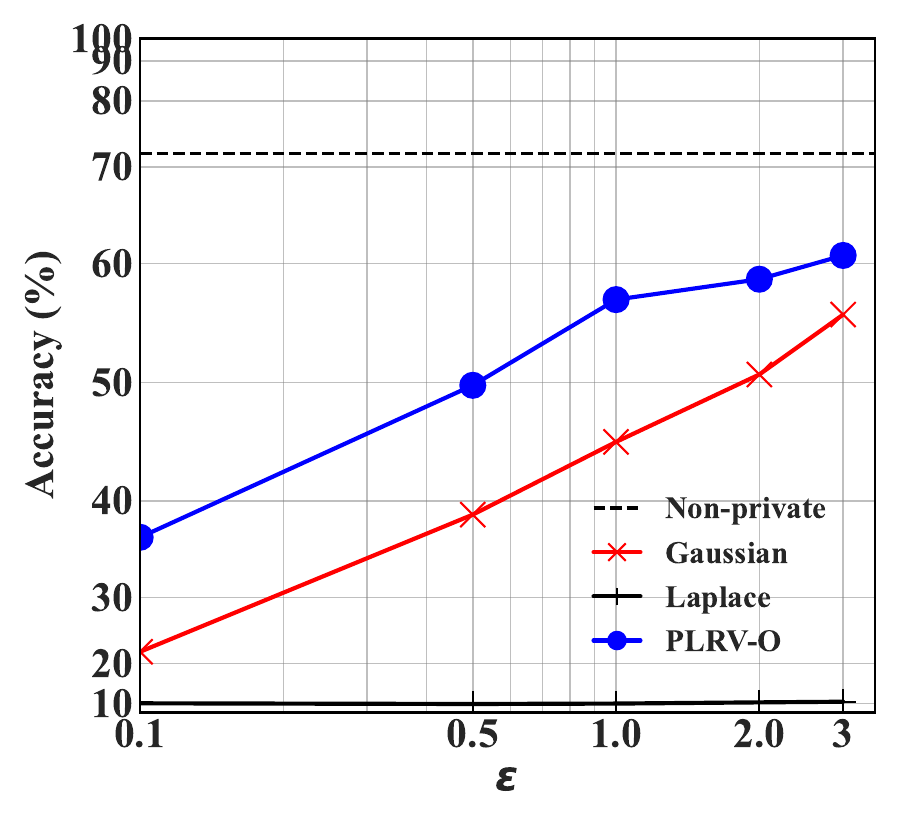}
        \caption{ResNet18 with CIFAR-10}
    \end{subfigure}
    \begin{subfigure}[b]{0.22\textwidth}
\includegraphics[width=\linewidth]{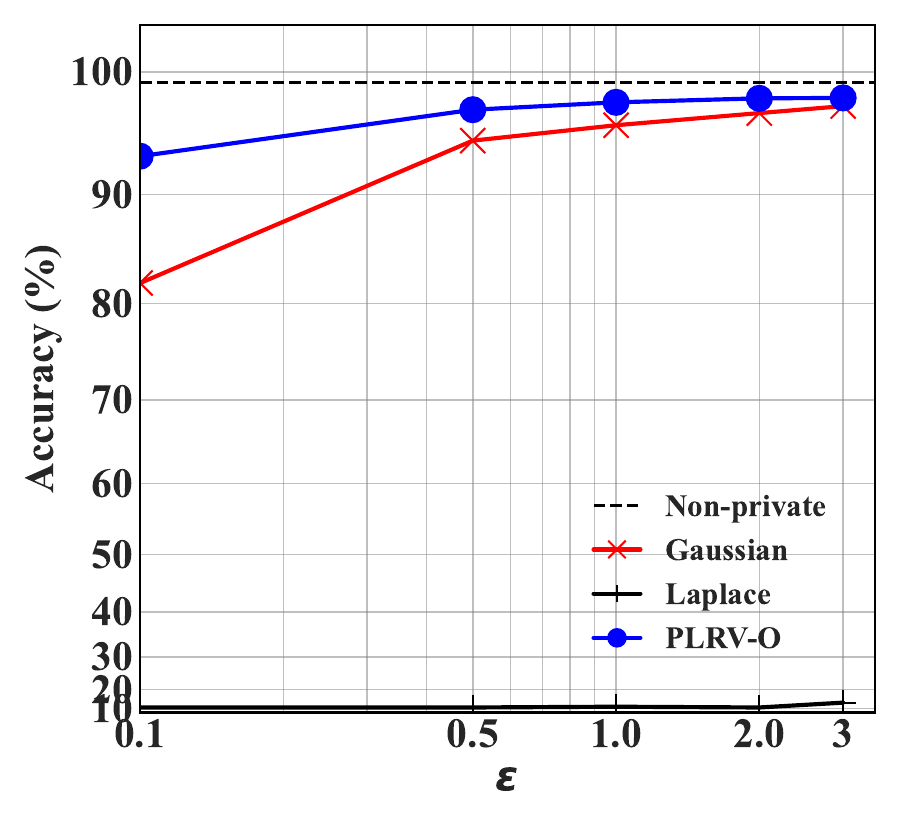}
\caption{Small CNN with MNIST}
    \end{subfigure}
    \vspace{0.15in}
   \caption{Model evaluation results for CV tasks.}
    \label{fig:utility_cv}
\end{figure}

Figure~\ref{fig:utility_cv} presents the model utility across a range of privacy budgets 
$\epsilon$, comparing Gaussian, Laplace, and \textsf{PLRV-O} mechanisms, on the convolutional network including ResNet18 and 4-layer CNN, trained on CIFAR-10 and CIFAR-100, respectively. Across both subplots, \textsf{PLRV-O} consistently outperforms Gaussian and Laplace. Specifically, for ResNet18 on CIFAR-10, \textsf{PLRV-O} achieves up to 30\% higher accuracy than Gaussian under low privacy budgets. For ViT with CIFAR-10, \textsf{PLRV-O} surpasses Gaussian by approximately 10\% 
(see Table~\ref{tab:dp_mechanism_comparison}). 
These results highlight the effectiveness of \textsf{PLRV-O}, especially in low privacy regimes, and demonstrate its robustness across both CNNs and transformer-based models.

\vspace{0.1in}

\begin{table}[!h]
\small
\setlength{\tabcolsep}{2pt}
\centering
\caption{\textsf{PLRV-O} and Gaussian mechanisms for fine-tuning a ViT model (stride = 16). The base model was trained on ImageNet and fine-tuned on CIFAR-10. Noise parameters are ($k$, $\theta$) for \textsf{PLRV-O} and $\sigma$ for Gaussian.}
\begin{adjustbox}{max width=\columnwidth}
\begin{tabularx}{\columnwidth}{l c c c c c c c p{2.6cm}}
\toprule
 & Acc. (\%) & $\mathbf{\epsilon}$ & Steps & S. rate & $\delta$ & Dist. & Clip & Noise Params \\
\midrule
\textsf{PLRV-O}     & \bf 93.63 & 1.7  & 250 & 0.01024 & $2{\times}10^{-5}$ & 8.58 & 10.0 & $k{=}141.06$, $\theta{=}8.32\times10^{-4}$ \\
\textsf{PLRV-O}     & \bf 92.90 & 0.46  & 250 & 0.01024 & $2{\times}10^{-5}$ & 9.17 & 5.0  & $k{=}5242.4$, $\theta{=}2.08\times10^{-5}$ \\
\midrule
Gaussian     & 89.01 & 1.7  & 250 & 0.01024 & $2{\times}10^{-5}$ & 3.77 & 5.0  & $\sigma{=}0.9456$ \\
Gaussian     & 83.93 & 0.46 & 250 & 0.01024 & $2{\times}10^{-5}$ & 22.51  & 15.0 & $\sigma{=}1.8812$ \\
\bottomrule
\end{tabularx}
\end{adjustbox}
\vspace{-0.05in}
\label{tab:dp_mechanism_comparison}
\end{table}

\noindent \textbf{Interpretability of \textsf{PLRV-O} Parameters.} See Table~\ref{tab:cnn_mnist_results}--\ref{tab:cnn_fmnist_results} in Appendix~\ref{exp_appendix:CV} for results on multiple $\epsilon$ values in the high-privacy regime ($\epsilon \leq 1.6$)  
for the CNN model on the MNIST and Fashion-MNIST datasets.
The optimized clipping threshold $C$, $\theta$, and $k$ for the given steps, $q$, $\delta$ were found.
After training with these parameters, the \textsf{PLRV-O} performed far better than the Gaussian mechanism, with accuracy approaching training without noise. 
More results related to the optimized parameters are shown in Appendix~\ref{exp_appendix:parameter}.

\subsection{Performance for NLP Tasks}
We also evaluate \textsf{PLRV-O} on two NLP tasks using language models: sentence classification and text generation.
\subsubsection{Performance on Sentence Classification}
Our evaluation is on GLUE benchmarks (SST-2: distinguish between positive and negative emotions; QNLI: determine whether a context sentence contains the answer to a given question)~\cite{wang2018glue} with experiments conducted on RoBERTa-base, RoBERTa-large~\cite{liu2019roberta}, BERT-base, BERT-large~\cite{devlin2018bert}, comparing performance with the non-private training and baseline DP-SGD method~\cite{li2021large}. 
In detail, SST-2 has more than 60k+ samples in the training set, and QNLI has more than 100k+ samples; SST-2 and QNLI include two classes each. In this experiment, we use the full training and test sets. 

\begin{figure*}[!h]
    \centering
    \begin{subfigure}[b]{0.24\textwidth}       \includegraphics[width=\linewidth]{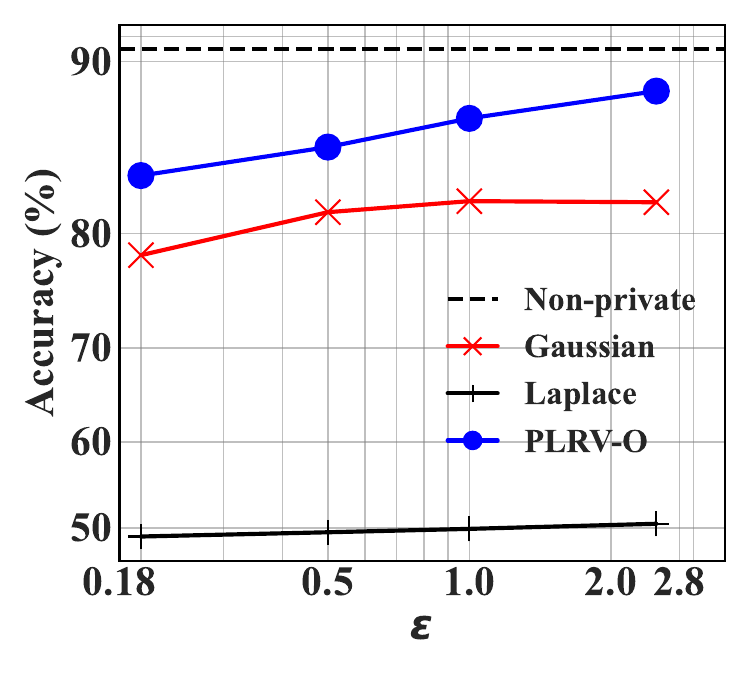}
        \caption{QNLI with BERT-base}
    \end{subfigure} 
    \begin{subfigure}[b]{0.24\textwidth}
\includegraphics[width=\linewidth]{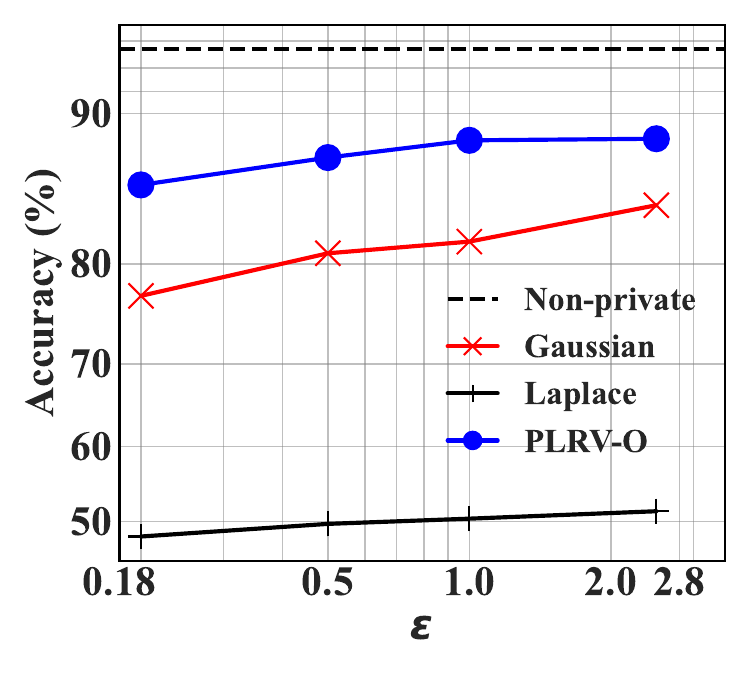}
        \caption{QNLI with BERT-large}
    \end{subfigure}    
    \begin{subfigure}[b]{0.24\textwidth}        \includegraphics[width=\linewidth]{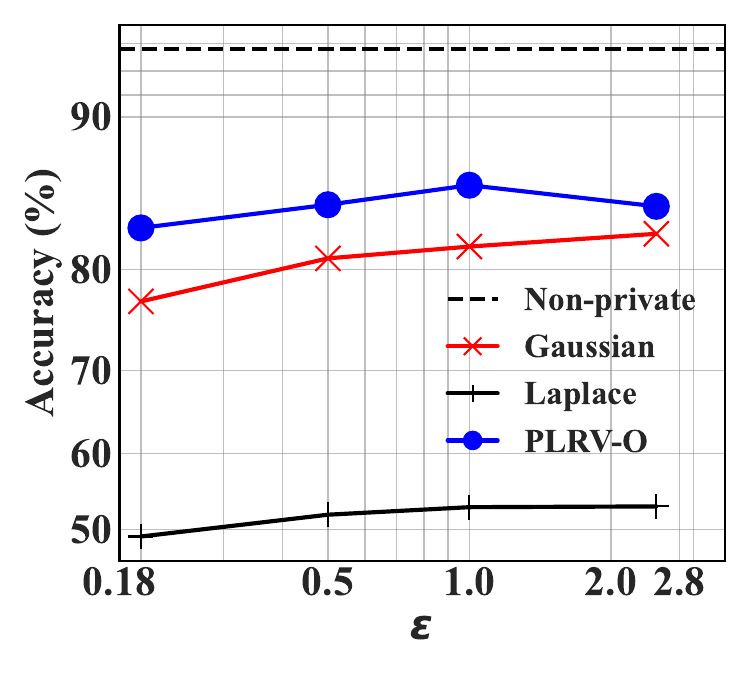}
        \caption{QNLI with RoBERTa-base}
    \end{subfigure}   
    \begin{subfigure}[b]{0.24\textwidth}      \includegraphics[width=\linewidth]{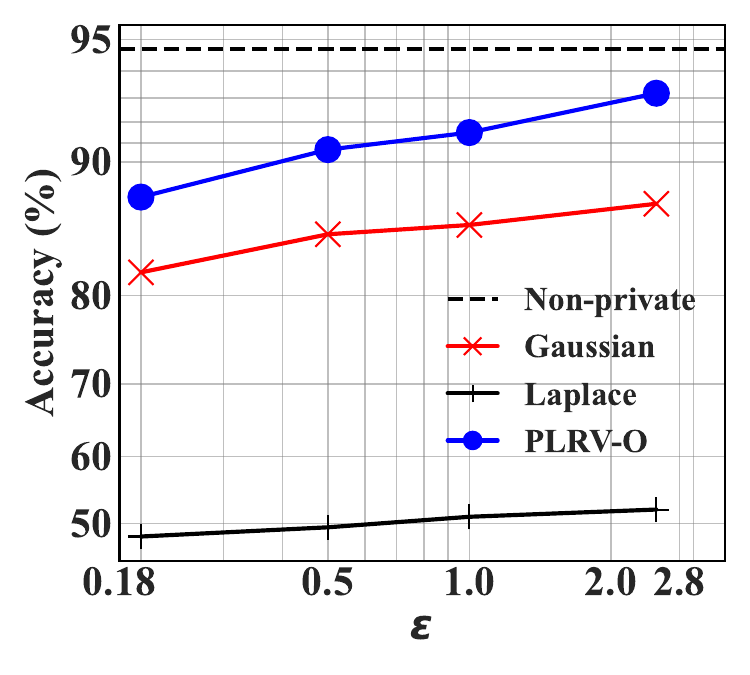}
    \caption{QNLI with RoBERTa-large}
    \end{subfigure}
    \vspace{0.15in}
    \caption{Model evaluation results for QNLI tasks using BERT and RoBERTa (Base and Large).}
    \vspace{0.1in}
    \label{fig:utility_qnli}
\end{figure*}

\begin{figure*}[!h]
    \centering
    \begin{subfigure}[b]{0.24\textwidth}      \includegraphics[width=\linewidth]{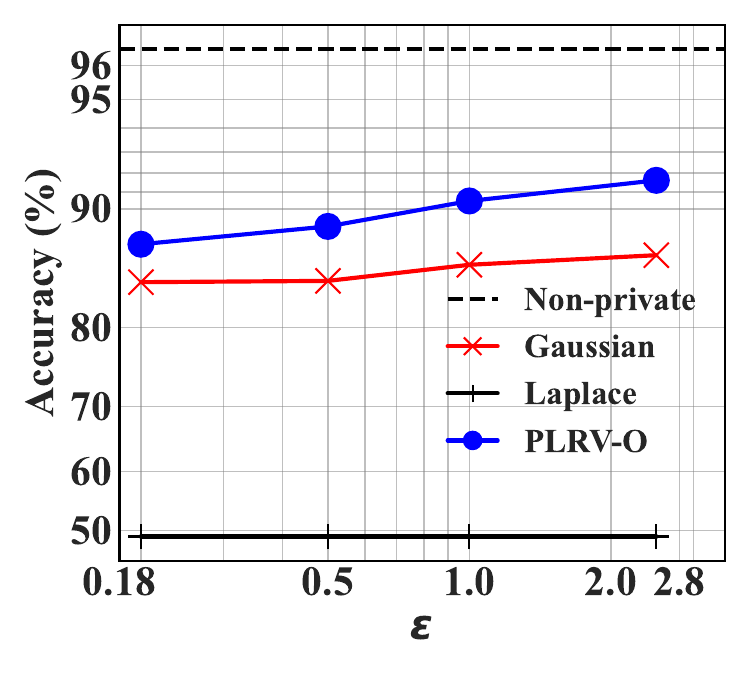}
        \caption{SST-2 with BERT-base}
    \end{subfigure}
    \begin{subfigure}[b]{0.24\textwidth}
    \includegraphics[width=\linewidth]{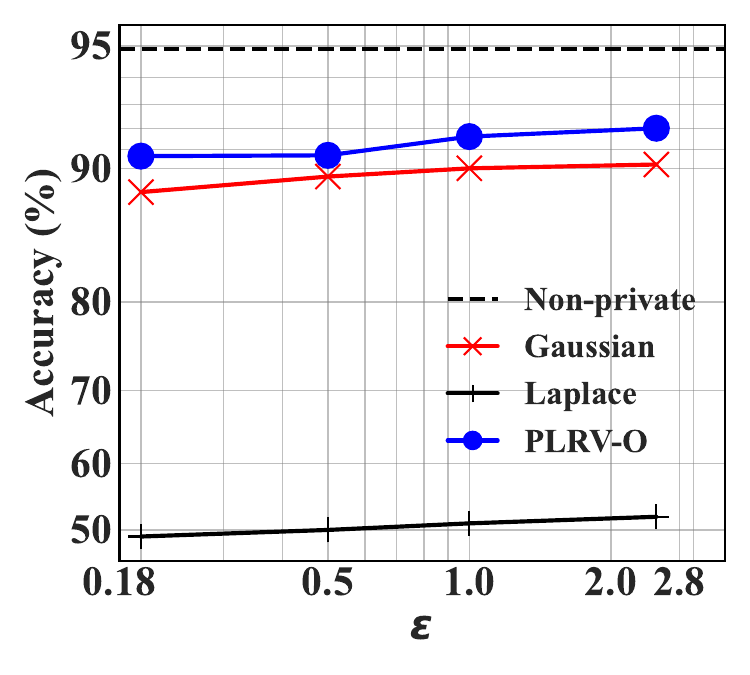}
        \caption{SST-2 with BERT-large}
    \end{subfigure}
    \begin{subfigure}[b]{0.24\textwidth}    \includegraphics[width=\linewidth]{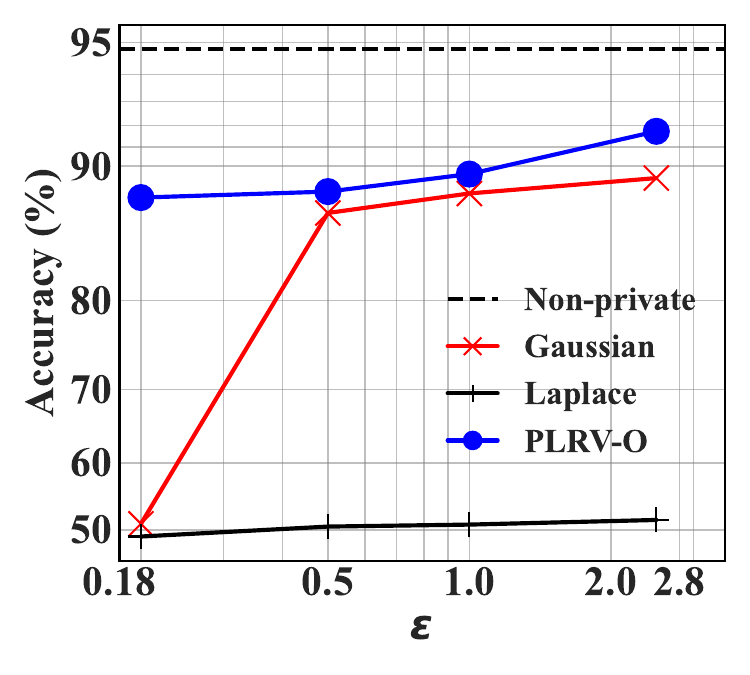}
        \caption{SST-2 with RoBERTa-base}
    \end{subfigure}
    \begin{subfigure}[b]{0.24\textwidth}
  \includegraphics[width=\linewidth]{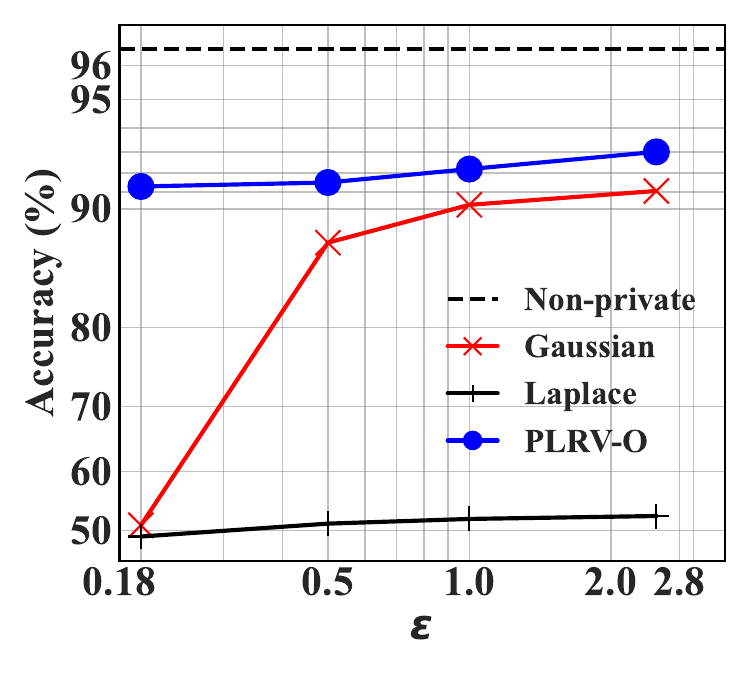}
        \caption{SST-2 with RoBERTa-large}
    \end{subfigure}\vspace{0.15in}
    \caption{Model evaluation results for SST-2 tasks using BERT and RoBERTa (Base and Large).}\vspace{0.05in}
    \label{fig:utility_sst2}
\end{figure*}

In Figures~\ref{fig:utility_qnli} and \ref{fig:utility_sst2}, the $x$-axis shows total privacy loss $\epsilon$, and the $y$-axis shows the accuracy under specific $\epsilon$ after $T$-step composition during fine-tuning. 
From (a) to (d) of Figures ~\ref{fig:utility_qnli} and \ref{fig:utility_sst2}, the non-private model accuracy are $90.5$, $92.7$, $92.8$, $94.7$, $96.4$, $94.9$ $94.8$, and $96.4$ respectively. 
When $\epsilon\leq 3$, our mechanism largely outperforms the baselines, exhibiting close performance with non-private training. 
Specifically, we can observe that the fine-tuning with the Laplace mechanism and $\ell_1$ norm, the accuracy is around $50\%$, which means the model is not trainable since the language model is large and the $\ell_1$ norm is much larger than the $\ell_2$ norm of Gaussian noise. We can also observe that our performance (blue line) is better than the performance of the Gaussian mechanism with around $3\%$ - $7\%$. 

\subsubsection{Performance on Text Generation Task}
\label{5.3}

We experiment with the official
pipeline on the E2E dataset~\cite{novikova2017e2e} and DART dataset~\cite{nan-etal-2021-dart}. We fine-tune the DistilGPT2, GPT-2, GPT2-medium, and GPT2-large models on the E2E dataset with $\delta=10^{-5}$, evaluating five metrics: BLEU, NIST, METEOR, ROUGE-L, and CIDEr.
We evaluate \textsf{PLRV-O} and DP-SGD~\cite{li2021large} under  $\delta=10^{-5}$. We also apply the same settings of privacy budget at each iteration, weights, and clipping threshold as in the sentence classification. We only employ a batch size of $1024$, which causes the total privacy budget $\epsilon<0.2$. 

\vspace{-0.05in}

\begin{table}[!h]
\centering
\vspace{0.1in}
\caption{DistilGPT2 model on the E2E dataset.}
\resizebox{0.45\textwidth}{!}{%
\begin{tabular}{|>{\centering\arraybackslash}m{1.4cm}|>{\centering\arraybackslash}m{0.7cm}|%
>{\centering\arraybackslash}m{0.9cm}|>{\centering\arraybackslash}m{1.0cm}|>{\centering\arraybackslash}m{1.2cm}|%
>{\centering\arraybackslash}m{1.5cm}|>{\centering\arraybackslash}m{1.1cm}|}
\hline
\textbf{Noise} & \textbf{$\epsilon$} & \textbf{BLEU} & \textbf{NIST} & \textbf{METEOR} & \textbf{ROUGE-L} & \textbf{CIDEr} \\
\hline
\multirow{4}{*}{Gaussian} 
 & 0.2 & 25.29 & 0.3424 & 0.2304 & 0.4925 & 0.5883 \\
\cline{2-7}
 & 0.5 & 36.73 & 1.0098 & 0.2829 & 0.5733 & 0.7638 \\
\cline{2-7}
 & 1   & 50.55 & 3.5693 & 0.3225 & 0.6159 & 1.0901 \\
\cline{2-7}
 & 2.5 & 59.82 & 5.6407 & 0.3646 & 0.6589 & 1.5164 \\
\hline
\multirow{4}{*}{\textsf{PLRV-O}}
 & 0.2 & 58.50 & 5.3100 & 0.3578 & 0.6534 & 1.3862 \\
\cline{2-7}
 & 0.5 & 64.16 & 7.4953 & 0.3944 & 0.6681 & 1.8563 \\
\cline{2-7}
 & 1   & 63.63 & 7.2425 & 0.3894 & 0.6675 & 1.7844 \\
\cline{2-7}
 & 2.5 & 64.94 & 8.1061 & 0.4140 & 0.6790 & 2.0404 \\
\hline
\end{tabular}
}
\label{tab:distilgpt2_halfwidth}
\end{table}

\vspace{-0.05in}

\begin{table}[!h]
\centering
\small
\caption{GPT-2 model on the E2E dataset.}
\resizebox{0.45\textwidth}{!}{%
\begin{tabular}{|>{\centering\arraybackslash}m{1.4cm}|>{\centering\arraybackslash}m{0.7cm}|%
>{\centering\arraybackslash}m{0.9cm}|>{\centering\arraybackslash}m{1.0cm}|>{\centering\arraybackslash}m{1.2cm}|%
>{\centering\arraybackslash}m{1.5cm}|>{\centering\arraybackslash}m{1.1cm}|}
\hline
\textbf{Noise} & \textbf{$\epsilon$} & \textbf{BLEU} & \textbf{NIST} & \textbf{METEOR} & \textbf{ROUGE-L} & \textbf{CIDEr} \\
\hline
\multirow{4}{*}{Gaussian} 
 & 0.2 & 32.82 & 0.6487 & 0.2733 & 0.5642 & 0.7618 \\
\cline{2-7}
 & 0.5 & 42.72 & 1.9037 & 0.3037 & 0.5887 & 0.9699 \\
\cline{2-7}
 & 1   & 52.11 & 3.7071 & 0.3313 & 0.6281 & 1.2369 \\
\cline{2-7}
 & 2.5 & 59.02 & 5.6853 & 0.3602 & 0.6520 & 1.5457 \\
\hline
\multirow{4}{*}{\textsf{PLRV-O}}
 & 0.2 & 60.21 & 5.7860 & 0.3688 & 0.6628 & 1.5285 \\
\cline{2-7}
 & 0.5 & 62.98 & 7.5666 & 0.4025 & 0.6767 & 1.9438 \\
\cline{2-7}
 & 1   & 62.73 & 7.0415 & 0.3928 & 0.6750 & 1.8164 \\
\cline{2-7}
 & 2.5 & 65.27 & 8.3764 & 0.4360 & 0.6858 & 2.2612 \\
\hline
\end{tabular}
}\vspace{-0.15in}
\label{tab:gpt2_halfwidth}
\end{table}

\begin{table}[!h]
\centering
\vspace{0.15in}
\caption{E2E evaluation results of GPT2-medium model under different noise types and privacy budgets ($\epsilon$).}
\resizebox{0.45\textwidth}{!}{%
\begin{tabular}{|>{\centering\arraybackslash}m{1.4cm}|>{\centering\arraybackslash}m{0.7cm}|%
>{\centering\arraybackslash}m{0.9cm}|>{\centering\arraybackslash}m{1.0cm}|>{\centering\arraybackslash}m{1.2cm}|%
>{\centering\arraybackslash}m{1.5cm}|>{\centering\arraybackslash}m{1.1cm}|}
\hline
\textbf{Noise} & \textbf{$\epsilon$} & \textbf{BLEU} & \textbf{NIST} & \textbf{METEOR} & \textbf{ROUGE-L} & \textbf{CIDEr} \\
\hline
\multirow{4}{*}{Gaussian} 
 & 0.2 & 37.38 & 1.4594 & 0.2870 & 0.5662 & 0.8094 \\
\cline{2-7}
 & 0.5 & 52.96 & 3.8638 & 0.3447 & 0.6206 & 1.2907 \\
\cline{2-7}
 & 1   & 59.00 & 6.0868 & 0.3709 & 0.6452 & 1.6023 \\
\cline{2-7}
 & 2.5 & 62.22 & 8.0526 & 0.4063 & 0.6684 & 1.9906 \\
\hline
\multirow{4}{*}{\textsf{PLRV-O}}
 & 0.2 & 58.79 & 6.1828 & 0.3792 & 0.6580 & 1.6706 \\
\cline{2-7}
 & 0.5 & 64.47 & 8.3884 & 0.4295 & 0.6789 & 2.2338 \\
\cline{2-7}
 & 1   & 63.67 & 8.3098 & 0.4220 & 0.6785 & 2.1590 \\
\cline{2-7}
 & 2.5 & 66.04 & 8.4342 & 0.4423 & 0.6922 & 2.3728 \\
\hline
\end{tabular}
}\vspace{0.1in}
\label{tab:gpt2-medium_halfwidth}
\end{table}

\begin{table}[!h]
\centering
\caption{E2E evaluation results of GPT2-large model under different noise types and privacy budgets ($\epsilon$).}
\resizebox{0.45\textwidth}{!}{%
\begin{tabular}{|>{\centering\arraybackslash}m{1.4cm}|>{\centering\arraybackslash}m{0.7cm}|%
>{\centering\arraybackslash}m{0.9cm}|>{\centering\arraybackslash}m{1.0cm}|>{\centering\arraybackslash}m{1.2cm}|%
>{\centering\arraybackslash}m{1.5cm}|>{\centering\arraybackslash}m{1.1cm}|}
\hline
\textbf{Noise} & \textbf{$\epsilon$} & \textbf{BLEU} & \textbf{NIST} & \textbf{METEOR} & \textbf{ROUGE-L} & \textbf{CIDEr} \\
\hline
\multirow{4}{*}{Gaussian} 
 & 0.2 & 39.25 & 1.9838 & 0.2940 & 0.5892 & 0.9242 \\
\cline{2-7}
 & 0.5 & 52.51 & 3.9813 & 0.3525 & 0.6421 & 1.1923 \\
\cline{2-7}
 & 1   & 53.59 & 4.6533 & 0.3317 & 0.6191 & 1.2980 \\
\cline{2-7}
 & 2.5 & 62.24 & 8.0916 & 0.4118 & 0.6661 & 2.0894 \\
\hline
\multirow{4}{*}{\textsf{PLRV-O}}
 & 0.2 & 60.71 & 6.4304 & 0.3743 & 0.6501 & 1.6485 \\
\cline{2-7}
 & 0.5 & 65.02 & 8.4223 & 0.4226 & 0.6740 & 2.0988 \\
\cline{2-7}
 & 1   & 64.21 & 8.3605 & 0.4179 & 0.6676 & 2.0811 \\
\cline{2-7}
 & 2.5 & 67.37 & 8.6457 & 0.4496 & 0.6952 & 2.3657 \\
\hline
\end{tabular}
}
\label{tab:gpt2-large_halfwidth}
\end{table}
\normalsize

Tables~\ref{tab:distilgpt2_halfwidth}-\ref{tab:gpt2-large_halfwidth} present the results with the five different metrics, following~\cite{yu2021differentially}. 
We observe that \textsf{PLRV-O} yields results that are more closely aligned with the non-private results than Gaussian (larger values of all these metrics exhibit more accurate generated texts). Note the improvement can be up to $50\%$ on some metrics (e.g., CIDEr). 
Tables~\ref{tab:distilgpt2_dart} and \ref{tab:gpt2_dart} show the generation performance for the DART dataset. We can observe a similar trend to the E2E dataset.

\renewcommand{\arraystretch}{1.0}
\scriptsize

\begin{table}[!h]
\centering
\vspace{0.1in}
\caption{DART evaluation results of DistilGPT2 model under different noise types and privacy budgets ($\epsilon$).}
\resizebox{0.45\textwidth}{!}{%
\begin{tabular}{|>{\centering\arraybackslash}m{1.4cm}|>{\centering\arraybackslash}m{0.7cm}|%
>{\centering\arraybackslash}m{0.9cm}|>{\centering\arraybackslash}m{1.0cm}|>{\centering\arraybackslash}m{1.2cm}|%
>{\centering\arraybackslash}m{1.5cm}|>{\centering\arraybackslash}m{1.1cm}|}
\hline
\textbf{Noise} & \textbf{$\epsilon$} & \textbf{BLEU} & \textbf{NIST} & \textbf{METEOR} & \textbf{ROUGE-L} & \textbf{CIDEr} \\
\hline
\multirow{4}{*}{Gaussian} 
 & 0.2 & 3.55 & 0.4290 & 0.0467 & 0.1023 & 0.2909 \\
\cline{2-7}
 & 0.5 & 12.47 & 2.1889 & 0.1300 & 0.2416 & 0.5654 \\
\cline{2-7}
 & 1   & 17.48 & 2.4004 & 0.1838 & 0.3291 & 0.8345 \\
\cline{2-7}
 & 2.5 & 22.26 & 3.9404 & 0.2156 & 0.3816 & 1.0045 \\
\hline
\multirow{4}{*}{\textsf{PLRV-O}}
 & 0.2 & 23.29 & 3.8109 & 0.2292 & 0.4040 & 1.1161 \\
\cline{2-7}
 & 0.5 & 28.42 & 4.0702 & 0.2701 & 0.4779 & 1.5073 \\
\cline{2-7}
 & 1   & 27.89 & 4.1538 & 0.2658 & 0.4688 & 1.4639 \\
\cline{2-7}
 & 2.5 & 30.83 & 4.5753 & 0.2841 & 0.5036 & 1.6856 \\
\hline
\end{tabular}
}
\label{tab:distilgpt2_dart}
\end{table}

\vspace{-0.1in}

\renewcommand{\arraystretch}{1.0}
\scriptsize

\begin{table}[ht]
\centering
\vspace{0.1in}
\caption{DART evaluation results of GPT-2 model under different noise types and privacy budgets ($\epsilon$).}
\resizebox{0.45\textwidth}{!}{%
\begin{tabular}{|>{\centering\arraybackslash}m{1.4cm}|>{\centering\arraybackslash}m{0.7cm}|%
>{\centering\arraybackslash}m{0.9cm}|>{\centering\arraybackslash}m{1.0cm}|>{\centering\arraybackslash}m{1.2cm}|%
>{\centering\arraybackslash}m{1.5cm}|>{\centering\arraybackslash}m{1.1cm}|}
\hline
\textbf{Noise} & \textbf{$\epsilon$} & \textbf{BLEU} & \textbf{NIST} & \textbf{METEOR} & \textbf{ROUGE-L} & \textbf{CIDEr} \\
\hline
\multirow{4}{*}{Gaussian} 
 & 0.2 & 6.95 & 0.0761 & 0.1293 & 0.2625 & 0.6065 \\
\cline{2-7}
 & 0.5 & 20.31 & 2.0213 & 0.2234 & 0.4222 & 1.0794 \\
\cline{2-7}
 & 1   & 25.31 & 3.3642 & 0.2549 & 0.4631 & 1.3252 \\
\cline{2-7}
 & 2.5 & 29.32 & 4.7652 & 0.2772 & 0.4900 & 1.5449 \\
\hline
\multirow{4}{*}{\textsf{PLRV-O}}
 & 0.2 & 25.95 & 2.5579 & 0.2699 & 0.4911 & 1.4419 \\
\cline{2-7}
 & 0.5 & 30.76 & 2.9106 & 0.2891 & 0.5123 & 1.6783 \\
\cline{2-7}
 & 1   & 31.38 & 3.2593 & 0.3029 & 0.5212 & 1.7453 \\
\cline{2-7}
 & 2.5 & 33.66 & 4.7260 & 0.3142 & 0.5407 & 1.9371 \\
\hline
\end{tabular}
}\vspace{-0.1in}
\label{tab:gpt2_dart} 
\end{table}
\normalsize

\subsection{Moments Accountant}

To validate the privacy accountant of our \textsf{PLRV-O}, we evaluate our privacy loss through $T$. 
In this experiment, we choose two different settings of $C$, $\theta$, $k$, and $\zeta$ of our experiments with small distortion to show the privacy loss with $T$. 
For the Gaussian mechanism, we apply the same $C$ and $\zeta$ to get the privacy loss with $T$ with $\sigma=0.1$. 
See Figure~\ref{fig:privacy_loss} for the results.  
For a large clip such as $C=1.0$, the Gaussian mechanism and \textsf{PLRV-O} have similar privacy loss. 
However, the distortion of \textsf{PLRV-O} noise is lower, thus \textsf{PLRV-O} is the more desirable noise under the same privacy loss. 

\begin{figure}[!h]
    \centering
    \begin{subfigure}[b]{0.235\textwidth}       \includegraphics[width=\linewidth]{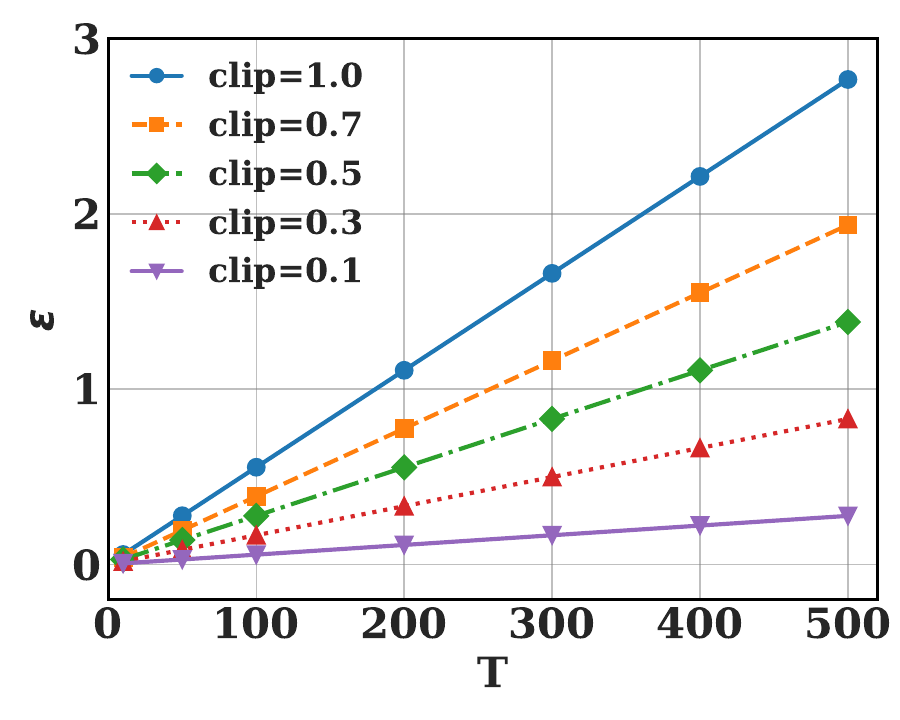}
        \caption{\textsf{PLRV-O} privacy loss vs. T}
    \end{subfigure}
        \begin{subfigure}[b]{0.233\textwidth}       \includegraphics[width=\linewidth]{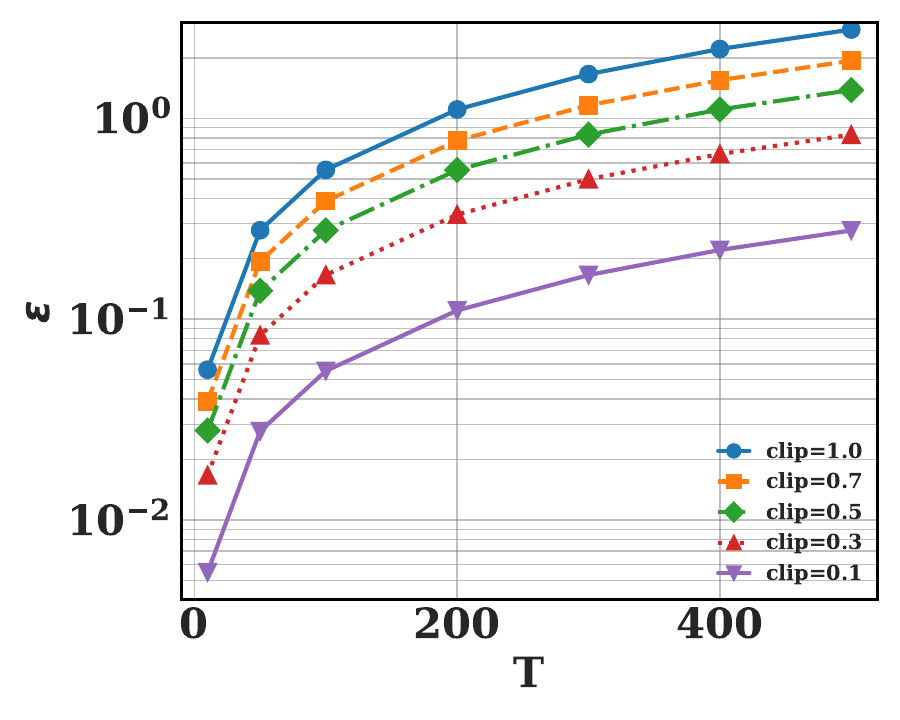}
        \caption{Gaussian privacy loss vs. T}
    \end{subfigure}
    \vspace{0.15in}
    \caption{Privacy loss vs. T, clip. (a) $k=414.2857$, $\theta=2.4196*10^{-4}$, $q=0.00977631$, $\delta=10^{-5}$ and $N=109482240$. (b) $\sigma=0.1$.}
    \label{fig:privacy_loss}
\end{figure}

\begin{figure*}[ht]
    \centering
    \begin{subfigure}[b]{0.24\textwidth}       \includegraphics[width=\linewidth]{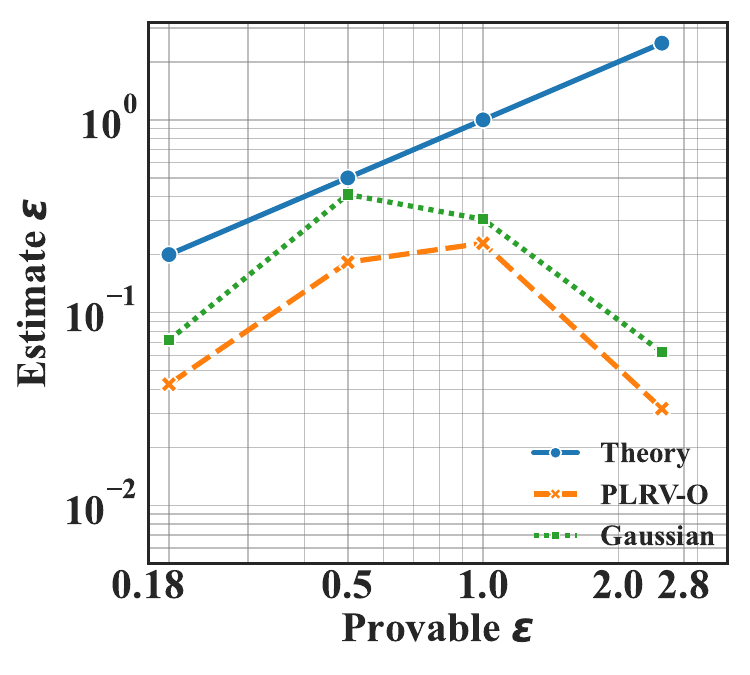}
        \caption{SST-2 with BERT-base}
    \end{subfigure} 
    \begin{subfigure}[b]{0.24\textwidth}
\includegraphics[width=\linewidth]{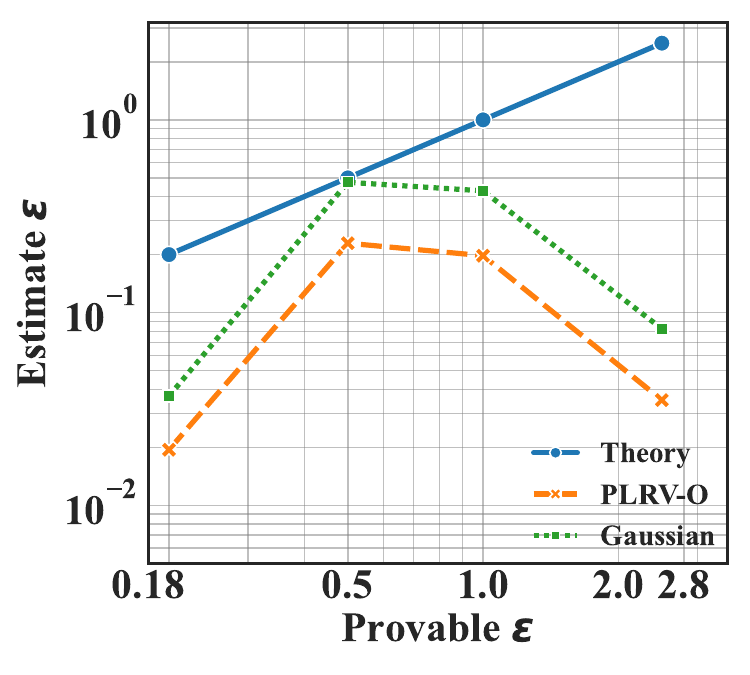}
        \caption{SST-2 with RoBERTa-base}
    \end{subfigure}    
    \begin{subfigure}[b]{0.24\textwidth}        \includegraphics[width=\linewidth]{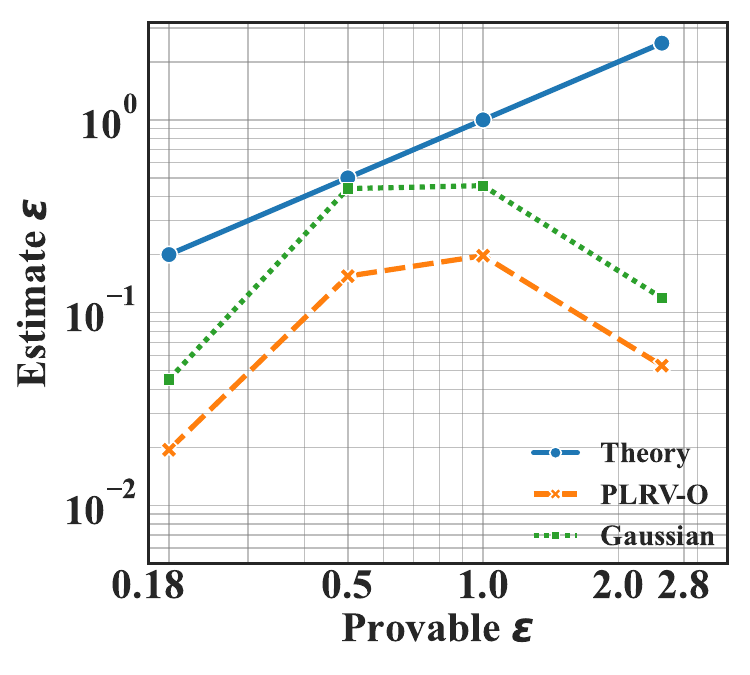}
        \caption{QNLI with BERT-base}
    \end{subfigure}   
    \begin{subfigure}[b]{0.24\textwidth}      \includegraphics[width=\linewidth]{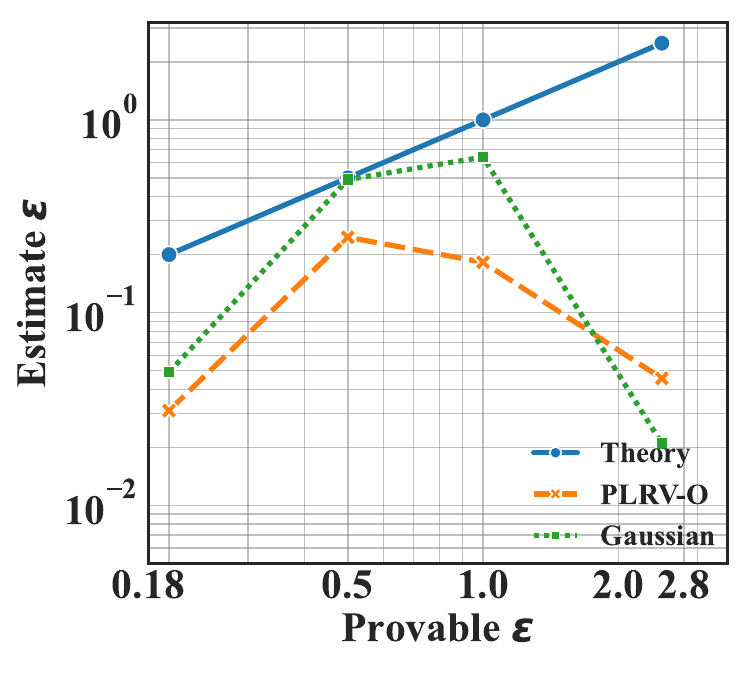}
    
\caption{QNLI with RoBERTa-base}
\end{subfigure}
    \vspace{0.15in}
\caption{NLP performance of privacy attacks ClipBKD.}\vspace{0.05in}
    \label{fig:auditing_nlp}
\end{figure*}

\begin{figure*}[ht]
    \centering
    \begin{subfigure}[b]{0.24\textwidth}       \includegraphics[width=\linewidth]{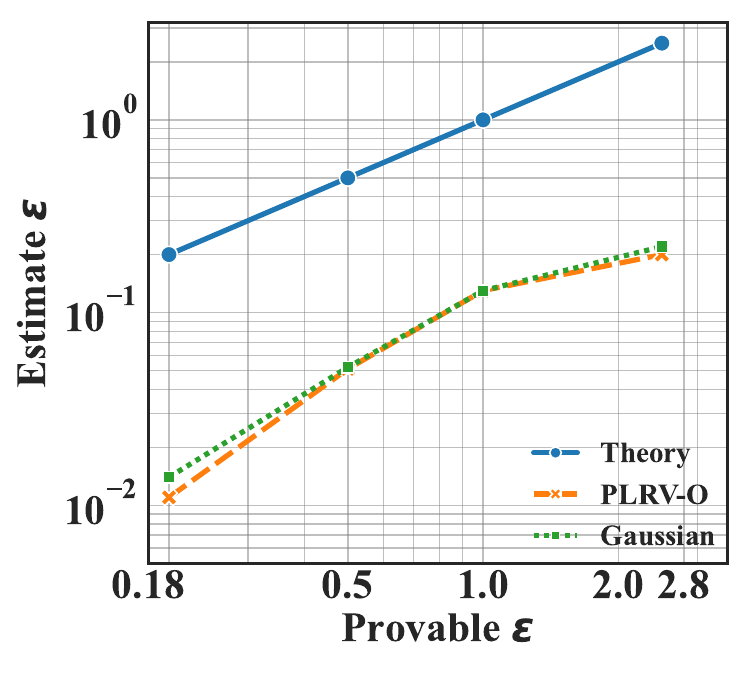}
        \caption{FMNIST with FNN}
    \end{subfigure} 
    \begin{subfigure}[b]{0.24\textwidth}
\includegraphics[width=\linewidth]{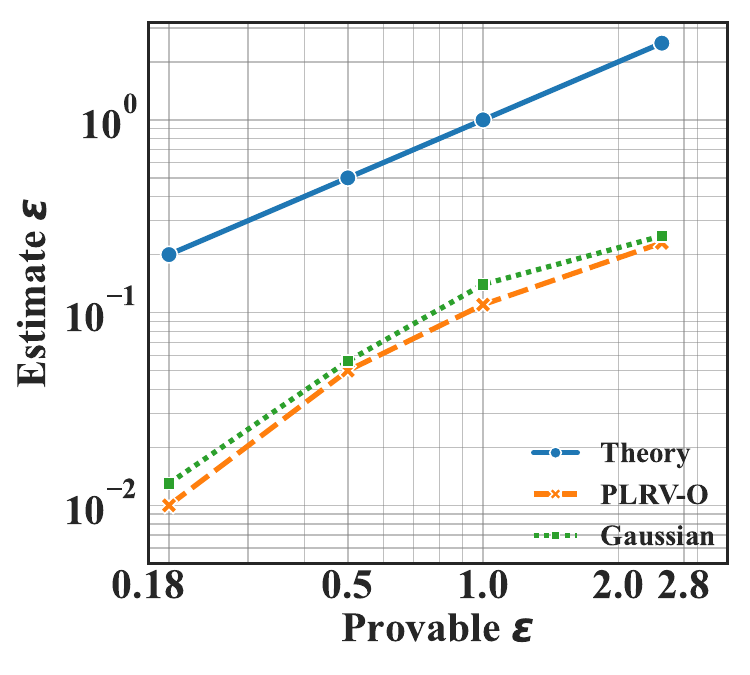}
        \caption{FMNIST with LR}
    \end{subfigure}    
    \begin{subfigure}[b]{0.24\textwidth}        \includegraphics[width=\linewidth]{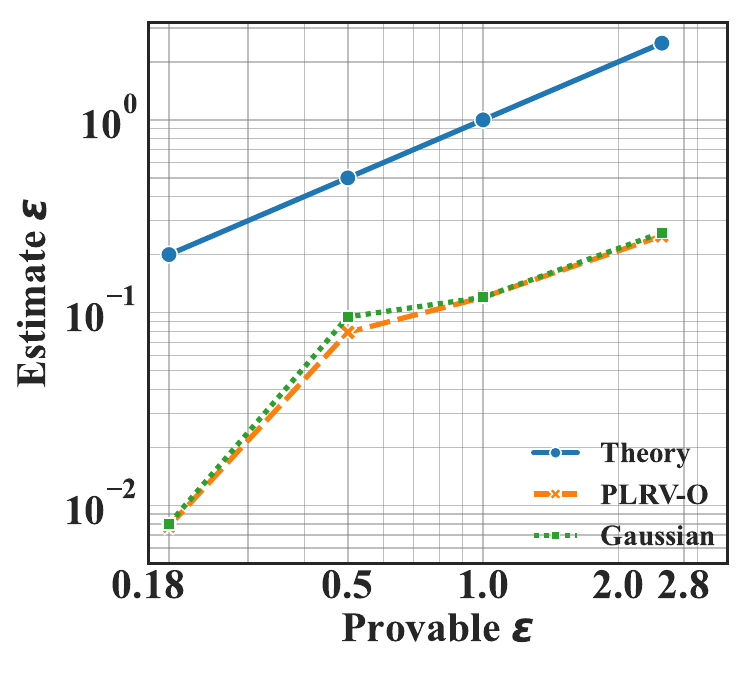}
        \caption{P100 with FNN}
    \end{subfigure}   
    \begin{subfigure}[b]{0.24\textwidth}      \includegraphics[width=\linewidth]{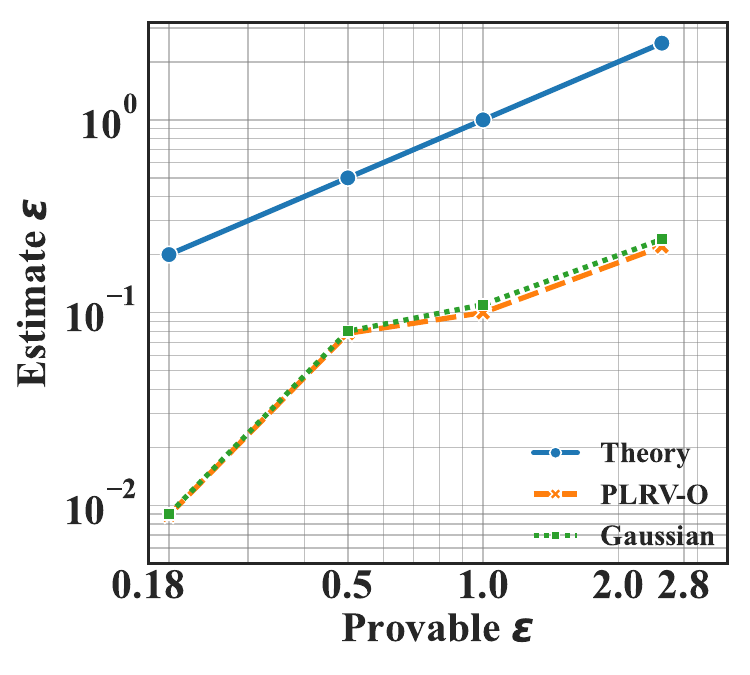}
    
\caption{P100 with LR}
\end{subfigure}
    \vspace{0.15in}
\caption{CV performance of privacy attacks ClipBKD.} \vspace{0.05in}
    \label{fig:auditing_cv}
\end{figure*}

\subsection{Privacy Audit}
\label{sec:auditexp}
While we present theoretical privacy loss via $T$, we also report empirical $\epsilon$ under attacks. Using the auditing model in Appendix~\ref{appendix:audit}, we conduct a privacy audit to validate our guarantees. Specifically, we implement a data poisoning attack with dataset and algorithm $\mathcal{A}$, generating a small poisoning set $S$ of $k$ points and a binary classifier $T$ that distinguishes $\mathcal{A}(D)$ from $\mathcal{A}(D \cup S)$ with significant advantage over random guessing.

We follow the auditing steps with ClipBKD, a clipping-aware backdoor attack robust to clipping~\cite{jagielski2020auditing}. For CV tasks, we use Fashion-MNIST~\cite{xiao2017fashion} and Purchase-100~\cite{shokri2017membership}; for NLP, SST-2 and QNLI. The ClipBKD test statistic checks if backdoored points are distinguishable by loss falling below a threshold. To set this threshold, we train models on unpoisoned and poisoned datasets differing by $k$ samples. Poisoning points are generated with TextAttack and used to create neighborhood datasets $D_0$ and $D_1$. Due to resource limits, we train 10 models ($T=20$ trials) for threshold learning, then 20 models per dataset–model pair for testing on $S$ to detect whether $S$ was included. We set confidence at 0.01 (i.e., our Monte Carlo estimates hold with 99\% confidence) to compute empirical $\epsilon_{\text{LB}}$, reporting the best result for $k=1$ poisoning point.

Figure~\ref{fig:auditing_nlp} shows the provable $\epsilon$ (x-axis) and empirical $\epsilon$ (y-axis) for DP-SGD and \textsf{PLRV-O} for nlp task. 
It shows that the empirical privacy bound of DP-SGD and \textsf{PLRV-O} is less than the theoretical one. 
And our empirical privacy bound is less than DP-SGD, which means that our work can provide a stricter privacy guarantee than DP-SGD empirically, which also guarantees that the comparison of DP-SGD and PLVR-O is fair. Similar to the fine-tuning of NLP, we discuss the performance of the proposed work on the CV task in Figure~\ref{fig:auditing_cv}. We evaluate the ClipBKD auditing experiments on the Logistic Regression (LR) model and the two-layer Feedforward Neural Network (FNN) with Fashion-MNIST and Purchase-100 datasets following the pipeline in~\cite{jagielski2020auditing}. The pattern we choose for backdoor attacks is a 5x5 white square in the top-left corner of an image, following~\cite{jagielski2020auditing}.
To produce the threshold of backdoor attack, we train 500 models on the unpoisoned dataset and 500 models on the poisoned dataset, and produce the best $\epsilon_{LB}$ with other $1000$ trained models using empirically measuring privacy in~\cite{jagielski2020auditing}. The empirical privacy of Gaussian and \textsf{PLRV-O} are similar, but they are still lower than the theoretical one. Thus, our work has strong potential for extensions to domains with training from scratch, such as CV and NLP tasks.

\subsection{Runtime and Convergence Analysis}
\label{appendix:runtime-conv}

\begin{figure}[!h]
  \centering
  \includegraphics[width=\columnwidth]{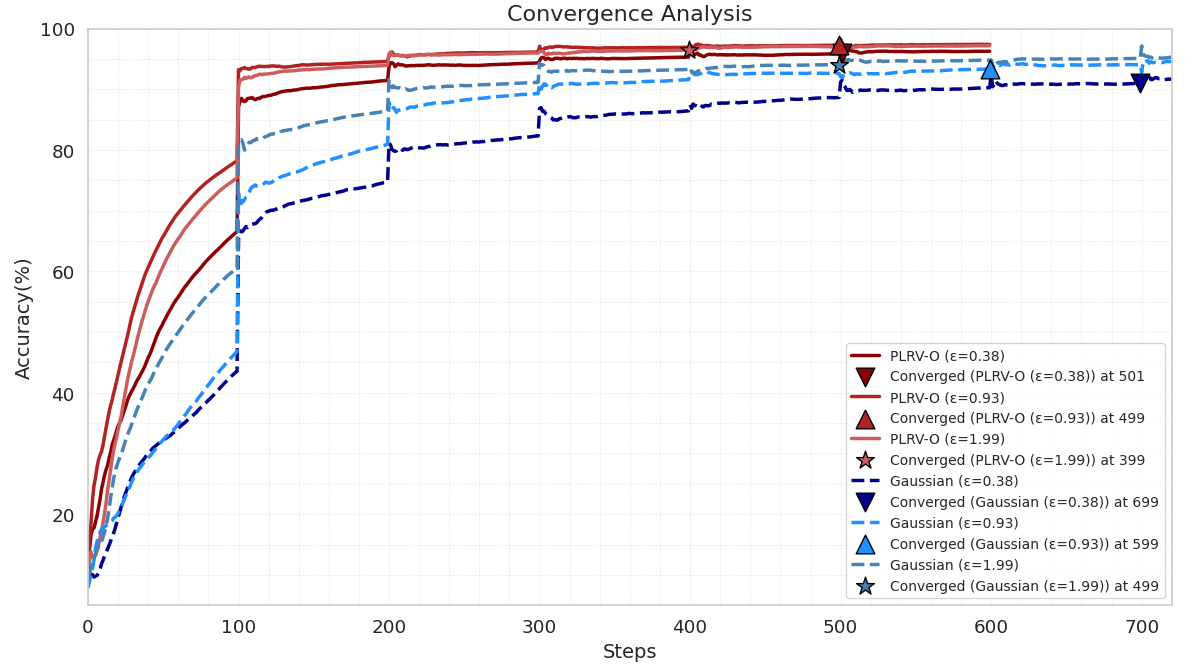}
  \caption{Training accuracy over training steps for \textsf{PLRV-O} and Gaussian mechanisms across different privacy budgets ($\epsilon$). \textsf{PLRV-O} shows faster and more stable convergence compared to the Gaussian at equivalent privacy levels.}
\label{fig:convergence-plot}
\end{figure}

\textsf{PLRV-O} introduces negligible extra overhead. The primary cost is offline, involving a one-time search for the optimal $(k, \theta)$ pair and the precomputation of constants for noise sampling. \textsf{PLRV-O} often enables faster online runtime by supporting much larger clipping thresholds than Gaussian at the same $\epsilon$, as shown in Table~\ref{tab:convergence}.

\noindent\textbf{Convergence Analysis.} We evaluate the effect of noise mechanisms on model convergence by comparing \textsf{PLRV-O} and Gaussian noise across different privacy budgets ($\epsilon$). The convergence step is defined as the first point where the accuracy improvement within the previous 200 epochs falls below 1\%. $\Delta$Acc measures the change in accuracy from 200 steps before convergence to the convergence point. Table~\ref{tab:convergence} summarizes the convergence step, accuracy at convergence, and 
$\Delta$Acc 
(change in accuracy over the final 200 steps) during CNN training on the MNIST dataset.

\begin{table}[!h]
\centering
\vspace{0.15in}
\caption{Convergence comparison between Gaussian and \textsf{PLRV-O} mechanisms at different privacy levels $(\epsilon)$ during CNN training on the MNIST dataset.}
\label{tab:convergence}
\resizebox{\columnwidth}{!}{
\begin{tabular}{|c|c|c|c|c|c|}
\hline
\textbf{Mechanism} & \textbf{$\epsilon$} & \textbf{Con. Step} & \textbf{Acc. at Con.} & \textbf{$\Delta$Acc} & \textbf{Time to Con. (s)} \\
\hline
\multirow{3}{*}{\textsf{PLRV-O}}   & 0.38 & 501 & 95.55\% & 0.69\%  & 57.60\\
& 0.93 & 499 & 96.23\% & 0.14\%  & 55.58\\
& 1.99 & 399 & 96.44\% & 0.61\% & 44.73 \\ \hline
\multirow{3}{*}{Gaussian}  & 0.38 & 699 & 91.00\% & -0.19\%  & 90.44 \\
& 0.93 & 599 & 93.31\% & 0.80\%  & 61.16\\
& 1.99 & 499 & 94.04\% & -0.27\% & 51.35 \\ \hline
\end{tabular}
}\vspace{-0.05in}
\end{table}

As shown in Figure~\ref{fig:convergence-plot}, \textsf{PLRV-O} achieves faster convergence and consistently higher accuracy compared to the Gaussian mechanism at equivalent $\epsilon$ values. This advantage stems from \textsf{PLRV-O}'s ability to tolerate larger clipping bounds under the same privacy budget, accelerating training. Our empirical results show consistent convergence improvements. Furthermore, our \textsf{PLRV-O} mechanism can also be applied to large datasets to achieve a better privacy utility trade-off, such as ImageNet or TinyImageNet.

\vspace{0.05in}

\begin{table}[!h]
\centering
\caption{Boosting DP-FTRL accuracy (\%) on MNIST/CIFAR-10 ($\delta = 1/n$). The non-private accuracy is 98.9\% for MNIST and 78.5\% for CIFAR-10 (the latter from Figure~1b in~\cite{kairouz2021practical}).
}
\label{tab:dp-ftrl-accuracy}
\begin{tabular}{|c|c|c|c|}
\hline
\textbf{Dataset} & $\epsilon$ & \textbf{Gaussian Acc.} & \textbf{\textsf{PLRV-O} Acc.} \\
\hline
\multirow{3}{*}{MNIST}     & 0.05 & 30.34\% & 92.07\% \\
& 0.23 & 95.15\% & 96.86\% \\
& 0.90 & 97.50\% & 98.28\% \\
\hline
\multirow{3}{*}{CIFAR-10}  & 0.05 & 58.18\% & 68.55\% \\
& 0.23 & 76.09\% & 76.12\% \\
& 0.90 & 76.09\% & 76.12\% \\
\hline
\end{tabular}
\end{table}

\subsection{Boosting Orthogonal Methods}
\label{appendix:Dart}

\textsf{PLRV-O} also provides an orthogonal improvement that can enhance the performance of existing methods. To demonstrate its applicability to advanced settings, we apply it to DP-FTRL~\cite{kairouz2021practical} as a representative example. \textsf{PLRV-O} noise consistently outperforms Gaussian noise across different privacy budgets, demonstrating that \textsf{PLRV-O} can boost DP-FTRL in its applications (without subsampling) in Table~\ref{tab:dp-ftrl-accuracy}, e.g., matrix factorization, federated learning~\cite{FengMHYKW025}.

\section{Related Work}
\label{sec:related}

The majority of research on DP-SGD has focused on the Gaussian mechanism due to its smooth noise distribution, which aligns well with gradient updates and facilitates privacy accounting using the moments accountant framework~\cite{abadi2016deep}. 

\vspace{0.05in}

\noindent\textbf{Recent Development in DP-SGD.} Several studies have been proposed for improving DP-SGD along different dimensions. They can be broadly categorized into: (i) memory-efficient methods: GHOST~\cite{li2021large}, PEFT~\cite{yu2021differentially} , and ViP~\cite{yu2023vip} reduce memory overhead by modifying gradient computation or clipping schemes; (ii) time-efficient methods: DP-SGD-JL~\cite{bu2021fast}, Mixed Ghost~\cite{bu2022scalable}, and DP-BiTFiT~\cite{bu2022differentially} focus on reducing runtime while maintaining DP; and (iii) accuracy-enhancing methods: DPSUR~\cite{fu2023dpsur}, DP-Forward~\cite{du2023dp} improve utility by refining gradient updates or leveraging structured noise.

While most of them focus on the Gaussian mechanism, no existing DP-SGD variants have applied Laplace noise for large-scale NLP or vision tasks. Beyond efficiency improvements, prior work has also refined privacy analysis in DP-SGD. Wang et al.~\cite{wang2019subsampled} examined subsampling effects on privacy guarantees, while Gopi et al.~\cite{gopi2021numerical} developed a numerical approach to compute privacy loss precisely. These accounting refinements could be leveraged alongside \textsc{PLRV-O} to achieve tighter privacy budgets and faster convergence.

\vspace{0.05in}

\noindent \textbf{DP-SGD with Laplace Mechanism.} Prior to the widespread adoption of DP-SGD, the Laplace mechanism was considered optimal in many DP applications due to its well-defined theoretical properties and robustness in certain regimes. It remains widely used in pure $\epsilon$-DP settings, offering minimal $\ell_1$ and $\ell_2$ errors under strong privacy constraints~\cite{Dwork06, 6875258}. While Gaussian noise has become the standard in DP-SGD due to its compatibility with the moments accountant, Laplace noise can outperform Gaussian in certain privacy regimes, particularly as $\epsilon \to 0^{+}$, where it achieves lower error~\cite{6875258}. More recent works~\cite{sommer2018privacy} have introduced privacy loss distribution (PLD)-based accounting, enabling tighter $(\epsilon, \delta)$-DP for Laplace, Gaussian, and Discrete Laplace mechanisms, especially under Poisson subsampling.

Despite these advantages, DP-SGD with Laplace noise has seen limited adoption due to the prohibitive nature of $\ell_1$-norm gradient clipping, which often destabilizes training and degrades privacy-utility trade-offs. Unlike Gaussian noise, which integrates well with $\ell_2$-norm clipping, Laplace noise interacts poorly with $\ell_1$-norm constraints, making its deployment in large-scale models challenging. Sommer et al.~\cite{sommer2018privacy} analyzed Laplace DP-SGD through privacy loss random variables (PLRVs), comparing it with Gaussian mechanisms via a numerical approach leveraging the central limit theorem. In contrast, our work provides a structured privacy accounting framework that avoids such approximations.

A separate line of research, including Jang et al.~\cite{jang2024rethinking} and Bernstein et al.~\cite{bernstein2018signsgd}, explores \texttt{DP-signSGD}, where gradients are compressed for efficient communication, with updates determined via majority voting rather than traditional DP-SGD averaging. Since \texttt{DP-signSGD} workers transmit only gradient signs, alternative privacy mechanisms such as the exponential and Laplace mechanisms have been studied in this context. However, as \texttt{DP-signSGD} follows a fundamentally different training paradigm, we consider it independent of our setting. Another application of Laplace noise is in Bayesian deep learning, as explored by Daxberger et al.~\cite{daxberger2021laplace}, where Laplace approximations facilitate Bayesian inference, offering theoretical advantages and practical benefits in uncertainty estimation. This approach, however, is distinct from our work, as it focuses on Bayesian learning rather than private SGD training.

\section{Conclusion}
\label{sec:conclusion}

We introduced \textsf{PLRV-O}, a flexible framework for optimizing differentially private deep learning by leveraging the structure of the privacy loss random variable (PLRV). Unlike conventional mechanisms such as Gaussian or Laplace noise—which are constrained by a single tunable parameter—\textsf{PLRV-O} introduces a multi-parameter design space that better captures the complexity of real-world privacy-utility trade-offs. By explicitly incorporating factors such as model size, training duration, sampling strategy, and clipping thresholds, \textsf{PLRV-O} enables fine-grained control over noise shaping for both training and fine-tuning. Empirical results across CV and NLP tasks demonstrate substantial performance gains under strong privacy constraints, where the proposed method enhances model utility without compromising privacy guarantees. 
Beyond training and fine-tuning, we plan to explore optimizable randomization mechanisms and noise to support broader domains, including machine unlearning~\cite{naderloui2025rectifying}, certified robustness~\cite{HongWH22,ZhangHHHWBR24}, and attacks~\cite{Hong0WBH24} via randomized smoothing.

\section*{Acknowledgments}
 We sincerely thank the anonymous reviewers for their constructive comments. This work is partially supported by the National Science Foundation (NSF) under Grants No. CNS-2302689, CNS-2308730, CNS-2319277, CNS-2432533, ITE-2452747, and ITE-2452749, as well as by a Cisco Research Award. Any opinions, findings, conclusions, or recommendations expressed in this paper are those of the authors and do not necessarily reflect the views of the funding agencies.

\bibliographystyle{ACM-Reference-Format}
\bibliography{bib}

\appendix
\section{Probability Theory Backgrounds}
\label{apndx:Background:prob}
\subsection{Moment Generating Function (MGF)}
\begin{definition} 
\label{defn:MGF}
Let \( f(x) \) be the probability density function (PDF) of a random variable \( X \). The moment generating function of \( X \), if it exists, is defined as:  
\[
\mathcal M_X(t) = \mathbb{E}[e^{tX}] = \int e^{tx} f(x) \, dx.
\]
\end{definition}

\begin{theorem}[Composability of MGF]  
\label{composeMGF}
If \( x_1, \dots, x_n \) are \( n \) independent random variables with existing MGFs \( \mathcal M_{x_i}(t) = \mathbb{E}[e^{t x_i}] \) for \( i = 1, \dots, n \), then the MGF of the linear combination \( Y = \sum_{i=1}^{n} a_i x_i \) is given by:  
\[
\mathcal M_Y(t) = \prod_{i=1}^{n} \mathcal  M_{x_i}(a_i t).
\]
\end{theorem}

\section{Laplace Mechanism Proofs}
\subsection{Proof of Theorem~\ref{thm:lap2privacyloss}}
\label{prooflaploss}

In the following, we prove the bound on the privacy loss of the Laplace mechanism stated in Theorem~\ref{thm:lap2privacyloss}.

\vspace{0.05in}

\begin{proof}
Let \( \bar{\mathbf{g}}(\text{aux},\mathbf{x},b) \) denote the output of a Laplace mechanism with scale \( b \). The privacy loss at an outcome \( o \) is defined as:
\begin{align}
\label{eq:marginal}
    c(o; \text{aux}, d, d', b) 
    &= \log \frac{\Pr[\bar{\mathbf{g}}(\text{aux},\mathbf{x},b) = o]}{\Pr[\bar{\mathbf{g}}(\text{aux},-,b) = o]},
\end{align}
where \( \bar{\mathbf{g}}(\text{aux}, -, b) \) denotes the outcome of the mechanism without access to the input \( \mathbf{x} \) (e.g., $\mathbf{x}$ is only in one data set). 
Thus, without loss of generality, the denominator follows a zero-mean Laplace PDF, as in the worst-case setting it lacks any gradient component in the direction of $\bar{\mathbf{g}}$, while the numerator is centered around \( \bar{\mathbf{g}} \).
\begin{align}
& c(o; \text{aux}, d, d',b)= \log \frac{\text{Lap}(\bar{\mathbf{g}}(\mathbf{x}), b I_n )}{\text{Lap}(0,b I_n)}\nonumber 
  \\
    &= \log \frac{ \left(\frac{1}{2b}\right)^n \exp\left(-\frac{\|o - \bar{\mathbf{g}}(\mathbf{x}) \|_1}{b}\right)  }{ \left(\frac{1}{2b}\right)^n \exp\left(-\frac{\|o\|_1}{b}\right)   }= \frac{\|o \|_1 - \|o - \bar{\mathbf{g}}(\mathbf{x})\|_1}{b} \nonumber.
\end{align}
For all real valued vectors $A$ and $B$, using $\|A\|_1 = \|(A-B)+B\|_1$, by the triangle inequality  we have:  $\|A\|_1 - \|B\|_1 \leq \|A-B\|_1$. Thus, \begin{equation}
\label{Laplacebound}
   \frac{\|o \|_1 - \|o - \bar{\mathbf{g}}(\mathbf{x}) \|_1}{b} \leq \frac{\|\bar{\mathbf{g}}(\mathbf{x})\|_1}{b} 
\end{equation}
Evaluating the  $\ell_1$ norm in terms of the elements of  
$\bar{\mathbf{g}}(\mathbf{x})=[\mathbf{g}_1, \mathbf{g}_2,\cdots, \mathbf{g}_n]$, we have: 
\begin{equation}
\label{Laplacebound}
   c(o; \text{aux}, d, d', b) \leq \frac{\sum_{j=0}^n |\mathbf{g}_j|}{b}. 
\end{equation}
\end{proof}

\subsection{Proof of Theorem~\ref{MAflaplace}}
\label{subsampledLap2}
In the following, we prove a tight bound on the moments accountant function of univariate Laplace mechanisms, as stated in Theorem~\ref{MAflaplace}.

\vspace{0.05in}

\begin{proof}
Consider two adjacent data sets $d$ and $d'$.  
Without loss of generality, suppose $d'$ has an extra training sample.
Let $\zeta$ denote a fixed sampling rate.  
Consider any sampling realization over $d\cup d' = d'$ using iid sampling with per-element selection of $\zeta$.

With probability $1-\zeta$, the extra sample in $d'$ will not be included and thus query values over the sub-sampled datasets will be identical.  Let $\mu_0$ denote the resulting density of the Laplace mechanism in this case.  Let $q(d,\zeta)$ denote the mean of $\mu_0$.  By construction $|q(d,\zeta)|\leq C$.

With probability $\zeta$, the extra sample in $d'$ will be kept resulting in different query values between the two data sets.  Let $\mu_1$ denote the resulting density of the Laplace mechanism of the query over the sub-sampled $d'$.  Let $q(d',\zeta)$ denote the mean of $\mu_1$.  By construction $|q(d',\zeta)|\leq C$.

Thus, we can identify the mechanism over $d'$ as having a mixture distribution,
\[
M_q(d,\zeta) \sim \mu_0, \quad M_q(d',\zeta) \sim \mu \triangleq (1 - \zeta)\mu_0 + \zeta \mu_1.
\]

For any $\lambda$, we aim to show:
\[
A=\mathbb{E}_{z \sim \mu}
\left[\left(\frac{\mu(z)}{\mu_0(z)}\right)^\lambda\right] \leq \alpha,  
\
\text{and} 
\
B= \mathbb{E}_{z \sim \mu_0}
\left[\left(\frac{\mu_0(z)}{\mu(z)}\right)^\lambda\right] \leq \alpha.
\]
for some explicit \(\alpha\) to be determined later.

Multiplying by $\mu_0(z)/\mu_0(z)$ and rearranging, we can express: 
\[
A
= \mathbb{E}_{z \sim \mu_0}\left[\left(\frac{\mu(z)}{\mu_0(z)}\right)^{\lambda+1}\right]
=\mathbb{E}_{z \sim \mu_0}\left[\left(1 - \zeta+\zeta\frac{ \mu_1(z)}{\mu_0(z)}\right)^{\lambda+1}\right],\]
and
\begin{align}
B&=\mathbb{E}_{z \sim \mu_0}\left[\left(
\frac{\mu_0(z)}{(1 - \zeta)\mu_0(z) + \zeta \mu_1(z)}
\right)^\lambda\right] \nonumber\\
&=\mathbb{E}_{z \sim \mu_0}\left[\left(
\frac{1}{1 - \zeta +  \zeta\frac{\mu_1(z)}{\mu_0(z)}}
\right)^\lambda\right].
\end{align}

Mironov et al.~\cite{mironov2019r} (Section 3.1) demonstrate that $A \geq B$ holds in general for centrally symmetric distributions.

We thus focus on analyzing $A$.  We start by  applying the binomial theorem and the linearity of expectation,
\begin{align*}
    A &=\mathbb{E}_{z \sim \mu_0}\left[\left(1 - \zeta+\zeta\frac{ \mu_1(z)}{\mu_0(z)}\right)^{\lambda+1}\right] \\
    &= \mathbb{E}_{z \sim \mu_0}\left[ \sum_{\eta =0}^{\lambda+1} \binom{\lambda+1}{\eta} (1-\zeta)^{\lambda+1-\eta} \zeta^\eta \left(\frac{ \mu_1(z)}{\mu_0(z)} \right)^\eta \right] \\
    &=  \sum_{\eta =0}^{\lambda+1} \binom{\lambda+1}{\eta} (1-\zeta)^{\lambda+1-\eta} \zeta^\eta \mathbb{E}_{z \sim \mu_0}
    \left[\left(\frac{ \mu_1(z)}{\mu_0(z)} \right)^\eta \right] .
\end{align*}

We can simplify the ratio of the densities as
\begin{align*}
    \frac{\mu_1(z)}{\mu_0(z)} 
    &= \frac{\frac{1}{2b} \exp( - \frac{|z - q(d',\zeta) |}{b}  )}%
    {\frac{1}{2b} \exp( - \frac{|z - q(d,\zeta) |}{b}  )} \\
     &= \exp \left(  \frac{1}{b} \left( |z - q(d,\zeta) | - |z - q(d',\zeta) |\right) \right).
\end{align*}

Without loss of generality, consider that $q(d,\zeta) < q(d', \zeta).$    We split up the real line into three intervals: $\mathbb{R} = (-\infty, q(d,\zeta)] \cup (q(d,\zeta), q(d',\zeta)) \cup [q(d',\zeta), \infty)$. 

We evaluate the expectation $\mathbb{E}_{z \sim \mu_0} \left[\left(\frac{ \mu_1(z)}{\mu_0(z)} \right)^\eta \right] $ over these three intervals separately (conditionally) and then will combine them afterwards (with probabilities of respective events occuring).  We will analyze the middle interval last.

\textbf{Case A1: $z<q(d,\zeta)$:}

Since $z<q(d,\zeta)$, $|z - q(d,\zeta) | = q(d,\zeta) - z$.  
Likewise, $z<q(d',\zeta)$, so  $|z - q(d',\zeta) | = q(d',\zeta) - z$.
The likelihood ratio simplifies
\begin{align*}
    \frac{\mu_1(z)}{\mu_0(z)} 
    &=\exp \left(  \frac{1}{b} \left( |z - q(d,\zeta) | - |z - q(d',\zeta) |\right) \right) \\
    &=\exp \left(  \frac{1}{b} \left( q(d,\zeta) - z - (q(d',\zeta) - z)\right) \right) \\
    &=\exp \left(  \frac{1}{b} \left( q(d,\zeta) - q(d',\zeta) )\right) \right). \\
\end{align*}
We observe that there is no dependence on $z$, so this ratio becomes a constant in the expectation.  
Under $\mu_0$, $z$ is equally distributed about $q(d,\zeta)$.  So the probability of this event is $1/2$.  Thus, the contribution to the total expectation is 
\[
\frac{1}{2} \exp \left(  \frac{\eta}{b} \left( q(d,\zeta) - q(d',\zeta) )\right) \right).
\]

\textbf{Case A2: $z>q(d',\zeta)$:}
Since $z>q(d',\zeta)$, $|z - q(d',\zeta) | = z - q(d',\zeta)$.  
Likewise, $z>q(d,\zeta)$, so  $|z - q(d,\zeta) | = z - q(d,\zeta)$.
The likelihood ratio simplifies
\begin{align*}
    \frac{\mu_1(z)}{\mu_0(z)} 
    &=\exp \left(  \frac{1}{b} \left( |z - q(d,\zeta) | - |z - q(d',\zeta) |\right) \right) \\
    &=\exp \left(  \frac{1}{b} \left( z - q(d,\zeta) - (z - q(d',\zeta))\right) \right) \\
    &=\exp \left(  \frac{1}{b} \left( q(d',\zeta) - q(d,\zeta) )\right) \right). \\
\end{align*}
Again, the ratio is a constant with respect to $z$.  The probability of the event is more complicated to analyze than the previous event.
\begin{align*}
    &\hspace{-1cm}\int_{q(d',\zeta)}^\infty \frac{1}{2b} \exp( - \frac{|z-q(d,\zeta)|}{b} ) db\\
    &=\int_{q(d',\zeta)}^\infty \frac{1}{2b} \exp( - \frac{z-q(d,\zeta)}{b} ) db\\
    &=\frac{1}{2b} \exp(  \frac{q(d,\zeta)}{b}) \int_{q(d',\zeta)}^\infty  \exp( - \frac{z}{b} ) db\\
    &=\frac{1}{2b} \exp(  \frac{q(d,\zeta)}{b})  \frac{1}{-\frac{1}{b}} \exp( - \frac{z}{b} ) \big|_{q(d',\zeta)}^{\infty} \\
    &=\frac{1}{2} \exp(  \frac{q(d,\zeta)}{b})   \exp( - \frac{q(d',\zeta)}{b} )  \\    
    &=\frac{1}{2} \exp(  \frac{q(d,\zeta) -  q(d',\zeta)}{b})     
\end{align*}

Thus, the contribution to the total expectation is
\begin{align*}
&\hspace{-.2cm} \frac{1}{2} \exp(  \frac{q(d,\zeta) -  q(d',\zeta)}{b})\exp \left(  \frac{\eta}{b} \left( q(d',\zeta) - q(d,\zeta) )\right) \right) \\
&= \frac{1}{2} \exp \left(  \frac{\eta - 1}{b} \left( q(d',\zeta) - q(d,\zeta) )\right) \right).
\end{align*}

\textbf{Case A3: $q(d,\zeta)\leq z\leq q(d',\zeta)$:}

Since $z>q(d,\zeta)$, $|z - q(d,\zeta) | = z - q(d,\zeta)$.  
Since $z<q(d',\zeta)$, so  $|z - q(d',\zeta) | = q(d',\zeta) - z$.
The likelihood ratio simplifies
\begin{align*}
    \frac{\mu_1(z)}{\mu_0(z)} 
    &=\exp \left(  \frac{1}{b} \left( |z - q(d,\zeta) | - |z - q(d',\zeta) |\right) \right) \\
    &=\exp \left(  \frac{1}{b} \left( z - q(d,\zeta) - (q(d',\zeta) - z)\right) \right) \\
    &=\exp \left(  \frac{1}{b} \left( - q(d,\zeta) - q(d',\zeta) +2z )\right) \right).
\end{align*}

Plugging this back into the expectation,
\begin{align*}
\mathbb{E}_{z \sim \mu_0}     \left[\left(\frac{ \mu_1(z)}{\mu_0(z)} \right)^\eta \right] 
&= \mathbb{E}_{z \sim \mu_0}     \left[\exp \left(  \frac{\eta}{b} \left( - q(d,\zeta) - q(d',\zeta) +2z )\right) \right) \right]\\
&=\exp \left(  \frac{\eta}{b} \left( -q(d,\zeta) - q(d',\zeta) \right) \right)  \\
&\qquad \times \mathbb{E}_{z \sim \mu_0}\left[\exp \left(  \frac{2\eta z}{b}  \right) \right]. 
\end{align*}

Evaluating the inner expectation (only over the interval here),
\begin{align*}
    &\hspace{-.2cm}
    \mathbb{E}_{z \sim \mu_0}\left[\exp \left(  \frac{2\eta z}{b}  \right) \right]  \\
    &=\int^{q(d',\zeta)}_{q(d,\zeta)} \frac{1}{2b} \exp( - \frac{|z-q(d,\zeta)|}{b} + \frac{2\eta z}{b}) db\\
    &=\frac{1}{2b} \exp(  \frac{q(d,\zeta)}{b})
    \int^{q(d',\zeta)}_{q(d,\zeta)} \exp( z(-\frac{1-2\eta}{b} ) ) db\\
    &=\frac{1}{2b} \exp(  \frac{q(d,\zeta)}{b})
    \frac{b}{2\eta-1} \Big[ \exp( q(d',\zeta)(-\frac{1-2\eta}{b} ) \\
    &\qquad - \exp( q(d,\zeta)(-\frac{1-2\eta}{b} ) \Big]     \\
    &= \frac{1}{2(2\eta-1)} \Big[ \exp\Big( q(d',\zeta)\frac{2\eta - 1}{b}  + q(d,\zeta)\frac{1 }{b} \Big) \\
    &\qquad - \exp\Big( q(d,\zeta)\frac{2\eta }{b}  \Big)\Big] .   
\end{align*}

Thus, the contribution to the expectation from this case is 
\begin{align*}
\mathbb{E}_{z \sim \mu_0}     \left[\left(\frac{ \mu_1(z)}{\mu_0(z)} \right)^\eta \right] 
&=\exp \left(  \frac{\eta}{b} \left( -q(d,\zeta) - q(d',\zeta) \right) \right)  \\
&\quad \times \frac{1}{2(2\eta-1)} \Big[ \exp\Big( q(d',\zeta)\frac{2\eta - 1}{b}  + q(d,\zeta)\frac{1 }{b} \Big) \\
&\qquad - \exp\Big( q(d,\zeta)\frac{2\eta }{b}  \Big)\Big] . 
\end{align*}

\textbf{Combining Cases A1-A3:}
Combining the results, we have that
\begin{align*} 
A
&= \mathbb{E}_{z \sim \mu_0}\left[\left(\frac{\mu(z)}{\mu_0(z)}\right)^{\lambda+1}\right]\\
&=  \frac{1}{2} \exp \left(  \frac{\eta}{b} \left( q(d,\zeta) - q(d',\zeta) )\right) \right)\\
&\quad +\frac{1}{2} \exp \left(  \frac{\eta - 1}{b} \left( q(d',\zeta) - q(d,\zeta) )\right) \right)\\
&\quad + \exp \left(  \frac{\eta}{b} \left( -q(d,\zeta) - q(d',\zeta) \right) \right)  \\
&\quad \times \frac{1}{2(2\eta-1)} \Big[ \exp\Big( q(d',\zeta)\frac{2\eta - 1}{b}  + q(d,\zeta)\frac{1 }{b} \Big) \\
&\qquad - \exp\Big( q(d,\zeta)\frac{2\eta }{b}  \Big)\Big]
\end{align*}

Recall that the query values over the sub-sampled data sets $q(d,\zeta)$ and $q(d',\zeta)$ are averaged over the queries (gradients) of the included samples, so the effect of a single sample is smaller the more samples are included. 
For simplicity, by inspection of the formula, we consider a worst case bound using  $q(d)=0$ and $q(d',\zeta) = C$.

\begin{align*} 
A
&=  \frac{1}{2} \exp \left(  \frac{- \eta C}{b}  \right)\\
&\quad +\frac{1}{2} \exp \left(  \frac{(\eta - 1)C}{b}  \right)\\
&\quad + \exp \left(  \frac{-\eta C}{b}  \right) \frac{1}{2(2\eta-1)} \Big[ \exp\Big( \frac{(\eta - 1)C}{b} \Big)
- \exp( 0  )\Big] \\
&=  \exp \left(  \frac{- \eta C}{b}  \right) \left[ \frac{1}{2} - \frac{1}{2(2\eta-1)}\right]
 \\  
&\quad +\exp\Big( \frac{(\eta - 1)C}{b} \Big) \left[ \frac{1}{2} + \frac{1}{2(2\eta-1)}\right] \\
\end{align*}
\begin{align}
\label{equationA}
&=  \exp \left(  \frac{- \eta C}{b}  \right) \left[  \frac{\eta-1}{2\eta-1}  \right]+\exp\Big( \frac{(\eta - 1)C}{b} \Big) \left[ \frac{\eta}{2\eta-1}\right]. 
\end{align}

\textbf{Analysis for B:} Following the argument in Proof~\ref{subsampledLap2} (Theorem 3.2), and using binomial expansion with term-wise comparison, we find that \( B \geq A \), consistent with the result of Mironov et al.

\end{proof}

\subsection{Proof of Theorem~\ref{schurmaf}}
\label{mafschur}
In the following, we prove Theorem~\ref{schurmaf}, that the moments accounting function of the univariate Laplace mechanism is Schur-convex. 
We first prove the following technical lemma, involving second derivatives of the MAF, before continuing on to the main proof.

\begin{lemma}
\label{cor:maf-convex}
The second derivative of the moments accountant function $\alpha(\lambda)$ {in Theorem~\ref{MAflaplace}} 
with respect to the marginal clipped gradients $|\mathbf{g}_i|$ is non-negative.
\end{lemma}

\begin{proof}
Let $a_\eta = \binom{\lambda+1}{\eta} (1-\zeta)^{\lambda+1-\eta} \zeta^\eta$ be the positive weight associated with each $\eta$. The second derivative of $\alpha(\lambda)$ can be written as
\[
\frac{d^2 \alpha(\lambda)}{d |\mathbf{g}_i|^2} = \frac{
\left( \sum_{\eta} a_\eta F(|\mathbf{g}_i|,\eta) \right)
\left( \sum_{\eta} a_\eta \frac{d^2F}{d|\mathbf{g}_i|^2} \right)
-
\left( \sum_{\eta} a_\eta \frac{dF}{d|\mathbf{g}_i|} \right)^2
}{
\left( \sum_{\eta} a_\eta F(|\mathbf{g}_i|,\eta) \right)^2
}.
\]
Recall:
\begin{align*}
F(|\mathbf{g}_i|, \eta) &= \frac{\eta}{2\eta-1} e^{\frac{(\eta-1)|\mathbf{g}_i|}{b}} + \frac{\eta-1}{2\eta-1} e^{-\frac{\eta|\mathbf{g}_i|}{b}}, \\
\frac{dF}{d|\mathbf{g}_i|} &= \frac{\eta(\eta-1)}{b(2\eta-1)} \left( e^{\frac{(\eta-1)|\mathbf{g}_i|}{b}} - e^{-\frac{\eta|\mathbf{g}_i|}{b}} \right), \\
\frac{d^2F}{d|\mathbf{g}_i|^2} &= \frac{\eta(\eta-1)}{b^2(2\eta-1)} \left( (\eta-1) e^{\frac{(\eta-1)|\mathbf{g}_i|}{b}} + \eta e^{-\frac{\eta|\mathbf{g}_i|}{b}} \right).
\end{align*}
Expand \( F \times F'' \):
\begin{align*}
F \times F'' &= \left( \frac{\eta}{2\eta-1} e^{\frac{(\eta-1)|\mathbf{g}_i|}{b}} + \frac{\eta-1}{2\eta-1} e^{-\frac{\eta|\mathbf{g}_i|}{b}} \right)
\\
&\quad \times \frac{\eta(\eta-1)}{b^2(2\eta-1)} \left( (\eta-1) e^{\frac{(\eta-1)|\mathbf{g}_i|}{b}} + \eta e^{-\frac{\eta|\mathbf{g}_i|}{b}} \right)
\\
&= \frac{\eta(\eta-1)}{b^2(2\eta-1)^2} \Bigg[
\eta(\eta-1) e^{2\frac{(\eta-1)|\mathbf{g}_i|}{b}}
+ \eta^2 e^{\frac{(\eta-1)|\mathbf{g}_i|/b} - \eta|\mathbf{g}_i|/b}
\\
&\quad + (\eta-1)^2 e^{-\frac{\eta|\mathbf{g}_i|}{b} + \frac{(\eta-1)|\mathbf{g}_i|}{b}}
+ \eta(\eta-1) e^{-2\frac{\eta|\mathbf{g}_i|}{b}}
\Bigg].
\end{align*}

Notice:
\[
e^{\frac{(\eta-1)|\mathbf{g}_i|}{b}} e^{-\frac{\eta|\mathbf{g}_i|}{b}} = e^{-\frac{|\mathbf{g}_i|}{b}},
\quad
e^{-\frac{\eta|\mathbf{g}_i|}{b}} e^{\frac{(\eta-1)|\mathbf{g}_i|}{b}} = e^{-\frac{|\mathbf{g}_i|}{b}}.
\]

Thus:
\begin{eqnarray*}
F \times F'' &=& \frac{\eta(\eta-1)}{b^2(2\eta-1)^2} \times \\ &&\Bigg[
\eta(\eta-1) e^{2\frac{(\eta-1)|\mathbf{g}_i|}{b}}
+ (\eta^2 + (\eta-1)^2) e^{-\frac{|\mathbf{g}_i|}{b}}
+ \eta(\eta-1) e^{-2\frac{\eta|\mathbf{g}_i|}{b}}
\Bigg].
\end{eqnarray*}

Expand \( \left( \frac{dF}{d|\mathbf{g}_i|} \right)^2 \):
\begin{align*}
\left( \frac{dF}{d|\mathbf{g}_i|} \right)^2
&= \left( \frac{\eta(\eta-1)}{b(2\eta-1)} \right)^2 \left( e^{\frac{(\eta-1)|\mathbf{g}_i|}{b}} - e^{-\frac{\eta|\mathbf{g}_i|}{b}} \right)^2
\\
&= \left( \frac{\eta(\eta-1)}{b(2\eta-1)} \right)^2 \left(
e^{2\frac{(\eta-1)|\mathbf{g}_i|}{b}}
- 2 e^{-\frac{|\mathbf{g}_i|}{b}}
+ e^{-2\frac{\eta|\mathbf{g}_i|}{b}}
\right).
\end{align*}
Comparing individual terms pointwise shows that $F \times F'' \geq (F')^2$ holds for all $|\mathbf{g}_i| \geq 0$ and $\eta \geq 0$. Applying the Cauchy--Schwarz inequality to the positive sequence $\{ \sqrt{a_\eta} \sqrt{F(|\mathbf{g}_i|,\eta)} \}$ and $\{ \sqrt{a_\eta} \sqrt{d^2F/d|\mathbf{g}_i|^2} \}$ yields
\[
\left( \sum_{\eta} a_\eta \sqrt{F(|\mathbf{g}_i|,\eta)} \sqrt{\frac{d^2F}{d|\mathbf{g}_i|^2}} \right)^2
\leq
\left( \sum_{\eta} a_\eta F(|\mathbf{g}_i|,\eta) \right)
\left( \sum_{\eta} a_\eta \frac{d^2F}{d|\mathbf{g}_i|^2} \right).
\]
but since $F \times F'' \geq (F')^2$, pointwise, we have $\sqrt{F(|\mathbf{g}_i|,\eta)} \sqrt{\frac{d^2F}{d|\mathbf{g}_i|^2}} \geq \frac{dF}{d|\mathbf{g}_i|}$. Hence, 

\[
\left( \sum_{\eta} a_\eta \frac{dF}{d|\mathbf{g}_i|} \right)
\leq
\left( \sum_{\eta} a_\eta \sqrt{F(|\mathbf{g}_i|,\eta)} \sqrt{\frac{d^2F}{d|\mathbf{g}_i|^2}}\right).
\]
The last two inequalities together yield:

\[
\left( \sum_{\eta} a_\eta \frac{dF}{d|\mathbf{g}_i|} \right)^2
\leq
\left( \sum_{\eta} a_\eta F(|\mathbf{g}_i|,\eta) \right)
\left( \sum_{\eta} a_\eta \frac{d^2F}{d|\mathbf{g}_i|^2} \right),
\]
which shows the numerator is non-negative. Therefore,
\[
\frac{d^2 \alpha(\lambda)}{d |\mathbf{g}_i|^2} \geq 0,
\]
and $\alpha(\lambda)$ is Schur-convex.
\hfill \qedsymbol
\end{proof}

\noindent
We are now ready to prove Theorem~\ref{schurmaf}.

\vspace{0.05in}

\begin{proof}
We apply Schur's condition (also known as the Schur–Strowski criterion) to prove that $\alpha(\lambda)$ is Schur-convex. Recall that a symmetric function \( f(x_1, \dots, x_n) \) is Schur-convex if and only if for all \( i \neq j \),
\[
(x_i - x_j) \left( \frac{\partial f}{\partial x_i} - \frac{\partial f}{\partial x_j} \right) \geq 0.
\]

In our case, the function is $\alpha(\lambda) = \sum_{i=1}^n \alpha_{\mathbf{\bar{g}}_i}(\lambda)$, where $\mathbf{\bar{g}}_i$ is the noisy version of the $\ell_2$-clipped marginal gradients. Denote by $\mathbf{g}_i$, $i \in [n]$, marginal gradients after $\ell_2$ clipping and before the addition of DP noise. Then, with Theorem~\ref{MAflaplace}, the univariate moments accountant for the $i$-th coordinate satisfies
\begin{align}
\label{lap_2account11}
\alpha_{\mathbf{\bar{g}}_i}(\lambda) \leq \log \left[
\sum_{\eta = 0}^{\lambda + 1} \binom{\lambda + 1}{\eta} (1 - \zeta)^{\lambda + 1 - \eta} \zeta^\eta F(|\mathbf{g}_i|, \eta)
\right],
\end{align}
where the function \( F(|\mathbf{g}_i|, \eta) \) is defined as
\begin{align}
\label{eqn:G}
F(|\mathbf{g}_i|, \eta) = 
\frac{e^{\frac{(\eta - 1)|\mathbf{g}_i|}{b}}}{2} + 
\frac{e^{-\frac{\eta |\mathbf{g}_i|}{b}}}{2} +
\frac{e^{\frac{(\eta - 1) |\mathbf{g}_i|}{b}} - e^{\frac{-\eta |\mathbf{g}_i|}{b}}}{2(2\eta - 1)}.
\end{align}

Define the term inside the square brackets in ~\eqref{lap_2account11} as \( X \). Then, the derivative of \( \alpha_{\mathbf{\bar{g}}_i}(\lambda) \) with respect to \( |\mathbf{g}_i| \) satisfies
\begin{equation}
\label{alphaderiv}
\frac{d\alpha_{\mathbf{\bar{g}}_i}(\lambda)}{d|\mathbf{g}_i|} = \frac{\sum_{\eta = 0}^{\lambda + 1} \binom{\lambda + 1}{\eta} (1 - \zeta)^{\lambda + 1 - \eta} \zeta^\eta \frac{dF}{d|\mathbf{g}_i|}}{X}.
\end{equation}

Lets compute the derivative $\frac{dF}{d|\mathbf{g}_i|}$:
\begin{eqnarray}
\frac{dF}{d|\mathbf{g}_i|} &=& \frac{(\eta-1)}{2b} e^{\frac{(\eta-1)|\mathbf{g}_i|}{b}}
- \frac{\eta}{2b} e^{-\frac{\eta |\mathbf{g}_i|}{b}} \nonumber \\
&+& \frac{1}{2(2\eta-1)} \left( \frac{(\eta-1)}{b} e^{\frac{(\eta-1)|\mathbf{g}_i|}{b}} + \frac{\eta}{b} e^{-\frac{\eta |\mathbf{g}_i|}{b}} \right).
\end{eqnarray}

Special cases:

- For $\eta = 0$, the terms cancel symmetrically and $\frac{dF}{d|\mathbf{g}_i|} = 0$.

- For $\eta = 1$, the derivative simplifies to $0$ by symmetry.

- For $\eta > 1$, expanding and grouping terms, we have
\begin{align}
\label{eqn:finalschur}
\frac{dF}{d|\mathbf{g}_i|} = \frac{(\eta-1)\eta}{b(2\eta-1)} \left( e^{\frac{(\eta-1)|\mathbf{g}_i|}{b}} - e^{-\frac{\eta|\mathbf{g}_i|}{b}} \right).
\end{align}
Since $\eta > 1$, the prefactor $\frac{(\eta-1)\eta}{b(2\eta-1)}$ is positive. Moreover, for any $|\mathbf{g}_i| \geq 0$, $e^{\frac{(\eta-1)|\mathbf{g}_i|}{b}} \geq e^{-\frac{\eta|\mathbf{g}_i|}{b}}$. Thus, $\frac{dF}{d|\mathbf{g}_i|} \geq 0$ for all $\mathbf{g}_i$ and all $\eta$. 

As stated earlier,
\[
\frac{d\alpha_{\mathbf{\bar{g}}_i}(\lambda)}{d|\mathbf{g}_i|} = \frac{\sum_{\eta = 0}^{\lambda + 1} \binom{\lambda + 1}{\eta} (1 - \zeta)^{\lambda + 1 - \eta} \zeta^\eta \frac{dF}{d|\mathbf{g}_i|}}{X}.
\]
Since each $\frac{dF}{d|\mathbf{g}_i|} \geq 0$, $X>0$, and all the coefficients are positive, it follows that
\[
\frac{\partial \alpha(\lambda)}{\partial |\mathbf{g}_i|} \geq 0.
\]
Since $\alpha(\lambda) = \sum_{i=1}^n \alpha_{\mathbf{\bar{g}}i}(\lambda)$ with each $\alpha_{\mathbf{\bar{g}}_i}(\lambda)$ depending only on $|\mathbf{g}_i|$, we have $\frac{\partial \alpha(\lambda)}{\partial |\mathbf{g}i|} = \frac{d \alpha_{\mathbf{\bar{g}}_i}(\lambda)}{d |\mathbf{g}_i|}$. Thus, the overall moments accountant function (MAF) is non-decreasing. To satisfy the Schur–Ostrowski criterion, it suffices to show that the second derivative of $\alpha(\lambda)$ is non-negative (MAF is convex). A positive second derivative ensures that for any $i \neq j$, both $|\mathbf{g}_i| - |\mathbf{g}_j|$ and $\frac{\partial \alpha(\lambda)}{\partial |\mathbf{g}_i|} - \frac{\partial \alpha(\lambda)}{\partial |\mathbf{g}_j|}$ share the same sign, thereby satisfying the criterion. 
In Section~\ref{cor:maf-convex} we proved that the second derivatives are non-negative, concluding the proof for Theorem~\ref{schurmaf}.  \hfill $\square$
\end{proof}

An \textit{alternative} approach is to prove the Schur-convexity of the univariate MAF and apply the following results.

\begin{proposition}[Schur\cite{schur1923uber}; Hardy–Littlewood–Pólya~\cite{hardy1929some}]
\label{summajor}
Let \( I \subset \mathbb{R} \) be an interval and \( g : I \to \mathbb{R} \) be a convex function. Then the function
\[
\varphi(x) = \sum_{i=1}^n g(x_i)
\]
is Schur-convex on \( I^n \). Consequently, if \( x \prec y \) on \( I^n \), then
\[
\varphi(x) \leq \varphi(y).
\]
\end{proposition}

\subsection{Proof of Lemma~\ref{lem:majorset}}
\label{majorseet}

In the following, we prove Lemma~\ref{lem:majorset} 
on the majorization set for $\ell_2$ clipped gradients.

\vspace{0.05in}

\begin{proof}
Since the $\ell_2$ clipping ensures that $\|\mathbf{g}\|_2 \leq C$, we have $|\mathbf{g}_1| \leq C = x_1$. Applying the AM–QM inequality to the first $i$ marginal clipped gradients, we have
\[
\frac{1}{i} \sum_{j=1}^i |\mathbf{g}_j| \leq \sqrt{\frac{1}{i} \sum_{j=1}^i (|\mathbf{g}_j|)^2} \leq \frac{C}{\sqrt{i}},
\]
which implies
\[
\sum_{j=1}^i |\mathbf{g}_j| \leq C \sqrt{i}.
\]
On the other hand, the majorization set $x$ satisfies
\[
\sum_{j=1}^i x_j = C \sqrt{i},
\]
since
\[
\sum_{j=1}^i x_j = C (\sqrt{i} - \sqrt{0})
\]
by telescoping the differences.

Thus, for each $i = 1, \dots, n$, we have
\[
\sum_{j=1}^i |\mathbf{g}_j| \leq \sum_{j=1}^i x_j,
\]
establishing that $|\mathbf{G}| \prec_w x$.
\end{proof}

\section{PLRV Mechanism Proofs}
\subsection{Proof of Theorem~\ref{the:subsampledMAF2232}}
\label{proofplrvuni}

In the following, we prove a tight bound on the moments accountant function of univariate PLRV mechanisms, as stated in Theorem~\ref{the:subsampledMAF2232}.

\vspace{0.05in}

\begin{proof}Let $J\in[n]$ denote the indices of elements randomly selected through mini-batch sub-sampling, where elements are included i.i.d. with sub-sampling probability $\zeta=L/N$. We express the distribution of $M_{\bf{g}}(d',f)$ as a double-fold mixture. 

Let $\mu_0(b)$ denote a PLRV  distribution centered at $\bf{g}$, with scale distributed with $f(b)$. Thus $M_{\bf{g}}(d,f) \sim \mu_0$.

Let $\mu_1(b)$ denote a PLRV distribution centered at the value of $\bar{g}$ conditioned on the event $n\in J$.  Then the distribution of $M_{\bf{g}}(d',f)$ is  the mixture $\mu_\Delta = (1 - \zeta)\mu_0(b) + \zeta \mu_1(b)$. 

The MAF of the PLRV mechanism with a pre-specified (but arbitrary) pair $\{d,d'\}$ and using worst case sensitivity (to bound over worst case $\mathrm{aux}$), can be expressed as
\begin{align}
    \alpha_{M_{\bar{g}}}(\lambda)
  &=
  \max \Bigl( \alpha_{M_{\bar{g}}} (\lambda; d,f), \alpha_{M_{\bar{g}}}(\lambda; d',f) \Bigr) \\
  &= \log \max \left\{
  \mathbb{E}_{z \sim \mu_0}\left[\left(\frac{\mu_0(z)}{\mu_\Delta(z)}\right)^\lambda\right]
  ,\,
  \mathbb{E}_{z \sim \mu_\Delta}\left[\left(\frac{\mu_\Delta(z)}{\mu_0(z)}\right)^\lambda\right]   \right\},
\end{align}
where we use that by monotonicity of logarithms,  
\[\max  \{\log a_1, \log a_2\} = \log \max \{a_1, a_2\}.\]

Similar to the argument in Proof~\ref{subsampledLap2} (of Theorem 3.2), we multiply by \( \mu_0(z)/\mu_0(z) \) and rearrange to obtain:

\[
A
= \mathbb{E}_{z \sim \mu_0(b)}\left[\left(\frac{\mu(z;b)}{\mu_0(z;b)}\right)^{\lambda+1}\right]
=\mathbb{E}_{z \sim \mu_0(b)}\left[\left(1 - \zeta+\zeta\frac{ \mu_1(z;b)}{\mu_0(z;b)}\right)^{\lambda+1}\right],\]
and
\[
B=\mathbb{E}_{z \sim \mu_0(b)}\left[\left(
\frac{\mu_0(z;b))}{(1 - \zeta)\mu_0(z;b)) + \zeta \mu_1(z;b))}
\right)^\lambda\right]\]
\[=\mathbb{E}_{z \sim \mu_0(b)}\left[\left(
\frac{1}{1 - \zeta +  \zeta\frac{\mu_1(z;b))}{\mu_0(z;b))}}
\right)^\lambda\right].
\]
\begin{remark}
The expressions for \( A \) and \( B \) are conditional (on \( b \)) variants of those in Proof~\ref{subsampledLap2}. By the law of total expectation, the bound on \( A \) becomes \( \mathbb{E}_b[A(z \mid b)] \), which simplifies to:
\begin{align}
\label{equationA}
A = \mathbb{E}_b\left[
\exp\left( \frac{-\eta C}{b} \right) \cdot \frac{\eta - 1}{2\eta - 1}
+ \exp\left( \frac{(\eta - 1)C}{b} \right) \cdot \frac{\eta}{2\eta - 1}
\right].
\end{align}

This yields the bound:
\[
A =
\frac{\eta}{2\eta - 1} \mathcal{M}_u\big((\eta - 1)C\big) + \frac{\eta - 1}{2\eta - 1} \mathcal{M}_u\big(-\eta C\big),
\]
where \( u = 1/b \) is the reciprocal of the PLRV scale.

Thus, the moments accounting function satisfies:
\[
\alpha_{M_{\mathbf{g}}}(\lambda) = \log \left\{
\sum_{\eta = 0}^{\lambda + 1} \binom{\lambda + 1}{\eta} (1 - \zeta)^{\lambda + 1 - \eta} \zeta^\eta \mathcal{G}(C, \eta)
\right\}.
\]
\end{remark}

We proceed to formally prove the above result. Mironov et al.~\cite{mironov2019r} (Section 3.1) demonstrate that $A \geq B$ holds in general for centrally symmetric distributions.

We begin by applying the binomial theorem,
\begin{align*}
    B&= \mathbb{E}_{z \sim \mu_0}\left[\left(1-\zeta+\zeta\frac{ \mu_1(z)}{\mu_0(z)}\right)^{\lambda+1}\right] \\
    &= \mathbb{E}_{z \sim \mu_0}\left[ \sum_{\eta =0}^{\lambda+1} \binom{\lambda+1}{\eta} (1-\zeta)^{\lambda+1-\eta} (\zeta \frac{ \mu_1(z)}{\mu_0(z)})^\eta   \right] \\
    &=  \sum_{\eta =0}^{\lambda+1} \binom{\lambda+1}{\eta} (1-\zeta)^{\lambda+1-\eta} \zeta^\eta \mathbb{E}_{z \sim \mu_0}\left[
 \Bigl(\frac{ \mu_1(z)}{\mu_0(z)}\Bigr)^\eta   \right] \\
\end{align*}

\[
A= \mathbb{E}_{z \sim \mu_\Delta}\left[\left(1 - \zeta+\frac{\zeta \mu_1(z)}{\mu_0(z)}\right)^\lambda\right]=\mathbb{E}_{z \sim \mu_0}\left[\left(1-\zeta+\frac{\zeta \mu_1(z)}{\mu_0(z)}\right)^{\lambda+1}\right],\]
\[B=\mathbb{E}_{z \sim \mu_0}\left[\left(
\frac{\mu_0(z)}{(1-\zeta)\mu_0(z) + \zeta \mu_1(z)}
\right)^\lambda\right]=\mathbb{E}_{z \sim \mu_0}\left[\left(
\frac{1}{1-\zeta +  \frac{\zeta\mu_1(z)}{\mu_0(z)}}
\right)^\lambda\right]
\]

Let us begin with "$A$." While $A$ involves two cases—$\mu_0(z) \propto e^{-\frac{|z|_1}{b}}$ and $\mu_0(z) \propto e^{-\frac{|z - C|_1}{b}}$—both yield the same bound, likely due to an underlying symmetry. Thus, we focus on the former.

\textbf{Analysis for  A:}
\[
A=\mathbb{E}_{z \sim \mu_0(b)}\left[\left(1-\zeta+\zeta \cdot e^{\frac{-\|z-C\|_1+\|z\|_1}{b}}\right)^{\lambda+1}\right],\]

To tackle each we must cover three cases, $z \geq c$, $z<0$ and $z \in[0,C)$. We will first analyze easier $z \geq c$ and $z<0$ cases and lastly derive the harder case $z \in[0,C)$. 

\textbf{Case A-I ($z \geq C$):}
\[
A=\mathbb{E}_{\{[z \geq C] \sim  \mu_0(b)\}}\left[\left(1-\zeta+\zeta \cdot e^{C/b}\right)^{\lambda+1}\right],\]
\[A=\mathbb{E}_{\{[z \geq C] \sim  \mu_0(b)\}} \left\{
\sum_{\eta =0}^{\lambda+1} \binom{\lambda+1}{\eta} (1-\zeta)^{\lambda+1-\eta} \zeta^\eta \cdot e^{\eta C/b}
\right\}\]

\[A=
\sum_{\eta =0}^{\lambda+1}\binom{\lambda+1}{\eta} (1-\zeta)^{\lambda+1-\eta} \zeta^\eta \cdot \mathbb{E}_{\{[z \geq C] \sim  \mu_0(b)\}} \{e^{\eta C/b}
\}\]
Note that $e^{\eta C/b}$ has no variable $z$ and purely written in terms of b. Thus we can plug the Laplace CDF followd by expectation over b. Precisely,

\[A=
\sum_{\eta =0}^{\lambda+1} \binom{\lambda+1}{\eta} (1-\zeta)^{\lambda+1-\eta} \zeta^\eta \cdot  \mathbb{E}_{b \sim  f(b)\}}  \{\Proba_{z\sim Lap(0, b)}(z \geq C) \cdot e^{\eta C/b} 
\}\]
but we know that $\Proba_{z\sim Lap(0, b)}(z \geq C)=0.5 e^{-C/b}$. Thus
\[A=
\sum_{\eta =0}^{\lambda+1} \binom{\lambda+1}{\eta} (1-\zeta)^{\lambda+1-\eta} \zeta^\eta \cdot \mathbb{E}_{b \sim  f(b)\}} \{0.5 e^{-C/b} \cdot e^{\eta C/b} 
\}\]

Given that the expectation is half of MGF at $(\eta-1) \cdot C$, we have 

\[A= 0.5 *
\sum_{\eta =0}^{\lambda+1} \binom{\lambda+1}{\eta} (1-\zeta)^{\lambda+1-\eta} \zeta^\eta \cdot \mathcal{M}_{u}((\eta-1) \cdot C)\] 

\textbf{Case A-II ($z<0$):}
\[
A=\mathbb{E}_{\{[z < 0] \sim  \mu_0(b)\}}\left[\left(1-\zeta+\zeta \cdot e^{-\frac{C}{b}}\right)^{\lambda+1}\right],\]

\[A=\mathbb{E}_{\{[z < 0] \sim  \mu_0(b)\}} \left\{
\sum_{\eta =0}^{\lambda+1} \binom{\lambda+1}{\eta} (1-\zeta)^{\lambda+1-\eta} \zeta^\eta \cdot e^{-\eta C/b}
\right\}\]

\[A=
\sum_{\eta =0}^{\lambda+1}\binom{\lambda+1}{\eta} (1-\zeta)^{\lambda+1-\eta} \zeta^\eta \cdot \mathbb{E}_{\{[z<0] \sim  \mu_0(b)\}} \{e^{-\eta C/b}
\}\]
Note that $e^{-\eta C/b}$ has no variable $z$ and purely written in terms of b. Thus we can plug the Laplace CDF followd by expectation over b. Precisely,

\[A=
\sum_{\eta =0}^{\lambda+1} \binom{\lambda+1}{\eta} (1-\zeta)^{\lambda+1-\eta} \zeta^\eta \cdot  \mathbb{E}_{b \sim  f(b)\}}  \{\Proba_{z\sim Lap(0, b)}(z < 0) \cdot e^{-\eta C/b} 
\}\]
but  we know that $\Proba_{z\sim Lap(0, b)}(z<0)=0.5$. Thus
\[A=
\sum_{\eta =0}^{\lambda+1} \binom{\lambda+1}{\eta} (1-\zeta)^{\lambda+1-\eta} \zeta^\eta \cdot \mathbb{E}_{b \sim  f(b)\}} \{0.5 \cdot e^{-\eta C/b} 
\}\]

Given that the expectation is half of MGF at $-\eta \cdot C$, we have 

\[A= 0.5 *
\sum_{\eta =0}^{\lambda+1} \binom{\lambda+1}{\eta} (1-\zeta)^{\lambda+1-\eta} \zeta^\eta \cdot \mathcal{M}_{u}(-\eta \cdot C)\] 

\textbf{Case A-III ($0 \leq z<C$):}
\[
A=\mathbb{E}_{\{[0 \leq z < C] \sim  \mu_0(b)\}}\left[\left(1-\zeta+\zeta \cdot e^{\frac{2z-C}{b}}\right)^{\lambda+1}\right],\]

\[A=
\sum_{\eta =0}^{\lambda+1}\binom{\lambda+1}{\eta} (1-\zeta)^{\lambda+1-\eta} \zeta^\eta \cdot \mathbb{E}_{\{[0 \leq z < C] \sim  \mu_0(b)\}} \{e^{\frac{\eta(2z-C)}{b}}
\}\]
So we need to calculate $\mathbb{E}_{\{[0 \leq z < C] \sim  \mu_0(b)\}} \{e^{\frac{\eta(2z-C)}{b}}
\}$ which can be written as a sequence of two expectations (law of total expectation):
\[\mathbb{E}_{b \sim  f(b)} \mathbb{E}_{z\sim Lap(0, b)} \left\{e^{\frac{\eta(2z-C)}{b}}
\right\}= \mathbb{E}_{b \sim  f(b)} \left\{\int_0^C \frac{e^{\frac{\eta(2z-C)-|z|}{b}}}{2b}
dz \right\} \]
We focus on the one-dimensional Laplace distribution in our analysis to simplify the presentation and avoid excessive jargon. The extension to the multi-dimensional case follows straightforwardly. Calculating the integral is easy and it is given as:
\[\int_0^C \frac{e^{\frac{z(2\eta-1)-(C\eta)}{b}}}{2b}
dz = \frac{e^{\frac{C(2\eta-1)-(C\eta)}{b}}-e^{\frac{-C\eta}{b}}}{2*(2\eta-1)}\]
Moving forward with the expectation:
\[\mathbb{E}_{b \sim  f(b)} \left\{\frac{e^{\frac{C(\eta-1)}{b}}-e^{\frac{-C\eta}{b}}}{2*(2\eta-1)}\right\}=\frac{\mathcal{M}_u(C(\eta-1))-\mathcal{M}_u(-C\eta)}{2*(2\eta-1)} \].

Combining the results in Cases A-I, A-II and A-III, we have that
$A = \frac{\eta}{2\eta - 1} \mathcal{M}_u\big((\eta - 1)C\big) + \frac{\eta - 1}{2\eta - 1} \mathcal{M}_u\big(-\eta C\big)$ concluding the proof.

\textbf{Analysis for B:} Following the argument in Proof~\ref{subsampledLap2} (Theorem 3.2), and using binomial expansion with term-wise comparison, we find that \( B \geq A \), consistent with the result of Mironov et al.

\end{proof}

\subsection{Justifying Differentiation (Leibniz's Rule)}
\label{sec:justify}

In the proof of Theorem~\ref{error_PLRV} (presented in Section~\ref{apdx:proof:error_PLRV}), we will use Leibniz's rule for differentiation under the integral sign.  
For completeness, we now discuss the conditions under which this holds.

We use Leibniz's rule for differentiation under the integral sign:
\[
\frac{d}{du} \int_{\mathbb{R}^k} f(u, O) \, dO = \int_{\mathbb{R}^k} \frac{\partial f(u, O)}{\partial u} \, dO,
\]
which is valid under the following conditions:
1. Continuity of the Integrand: The function \(f(u, O) = e^{-u \cdot \|O - \bar{\mathbf{g}}_t(d)\|_1}\) is continuous with respect to \(u > 0\) and \(O \in \mathbb{R}^k\).
2. Existence and Continuity of the Derivative: The partial derivative \(\frac{\partial f(u, O)}{\partial u} = -\|O - \bar{\mathbf{g}}_t(d)\|_1 e^{-u \cdot \|O - \bar{\mathbf{g}}_t(d)\|_1}\) exists and is continuous.
3. Convergence of the Integral: The integral
   \[
   \int_{\mathbb{R}^k} e^{-u \cdot \|O - \bar{\mathbf{g}}_t(d)\|_1} \, dO
   \]
   converges for all \(u > 0\), as the exponential decays rapidly for large \(\|O - \bar{\mathbf{g}}_t(d)\|_1\).
4. No Divergence from the Derivative: The derivative \(-\|O - \bar{\mathbf{g}}_t(d)\|_1 e^{-u \cdot \|O - \bar{\mathbf{g}}_t(d)\|_1}\) does not cause divergence because the exponential decay ensures the integral remains finite.

Since these conditions are met, applying Leibniz's rule to move \(\frac{d}{du}\) outside the inner integral is valid.

\subsection{Proof of Theorem~\ref{error_PLRV}}
\label{apdx:proof:error_PLRV}

\begin{proof}
The \(\ell_1\) error is defined as:
\[
\ell_1(z) = \mathbb{E}[\| z - \bar{\mathbf{g}}_t(d) \|_1].
\]
Using the PDF \( g(u) \) of the inverse scale parameter \( u \), it can be rewritten as:
\[
\ell_1(z) = \int_{0}^{\infty} g(u) \int_{\mathbb{R}^n} \| z - \bar{\mathbf{g}}_t(d) \|_1 e^{-u \| z - \bar{\mathbf{g}}_t(d) \|_1} \, dz \, du.
\]

a) applying Leibniz’s rule, which permits differentiation under the integral sign (see Appendix~\ref{sec:justify} for condition verification), and using the identity $\| z - \bar{\mathbf{g}}_t(d) \|_1 e^{-u \| z - \bar{\mathbf{g}}_t(d) \|_1} = -\frac{d}{du} e^{-u \| z - \bar{\mathbf{g}}_t(d) \|_1},$

b) moving \( \frac{d}{du} \) outside the inner integral: \[
\ell_1(z) = \int g(u) \left(\frac{u}{2}\right)^n \left( -\frac{d}{du} \int_{\mathbb{R}^n} e^{-u \| O - \bar{\mathbf{g}}_t(d) \|_1} \, dO \right) \, du, \] and c) integrating by parts yields: $
\ell_1(z) = \int_{0}^{\infty} \frac{g(u)}{u} \, du$. 

\end{proof}

\if 1
\color{brown}
\subsection{Laplace Mechanism in High Privacy Regimes}
\begin{theorem}[Privacy Loss Bound for Laplace Mechanism]
Let $M(x) = f(x) + Z$ and $M(x') = f(x') + Z$ be the outputs of the Laplace mechanism applied to neighboring databases $x$ and $x'$, where $Z \sim \text{Laplace}(0, b)$.

Define the privacy loss random variable
\[
L(O) = \log\left( \frac{p_{M(x)}(O)}{p_{M(x')}(O)} \right),
\quad \text{where } O = M(x).
\]
Then, for $\Delta = \|f(x) - f(x')\|_1$, the probability of positive privacy loss satisfies
\[
\delta(0) = \Pr(L(O) > 0) = 
\begin{cases}
\displaystyle \frac{e^{\frac{\Delta}{b}} (2 - \frac{\Delta}{b})}{4}, & \text{if } \Delta < 0, \\
\displaystyle 1 - \frac{e^{-\frac{\Delta}{b}} (2 + \frac{\Delta}{b})}{4}, & \text{if } \Delta > 0.
\end{cases}
\]
\end{theorem}
\begin{proof}
By definition of the Laplace mechanism, we have
\[
M(x) = f(x) + Z, \quad M(x') = f(x') + Z,
\]
where $Z \sim \text{Laplace}(0, b)$.

The privacy loss random variable at output $O$ is
\[
L(O) = \frac{1}{b} \left( |O - f(x')| - |O - f(x)| \right).
\]
Substituting $O = M(x) = f(x) + Z$ yields
\[
L(O) = \frac{1}{b} \left( |Z + (f(x) - f(x'))| - |Z| \right).
\]

Thus, the event $\{L(O) > 0\}$ is equivalent to
\[
\left\{ |Z + (f(x) - f(x'))| > |Z| \right\}.
\]
Let $\Delta = f(x) - f(x')$. We can then rewrite the event as
\[
\Pr(|Z + \Delta| > |Z|).
\]

Because $Z \sim \text{Laplace}(0, b)$, we can standardize by defining $Z = b Z_1$ where $Z_1 \sim \text{Laplace}(0,1)$.
Thus,
\[
\Pr\left( |\Delta + b Z_1| > |b Z_1| \right) = \Pr\left( \frac{\Delta}{b} + Z_1 > Z_2 \right),
\]
where $Z_1, Z_2 \sim \text{Laplace}(0,1)$ are independent standard Laplace random variables.
You have:
\[
P(X > Y) = P(\mu + b Z_1 > Z_2)
\quad \text{where} \quad Z_1, Z_2 \sim \text{Laplace}(0,1)
\]
with:
\[
b = \frac{b_X}{b_Y}, \quad \mu = \frac{\mu_Y - \mu_X}{b_Y}= \frac{f(x) - f'(x)}{b}.
\]

\medskip

\noindent

The probability
\[
P(\mu + b Z_1 > Z_2)
\]
is given by:
\[
P(\mu + b Z_1 > Z_2) =
\begin{cases}
\displaystyle \frac{b^2 e^{\mu/b} - e^{\mu}}{2(b^2 - 1)}, & \text{if } \mu < 0, \\[12pt]
\displaystyle 1 - \frac{b^2 e^{-\mu/b} - e^{-\mu}}{2(b^2 - 1)}, & \text{if } \mu > 0.
\end{cases}
\]

\end{proof}
\color{black}

\fi 

\section{Further Discussion}
\subsection{Interpreting \textsf{PLRV-O} Parameters}
\label{exp_appendix:parameter}

\noindent\textbf{Theta - Learnability.} 
Table~\ref{tab:theta_dependency} highlights the role of the $\theta$ parameter in model learnability under privacy constraints. By varying the $\theta$ while keeping all other hyperparameters fixed, we have observed a clear understanding of $\theta$ dependency on model performance. Smaller values of $\theta$ led to significantly lower performance, while larger values of $\theta$ substantially improved the model performance. Therefore, $\theta$ can be observed as a sensitivity control parameter, with larger values offering good model learnability in our setting.

\vspace{0.05in}
\noindent\textbf{K - Utility Parameter.} 
Table~\ref{tab:k_dependency} shows that the accuracy of \textsf{PLRV-O} boosts with higher values of $k$. This suggests that $k$ serves as a utility parameter, improving the model performance under DP.

\begin{table}[H]
    \centering
    \vspace{0.15in}
      \caption{
        CNN-MNIST: Memorization threshold dependency on $\theta$. Fixed parameters: $\delta$ = $10^{-5}$, Clip = 3, $q = 0.01$, Steps = 300. Larger $\theta$ values are less destructive in this setting. Small theta values can lead to either unstable models or highly protected yet accurate ones. PLRV-O Acc averaged over 5 trials. 
    }
    \begin{tabular}{|c|c|c|c|}
    \hline
    \textbf{$\epsilon$} & \textbf{k} & \textbf{$\theta$}    & \textbf{\textsf{PLRV-O} Acc} \\
     
    \hline
    0.0457 & 1418 & \cellcolor{yellow!50}1.01E-05   & \cellcolor{red!10}18.456\textpm1.87  \\
    
    0.7971 & 1418 & \cellcolor{yellow!50} 1.01E-04    & \cellcolor{red!10}64.174\textpm4.43  \\
    
    $238.10$ & 1418 & \cellcolor{yellow!50}1.01E-03    & \cellcolor{red!10}93.658\textpm0.34 \\
    \hline
    \end{tabular}
    \label{tab:theta_dependency}
\end{table}

\begin{table}[H]
    \centering
       \caption{CNN-MNIST: Model quality enhances with $k$. PLRV-O Acc taken over 5 trials and averaged. Fixed parameters: $\delta$ = $10^{-5}$, $\theta =1.01E-03 $, Clip = 1, $q$ = 0.01, Steps = 300. }
    \begin{tabular}{|c|c|c|}
    \hline
    \textbf{$\epsilon$} & \textbf{k}   & \textbf{\textsf{PLRV-O} Acc}  \\
    \hline
    1.0691 & \cellcolor{yellow!50}500   & \cellcolor{red!10}87.614\textpm1.34 \\
    
    4.1055 & \cellcolor{yellow!50}1000   & \cellcolor{red!10}92.146\textpm0.91 \\
    
    8.3310 & \cellcolor{yellow!50}1390   & \cellcolor{red!10}93.792\textpm0.57 \\
    
    8.7180 & \cellcolor{yellow!50}1418.922 & \cellcolor{red!10}93.634\textpm0.79 \\
    
    8.8879 & \cellcolor{yellow!50}1430    & \cellcolor{red!10}93.702\textpm0.63 \\
    
    9.9292 & \cellcolor{yellow!50}1500    & \cellcolor{red!10}93.438\textpm0.42 \\
    
    20.4878 & \cellcolor{yellow!50}2000    & \cellcolor{red!10}94.13\textpm0.49 \\
    
    66.5293 & \cellcolor{yellow!50}3000    & \cellcolor{red!10}95.114\textpm0.22 \\
    
    363.376 & \cellcolor{yellow!50}5000  & \cellcolor{red!10}96.00\textpm0.45 \\
    
    2.23E+03 & \cellcolor{yellow!50}10000 & \cellcolor{red!10}96.212\textpm0.40 \\
    \hline
    \end{tabular}
    \vspace{0.5cm}
    \label{tab:k_dependency}
\end{table}

\subsection{Tighter Search Space for \textsf{PLRV-O}}
\label{parameter}
\textsf{PLRV-O} finetuning for larger models, e.g., >100M parameters, is a harder problem to solve. We propose an approach to constructing the \textsf{PLRV-O} search space by filtering candidate configurations against a Gaussian DP-SGD baseline. A configuration is retained only if it achieves both (i) lower moment values at every order and (ii) lower distortion than its Gaussian counterpart.

\begin{theorem}[The \textsf{PLRV-O} Search Space]
Given a training dataset \( \mathcal{D} \), task configuration \( \mathbf{x} = (E, B, C) \), and privacy parameters \((\epsilon, \delta)\), the \textsf{PLRV-O} search space \( SS(\mathbf{x}, \epsilon, \delta) \) is defined as:
\begin{equation*}
SS(\mathbf{x}, \epsilon, \delta) = \left\{ \boldsymbol{\Theta}(\mathbf{x}) \ \bigg| \ \mathbb{I}_{DP}(\boldsymbol{\Theta}(\mathbf{x}) , \eta, \epsilon, \delta) \land \mathbb{I}_{Dist}(\boldsymbol{\Theta}(\mathbf{x}), \epsilon, \delta) = 1 \right\}
\end{equation*}
where the indicators enforce per-moment privacy bounds and distortion constraints:
\begin{align*}
&\mathbb{I}_{DP}(\boldsymbol{\Theta}, \eta, \epsilon, \delta) = 
\begin{cases}
1, & \forall \eta \in \mathbb{N} : \mathcal{G}(u(\boldsymbol{\Theta}(\mathbf{x}) ), \eta) \leq \exp\left(\frac{\eta^2 - \eta}{2\sigma^2(\mathbf{x}, \epsilon, \delta)}\right) \\
0, & \text{otherwise}
\end{cases} \\
&\mathbb{I}_{Dist}(\boldsymbol{\Theta}, \mathbf{x}, \epsilon, \delta) = 
\begin{cases}
1, & 
    \int_{0}^{\infty} \mathcal{M}_{u(\boldsymbol{\Theta})}(-z) \, dz 
    \leq 
    \sqrt{\frac{2}{\pi}} C\sigma(\mathbf{x}, \epsilon, \delta) \\
0, & \text{otherwise}
\end{cases}.
\end{align*}
\end{theorem}

In the case of \(\Gamma(k, \theta)\)-PLRV noise, the feasible region can be characterized analytically. The following corollary provides closed-form bounds for \(k\), conditioned on \(C\), \(\theta\), and the maximum moment order \(\lambda_{max}\).

\begin{corollary}[The \(\Gamma\)-PLRV Search Space]
\label{cor:ss}
Given task configuration \( \mathbf{x} = (E, B, C) \) and privacy parameters \((\epsilon, \delta)\), the valid \(\Gamma(k, \theta)\)-PLRV configurations satisfy:
\[
k_1(C, \theta, \eta) \leq k \leq k_2(C, \theta, \eta), \quad \forall \eta \leq \lambda_{max} + 1,
\]
subject to \( \theta < \frac{1}{C(\lambda_{max} + 1)} \) and \( k > 1 \), where:
\[
k_{1,2}(C, \theta, \eta) = \frac{4 \log(1 - C \theta (\eta - 1)) \mp \sqrt{\Delta}}{2 C^2 \theta^2 (\eta - \eta^2)},
\]
and
$
\Delta = 16 \log^2(C \theta - C \eta \theta + 1) + 11.09375 \cdot C^2 \theta^2 (\eta - \eta^2).
$
\end{corollary}
\begin{proof}
We have $\theta<\frac{1}{C (\lambda+1)}$ to ensure existence of MGF for the $\Gamma$-PDF.
For the \(\Gamma\)-PLRV mechanism, the indicator functions specialize to
\begin{align*}
&\mathbb{I}_{DP}(k, \theta, \eta, \epsilon, \delta) = 
\begin{cases}
1, & \forall \eta \in \mathbb{N} : \mathcal{G}(C, k, \theta, \eta) \leq \exp\left(\frac{\eta^2 - \eta}{2\sigma^2(\mathbf{x}, \epsilon, \delta)}\right) \\
0, & \text{otherwise}
\end{cases}
\\
&\mathbb{I}_{Dist}(k, \theta, \mathbf{x}, \epsilon, \delta) = 
\begin{cases}
1, &  
\frac{1}{(k - 1)\theta} 
\leq 
\sqrt{\frac{2}{\pi}} C \sigma(\mathbf{x}, \epsilon, \delta)
\\
0, & \text{otherwise}
\end{cases}
\end{align*}
where 
\begin{align*}
    \mathcal{G}(C, k, \theta, \eta) &= 
\frac{(1 - (\eta - 1) C \theta)^{-k} + (1 + \eta C \theta)^{-k}}{2} 
\\
&\qquad + \frac{(1 - (\eta - 1) C \theta)^{-k} - (1 + \eta C \theta)^{-k}}{2(2\eta - 1)}.
\end{align*}

Using the bound for $\mathbb{I}_{Dist}$, we have
\begin{align}
    \frac{1}{(k - 1)\theta} 
&\leq 
\sqrt{\frac{2}{\pi}} C \sigma(\mathbf{x}, \epsilon, \delta) \nonumber\\
\implies 
\frac{1}{\sigma} &\leq \sqrt{\frac{2}{\pi}} C (k - 1)\theta. \label{eq:SS:sigma-k-theta-bnd}
\end{align}
Applying \eqref{eq:SS:sigma-k-theta-bnd} in $\mathbb{I}_{DP}(k, \theta, \eta, \epsilon, \delta)$  we have that for every $\eta \in \mathbb{N}$, 
\begin{align*}
    \mathcal{G}(C, k, \theta, \eta) \leq \exp\left(\frac{\eta^2 - \eta}{\pi} C^2 (k - 1)^2\theta^2\right).
\end{align*}
\end{proof}
Strictly speaking, due to the inequality direction in~\eqref{eq:SS:sigma-k-theta-bnd}, this yields a larger region than that characterized for a particular $\sigma$.  
However, this nonetheless yields a well-defined search space purely in terms of $\Gamma$-PLRV parameters $k$ and $\theta$. 

\vspace{0.05in}

\noindent \textbf{Intersecting with $c_0$--$c_4$.}
To ensure practical feasibility, we further intersect this region with the heuristic constraints described earlier:
\begin{itemize}
    \item[\textbf{$c_0$}] Clip bounded: $C_{\min}\le C \le C_{\max}$.
    \item[\textbf{$c_1$}] Gamma tail: $\theta \ge \theta_{\min}(k)$ $\gets$ $\mathrm{GammaCDF}(0.1;k,\theta) <<1$.
    \item[\textbf{$c_2$}] (Deferred here) privacy accounting via $\tau(\epsilon)$.
    \item[\textbf{$c_3$}] Stability: $k>1$.
    \item[\textbf{$c_4$}] Distortion cap: $(k-1)\theta > 0.1$.
\end{itemize}

Together, the feasible region is characterized by
\[
\max\!\left\{ \theta_{\min}(k), \tfrac{0.1}{k-1} \right\}
\;\le\;
\theta
\;<\;
\frac{1}{\lambda_{\max}C}, 
\quad
k>1,\;\; C_{\min}\le C \le C_{\max},
\]
with the objective of maximizing $(k{-}1)\theta$ (equivalently minimizing distortion).  
In practice, accounting further constrains the product $(k{-}1)\theta$ by requiring
\[
(k{-}1)\theta \leq \tau(\epsilon).
\]
As illustrated in Figure~\ref{fig:task1_large_acc_qnli}, $\tau(\epsilon)$ grows approximately in proportion to $k\theta$, highlighting the condition under which feasible configurations concentrate. 

This result reduces the search region, allowing solvers to focus on feasible, high-quality configurations.
\begin{figure}[!h]
	\centering
    \vspace{-0.1in}\includegraphics[width=0.85\linewidth, trim=55 190 70 215, clip]{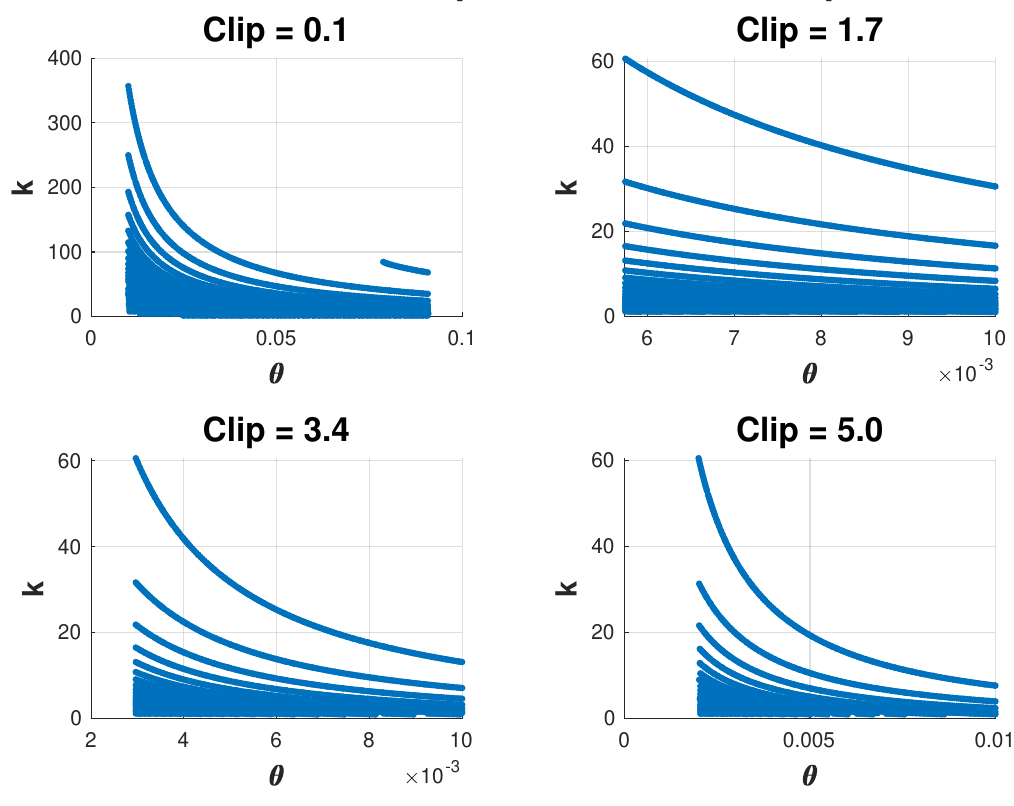}
\caption{$\Gamma$-\textsf{PLRV-O} SS}
\label{fig_gammPLRVss}
\end{figure}
\begin{figure}[ht]
    \centering
    \includegraphics[width=0.9\linewidth,clip]{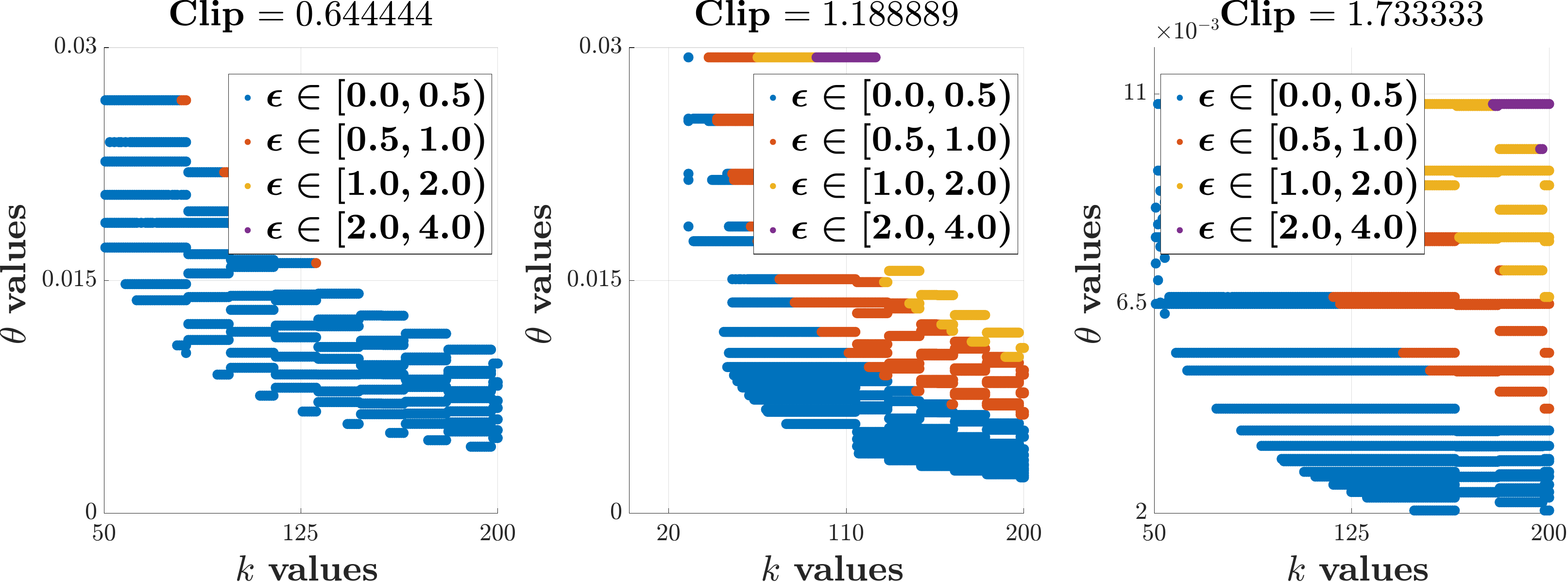}
    \caption{Search space for configurations in an MNLI task.}\vspace{0.1in}
\label{fig:task1_large_acc_qnli}
\end{figure}

Figure~\ref{fig_gammPLRVss}  visualizes this search space for $T,\zeta, n$ required for NLP task MNLI, showing how optimal configurations evolve with varying task size and budget.

\vspace{0.05in}

\noindent\textbf{Nonlinear Solver over the Feasible Search Space.}  
Given the pruned search space defined by constraints $c_0$–$c_4$, we optimize the SNR objective $J(k,\theta,C)=C(k{-}1)\theta$ using \texttt{fmincon}, MATLAB’s solver for constrained nonlinear minimization~\cite{mathworks_fmincon}. The feasible set is specified via bound constraints on $(k,\theta,C)$, together with nonlinear constraints encoding the Gamma CDF cutoff, distortion threshold, and MGF validity. \texttt{fmincon} flexibly handles such problems and provides multiple backends, including the interior-point method~\cite{byrd1999interior}, sequential quadratic programming (SQP)~\cite{boggs1995sequential}, and trust-region reflective methods~\cite{byrd2000trust}. This allows us to directly search within $SS(\mathbf{x}, \epsilon, \delta)$ and obtain $(k^*,\theta^*,C^*)$.

\vspace{0.05in}

\noindent\textbf{Gradient Boosting for Parameter Prediction.}  
To generalize \textsf{PLRV-O} configuration beyond explicitly optimized settings, one idea to explore is a regression-based prediction layer. For instance, training a Gradient Boosting Regression model, an ensemble method that incrementally improves predictions by fitting decision tree learners to residual errors. At iteration $m$, the predictor is updated as\(
F_m(x) = F_{m-1}(x) + \eta h_m(x),
\)
where $h_m(x)$ is the weak learner and $\eta$ is the learning rate. The model takes $(k,\epsilon,C,T,\zeta)$ as inputs and predicts the corresponding noise parameter $\theta$. This predictive layer captures nonlinear dependencies between training/privacy parameters and the optimal $\theta$, and provides feature-importance scores that highlight which inputs most strongly influence the choice of noise parameters.

Together, the nonlinear solver and boosting predictor form a hybrid system: \texttt{fmincon} yields high-fidelity solutions within a constrained region, while gradient boosting extends coverage to unseen settings by learning patterns in previously optimized configurations.

\vspace{-0.07in}
\begin{algorithm}[ht]
\caption{\textsf{PLRV-O} DP-SGD $\varphi_2$}
\label{alg:dp-adam}
\begin{algorithmic}[1]
\Require
  Data set $d$ with size $|d|$, Initial model parameters $\Phi_0$, Loss $\mathcal{L}(\Phi)$, Learning rate $lr$, Clipping threshold $C$, Number of epochs $E$, Batch size $B$, \textsf{PLRV-O} Parameterization $\varphi_{1}$, Privacy budget $(\epsilon^*,\delta^*)$, Population $N$, \textsf{PLRV-O} noise generator $\mathcal{P}$
\Ensure Model parameters $\Phi_{T}$
\Statex \vspace{0.5ex}\textbf{// Initialization}
\State Load model parameters $\Phi_{0}$
\State Sampling rate $\zeta \gets B/|d|$
\State Number of iterations $T \gets \lceil E/q \rceil$
\State $(k^*,\theta^*,\widehat{\epsilon}\,) \gets \varphi_{1}(\epsilon^*, \tau, T, \zeta, C, N, \delta^*)$
\Statex \vspace{0.5ex}\textbf{// Main loop}
\For{$t = 1$ \textbf{to} $T$}
  \State \parbox[t]{0.85\linewidth}{%
      Draw a batch $\mathcal{B}_{t}$ uniformly at random,
      where each element is included independently with probability $\zeta$.
    }
  \For{$\mathbf{x}_{i} \in \mathcal{B}_{t}$}
    \State $\mathbf{g}_{t}(\mathbf{x}_{i}) \gets \nabla_{\Phi}\,\mathcal{L}(\mathbf{x}_{i})$
    \State $\bar{\mathbf{g}}_{t}(\mathbf{x}_{i}) \gets 
      \mathbf{g}_{t}(\mathbf{x}_{i}) \cdot 
      \min\!\left(1, \dfrac{C}{\lVert \mathbf{g}_{t}(\mathbf{x}_{i}) \rVert_{2}}\right)$
  \EndFor
  \State $z \sim \mathcal{P}(k^*,\theta^*)$ \Comment{draw \textsf{PLRV-O} noise}
  \State $\bar{\mathbf{g}}_{t} \gets \dfrac{1}{B}\!\sum_{\mathbf{x}_{i}\in\mathcal{B}_{t}}\! \bar{\mathbf{g}}_{t}(\mathbf{x}_{i}) \;+\; z$
  \State $\bar{\mathbf{g}}_{t}+1 \gets \text{Optimizer}(\bar{\mathbf{g}}_{t})$ \Comment{Where Optimizer is a DP-ready optimizer, such as AdamUpdate from~\cite{li2021large}}
\EndFor
\State \Return $\Phi_{T}$
\end{algorithmic}

\end{algorithm}
\vspace{5pt}

\noindent\textbf{\textsf{PLRV-O} and Neural Collapse.} 
Recent work~\cite{pmlr-v235-wang24cu} highlights that DP-SGD training can induce \textit{neural collapse}—a phenomenon where class representations become degenerate—particularly under strong privacy budgets. This occurs when DP noise overwhelms inter-class variation, flattening feature geometry. Since \textsf{PLRV-O} supports adaptive noise shaping via parameters \(k\), \(\theta\), and \(C\), it enables more controlled distortion and finer-grained modulation of the injected noise. By minimizing expected \(\ell_1\) distortion, \textsf{PLRV-O} better preserves representational diversity during training, reducing the risk of collapse in the latent space and sustaining learning signal even in high-dimensional settings.

\subsection{Auditing DP-SGD}
\label{appendix:audit}

We consider two adversary models: a worst-case \emph{theoretical DP adversary} and a practical \emph{empirical auditing adversary}.

The Theoretical DP  adversary has unbounded computation, full access to mechanism outputs, and arbitrary auxiliary knowledge. Privacy is formalized by $(\epsilon,\delta)$-DP, which bounds the effect of any single record on the output distribution.

Real deployments face constrained attackers. We used DP auditing tools that execute practical tests (e.g., canary-based inference) under limited queries or partial distributional knowledge, estimating empirical privacy via hypothesis testing (e.g., Type-II error–based criteria) while tracking the utility–noise trade-off.

\vspace{0.5em}
\noindent\textbf{ClipBKD.} For ClipBKD, we run $T{=}500$ trials on CV tasks and $T{=}10$ on NLP tasks with confidence level $0.01$ (99\% confidence), and report the best result using $k{=}1$ poisoning point. For context, we also report (i) the best theoretical upper bound on $\epsilon$ and (ii) $\epsilon_{\text{OPT}}(500,0.01)$, the best achievable lower bound $\epsilon_{\text{LB}}$ under $T$ trials at confidence $0.01$.

\section{Additional Information}
\label{exp_appendix:CV}

\subsection{Experimental Setting}
Our experiments were conducted on the following servers:

\begin{itemize}
  \item AMD Ryzen Threadripper PRO 5975WX 32-Core CPU, 500~GB RAM, and 3$\times$NVIDIA Quadro RTX A6000 48GB GPUs.
  \item AMD EPYC 9354 32-Core Processor (128 cores), 1.5~TB RAM, and 2$\times$NVIDIA H100 PCIe 80GB GPUs.
\end{itemize}

\subsection{Additional Results}
Tables~\ref{tab:cnn_mnist_results}-\ref{tab:cnn_fmnist_results} present additional evaluations of the performance of the \textsc{PLRV-O} framework on CNN-MNIST.

\begin{table}[H]
\centering
\vspace{0.15in}
\caption{
\textbf{PLRV-O} performance on CNN-MNIST in high privacy regimes $\epsilon \leq 1.6$. Fixed parameters: $\delta = 10^{-5}$, $q = 0.01$, steps = 300. Gaussian baseline included and non-private (NP) accuracy range 96-98\%. \textbf{PLRV-O} accuracy was taken over 5 trials and averaged.
}
\begin{tabular}{|c|c|c|c|c|c|}
\hline
\textbf{$\epsilon$} & \textbf{Clip} & \textbf{$\theta$} & \textbf{$k$}  & \textbf{Gauss Acc} & \textbf{\textsf{PLRV-O} Acc} \\
\hline
0.065 & 0.1 & 1.00E-05 & 60000 & 66.56\textpm6.37 & \textbf{88.78\textpm1.06}  \\
0.171 & 0.1 & 5.00E-05 & 30000 & 82.20\textpm2.40 & \textbf{93.38\textpm0.90} \\
0.284 & 0.3 & 8.00E-05 & 10000  & 88.53\textpm2.11 & \textbf{89.96\textpm0.95}\\
0.596 & 0.3 & 8.00E-04 & 20000  & 91.44\textpm1.65 & \textbf{93.85\textpm0.74}\\
0.757 & 0.3 & 5.00E-04 & 40000  & 92.26\textpm0.26 & \textbf{94.30\textpm0.51}\\
0.921 & 0.3 & 6.00E-04 & 40000  & 92.65\textpm0.46 & \textbf{94.81\textpm0.13}\\
0.962 & 0.5 & 3.00E-04 & 50000  & 92.87\textpm0.35 & \textbf{93.69\textpm0.56} \\
1.172 & 0.5 & 3.00E-05 & 60000  & 92.78\textpm0.18 & \textbf{94.37\textpm0.65}\\
1.606 & 0.5 & 6.00E-05 & 40000 & 93.73\textpm0.43 & \textbf{94.91\textpm0.43}\\
\hline
\end{tabular}\vspace{0.05in}
\label{tab:cnn_mnist_results}
\end{table}

\begin{table}[H]
\caption{
\textbf{PLRV-O} performance on CNN-Fashion-MNIST in high privacy regimes $\epsilon \leq 1.6$. Fixed parameters: $\delta = 10^{-5}$, $q = 0.01$, steps = 300. Gaussian baseline included and non-private (NP) accuracy range 82-83\%. \textbf{PLRV-O} accuracy was taken over 5 trials and averaged.
}
\begin{tabular}{|c|c|c|c|c|c|}
\hline
\textbf{$\epsilon$} & \textbf{Clip} & \textbf{$\theta$} & \textbf{$k$} & \textbf{Gauss Acc} & \textbf{\textsf{PLRV-O} Acc}\\
\hline
0.065 & 0.1 & 1.00E-05 & 60000 & 59.28\textpm2.79 & \textbf{70.72\textpm0.61}  \\
0.171 & 0.1 & 5.00E-05 & 30000 & 68.19\textpm1.91 & \textbf{73.42\textpm1.49} \\
0.284 & 0.3 & 8.00E-05 & 10000  & 69.77\textpm1.67 & \textbf{72.34\textpm0.82}\\
0.596 & 0.3 & 8.00E-04 & 20000  & 72.34\textpm0.67 & \textbf{73.67\textpm1.29}\\
0.757 & 0.3 & 5.00E-04 & 40000  & 72.64\textpm0.79 & \textbf{74.57\textpm0.56}\\
0.921 & 0.3 & 6.00E-04 & 40000  & 73.57\textpm0.64 & \textbf{75.49\textpm0.82}\\
0.962 & 0.5 & 3.00E-04 & 50000  & 73.24\textpm0.61 & \textbf{74.48\textpm0.57} \\
1.172 & 0.5 & 3.00E-05 & 60000  & 73.04\textpm0.84 & \textbf{74.72\textpm0.90}\\
1.606 & 0.5 & 6.00E-05 & 40000 & 74.25\textpm0.53 & \textbf{76.25\textpm0.73}\\
\hline
\end{tabular}\vspace{-0.15in}
\label{tab:cnn_fmnist_results}
\end{table} 

\end{document}